\documentclass[a4paper,11pt]{article}
\pdfoutput=1

\usepackage{jheppub}
\usepackage{amsmath,amssymb,amsthm,amscd,graphicx}
\usepackage{hyperref}
\usepackage{bm}

\input epsf.sty

\newtheorem{theorem}{Theorem}[section]
\newtheorem{proposition}[theorem]{Proposition}
\newtheorem{corollary}[theorem]{Corollary}

\theoremstyle{definition}

\newtheorem{example}[theorem]{Example}
\newtheorem{remark}[theorem]{Remark}

\newcommand{\CB}{{\cal B}}
\newcommand{\CC}{{\cal C}}

\newcommand{\CG}{{\cal G}}

\newcommand{\CJ}{{\cal J}}

\newcommand{\CL}{{\cal L}}

\newcommand{\CO}{{\cal O}}
\newcommand{\CP}{{\cal P}}

\def\IZ{{\mathbb Z}}
\def\IR{{\mathbb R}}
\def\IC{{\mathbb C}}
\def\IQ{{\mathbb Q}}
\def\IN{{\mathbb N}}
\def\IP{{\mathbb P}}

\def\IF{{\mathbb F}}

\newcommand{\re}{{\rm e}}

\newcommand{\ri}{{\rm i}}
\newcommand{\rd}{{\rm d}}

\renewcommand{\d}{\partial}

\newcommand{\mA}{\mathsf{A}}
\newcommand{\mB}{\mathsf{B}}

\newcommand{\mO}{\mathsf{O}}
\newcommand{\mrho}{\mathsf{\rho}}

\newcommand{\mx}{\mathsf{x}}

\newcommand{\my}{\mathsf{y}}

\newcommand{\mb}{{\mathsf{b}}}

\newcommand{\be}{\begin{equation}}
\newcommand{\ee}{\end{equation}}
\newcommand{\ba}{\begin{aligned}}
\newcommand{\ea}{\end{aligned}}
\newcommand{\ben}{\begin{eqnarray}\displaystyle}
\newcommand{\een}{\end{eqnarray}}

\newcommand{\sectiono}[1]{\section{#1}\setcounter{equation}{0}}

\title{\LARGE{\boldmath Resurgence, Stokes constants, and arithmetic functions in topological string theory}}

\author{Claudia Rella}

\affiliation{D\'epartement de Physique Th\'eorique\\ Universit\'e de Gen\`eve, CH-1211 Gen\`eve, Switzerland}

\emailAdd{claudia.rella@unige.ch} 

\abstract{
The quantization of the mirror curve to a toric Calabi--Yau threefold gives rise to quantum-mechanical operators, whose fermionic spectral traces produce factorially divergent power series in the Planck constant.
These asymptotic expansions can be promoted to resurgent trans-series. They show infinite towers of periodic singularities in their Borel plane and infinitely many rational Stokes constants, which are encoded in generating functions expressed in closed form in terms of $q$-series.
We provide an exact solution to the resurgent structure of the first fermionic spectral trace of the local $\IP^2$ geometry in the semiclassical limit of the spectral theory, corresponding to the strongly-coupled regime of topological string theory on the same background in the conjectural TS/ST correspondence. Our approach straightforwardly applies to the dual weakly-coupled limit of the topological string. We present and prove closed formulae for the Stokes constants as explicit arithmetic functions and for the perturbative coefficients as special values of known $L$-functions, while the duality between the two scaling regimes of strong and weak string coupling constant appears in number-theoretic form.
A preliminary numerical investigation of the local $\IF_0$ geometry unveils a more complicated resurgent structure with logarithmic sub-leading asymptotics.
Finally, we obtain a new analytic prediction on the asymptotic behavior of the fermionic spectral traces in an appropriate WKB double-scaling regime, which is captured by the refined topological string in the Nekrasov--Shatashvili limit.
}

\makeatletter
\gdef\@fpheader{\null}
\makeatother

\begin{document}

\maketitle
\flushbottom

\sectiono{Introduction} \label{sec: introduction} 
Resurgent asymptotic series arise naturally as perturbative expansions in quantum theories. The machinery of resurgence uniquely associates them with a non-trivial collection of complex numbers, known as Stokes constants, which capture information about the large-order behavior of the perturbative expansion and about the non-perturbative sectors which are invisible in conventional perturbation theory. In some remarkable cases, the Stokes constants are integers, and they possess an interpretation as enumerative invariants based on the counting of BPS states. Recent progress in this direction~\cite{GrGuM, GGuM, GGuM2, GuM, GuM2, GuM3} advocates for the investigation of algebro-geometric theories within the analytic framework of the theory of resurgence. 

A first example is given by 4d $\mathcal{N} = 2$ supersymmetric gauge theory in the Nekrasov--Shatashvili limit~\cite{NS} of the Omega-background~\cite{N}, whose BPS spectrum~\cite{GMN} is encoded in the Stokes constants of the asymptotic series obtained by quantizing the Seiberg--Witten curve~\cite{GrGuM}.
A second example is given by complex Chern--Simons theory on the complement of a hyperbolic knot. The Stokes constants of the asymptotic series arising as saddle-point expansions of the quantum knot invariants around classical solutions are integer numbers~\cite{GGuM, GGuM2}, and they are closely related to the Dimofte--Gaiotto--Gukov index of the three-dimensional manifold~\cite{DGG}. 
In both types of quantum theories, it is conjectured in general, and verified numerically in some concrete cases, that the perturbative expansions of interest are resurgent series, and that their Borel transforms can be analytically continued as multivalued functions in the complex plane except for a discrete and infinite set of singular points, which are located along a finite number of vertical lines, spaced by integer multiples of some fundamental constant of the theory. Their arrangement in the complex plane is called a peacock pattern in~\cite{GGuM, GGuM2}. These infinite towers of periodic singularities lead to infinitely many Stokes constants, which are computed numerically in many examples, and often conjectured analytically in closed form. 

The formal connection between enumerative invariants and resurgence has been applied to the context of topological string theory on Calabi--Yau (CY) threefolds in the recent work of~\cite{GuM}.
A family of factorially divergent power series in the string coupling constant $g_s$ is obtained from the perturbative expansion of the conventional, unrefined topological string partition function, or exponential of the total free energy, in the weakly-coupled limit of the string $g_s \rightarrow 0$. It is observed in examples, and conjectured in general, that these asymptotic series lead to peacock patterns of singularities in their Borel plane and to infinite sets of integer Stokes constants.
However, differently from the previous two examples of $\mathcal{N}=2$ $SU(2)$ super Yang--Mills theory and complex Chern--Simons theory, a direct BPS interpretation of these integer invariants is still missing. 
Following the logic of~\cite{GuM}, we apply the tools of resurgence to a new family of asymptotic series which appear naturally in a dual strongly-coupled limit $g_s \rightarrow \infty$ of topological strings.

In this paper, we consider the perturbative expansion in the Planck constant $\hbar \propto g_s^{-1}$ of the logarithm of the fermionic spectral traces of the set of positive-definite, trace-class quantum operators which arise by quantization of the mirror curve to a toric CY. We obtain in this way a collection of factorially divergent perturbative series $\phi_{\bm{N}}(\hbar)$, which is indexed by a set of non-negative integers $\bm{N}$, one for each true complex modulus of the geometry. At fixed $\bm{N}$, the resurgent series $\phi_{\bm{N}}(\hbar)$ can be promoted to a full trans-series solution, unveiling an infinite collection of non-perturbative contributions. We conjecture that its resurgent analysis produces peacock-type arrangements of singularities in the complex Borel plane and an infinite set of rational Stokes constants. This numerical data is uniquely determined by the perturbative series under consideration, and it is yet intrinsically non-perturbative, representing a new, conjectural class of enumerative invariants of the CY, whose identification in terms of BPS counting is still to be understood. By means of the conjectural correspondence of~\cite{GHM, CGM2}, known as Topological Strings/Spectral Theory (TS/ST) correspondence, the semiclassical limit of the fermionic spectral traces turns out to be closely related to the all-orders WKB contribution to the total grand potential of the CY, and therefore to the total free energy of the refined topological string in the Nekrasov--Shatashvili (NS) limit. Note that the asymptotic series studied in~\cite{GuM} can be obtained as a perturbative expansion in $g_s$ of the same fermionic spectral traces. This dual weakly-coupled regime selects the worldsheet instanton contribution to the total grand potential, which is captured, in turn, by the standard topological string total free energy. 
We comment that, for each $\bm{N}$, the fermionic spectral traces are well-defined, analytic functions of $\hbar \in \IR_{>0}$ and provide a non-perturbative completion of the corresponding perturbative series $\phi_{\bm{N}}(\hbar)$. 
Although the definition of the proposed rational invariants as Stokes constants of appropriate perturbative expansions in topological string theory does not rely on the existence of a non-perturbative completion for the asymptotic series of interest, the fermionic spectral traces significantly facilitate the computational tasks undertaken in this paper, since they can be expressed as matrix integrals and factorized in holomorphic/anti-holomorphic blocks.

We perform a detailed resurgent analysis of the first fermionic spectral trace in the limit $\hbar \rightarrow 0$ for two well-known examples of toric CY threefolds, namely, local $\IP^2$ and local $\IF_0$. 
In the case of local $\IP^2$, the resurgent structure of the asymptotic series $\phi_1(\hbar)$ turns out to be analytically solvable, leading to proven exact formulae for the Stokes constants, which are rational numbers and simply related to an interesting sequence of integers. The Stokes constants have a transparent and strikingly simple arithmetic meaning as divisor sum functions, and they possess a generating function given by $q$-series, while the perturbative coefficients are encoded in explicit $L$-functions. We note that, differently from the dual case of~\cite{GuM}, the Stokes constants for the exponentiated series $\re^{\phi_1(\hbar)}$ appear to be generally complex numbers, and they can be expressed in terms of the Stokes constants of the series $\phi_1(\hbar)$ by means of a closed partition-theoretic formula. The symmetries and arithmetic properties described above are less easily accessible after exponentiation. 
Our analytic approach is then straightforwardly applied to the dual weakly-coupled limit $g_s \rightarrow 0$ of topological strings on the same background, confirming the results of the numerical study of~\cite{GuM}. We find that the Stokes constants and the perturbative coefficients are manifestly related to their semiclassical analogues. The duality between the weakly- and strongly-coupled scaling regimes re-emerges in a concrete and exact number-theoretic form.
Let us stress that we do not yet have a clear understanding of the possible generalization of the enticing number-theoretic formalism presented for local $\IP^2$ to arbitrary toric CY geometries.
The case of local $\IF_0$ immediately appears more complex. The resurgent structure of the first fermionic spectral trace is only accessible via numerical methods, which allow us to unveil the presence of logarithmic-type terms in the sub-leading asymptotics.

We note that, in contrast with what occurs in the dual limit $\hbar \rightarrow \infty$ studied in~\cite{GuM}, the semiclassical perturbative expansion of the first fermionic spectral trace of both local $\IP^2$ and local $\IF_0$ does not have a global exponential behavior of the form $\re^{-1/\hbar}$ at leading order, thus suggesting that there is no dual analogue of the conifold volume conjecture. We extend this observation to a general statement on the dominant semiclassical asymptotics of the fermionic spectral traces of toric CY threefolds.
We study the topological string total grand potential in an appropriate WKB 't Hooft-like regime associated to the semiclassical limit of the spectral theory, which selects the contribution from the total free energy of the refined topological string in the NS limit. After a suitable change of local symplectic frame in the moduli space of the geometry, we obtain a new, non-trivial analytic prediction of the TS/ST correspondence on the WKB asymptotic behavior of the fermionic spectral traces, which implies, in particular, the statement above.

This paper is organized as follows. 
In Section~\ref{sec: background}, we review the basic ingredients in the construction of quantum-mechanical operators from toric CY geometries and in the definition of the refined topological string theory, and its two one-parameter specializations, compactified on a toric CY background. The section ends with a short summary of the recent conjectural correspondence between topological strings and spectral theory, which offers a powerful, practical perspective for the computational tasks undertaken in this paper.
In Section~\ref{sec: main}, we provide the necessary background from the theory of resurgence, and we present a conjectural class of enumerative invariants of topological strings on a CY target as Stokes constants of appropriate asymptotic series, which arise naturally in the strong coupling limit $g_s \rightarrow \infty$ of the string theory, and which can be promoted to resurgent trans-series. The rational invariants studied in this paper represent a natural complement of the conjectural proposal of~\cite{GuM}, which addresses the dual weakly-coupled limit $g_s \rightarrow 0$.
In Section~\ref{sec: localP2}, we present an exact and complete solution to the resurgent structure of the first fermionic spectral trace of the local $\IP^2$ geometry in both scaling regimes $\hbar \rightarrow 0$ and $\hbar \rightarrow \infty$, and we provide an independent numerical analysis for the semiclassical limit.
In Section~\ref{sec: localF0}, we perform a preliminary numerical study of the resurgent structure of the first fermionic spectral trace of the local $\IF_0$ geometry for $\hbar \rightarrow 0$. 
In Section~\ref{sec: dual_limit}, we present a new analytic prediction on the asymptotics of the fermionic spectral traces of toric CY threefolds in the WKB double-scaling regime, and we analyze the example of local $\IP^2$ in detail. 
In Section~\ref{sec: conclusions}, we conclude and mention further perspectives to be addressed by future work and problems opened by this investigation.
There are four Appendices. 

\sectiono{From topological string theory to spectral theory and back} \label{sec: background}
In this section, we review the two-way connection between the spectral theory of quantum-mechanical operators and the topological string theory on toric CY manifolds, which builds upon notions of quantization and local mirror symmetry, and it has recently found an explicit formulation in the conjectural statement of~\cite{GHM, CGM2}, known as Topological Strings/Spectral Theory (TS/ST) correspondence. The conjecture, following the precursory work of~\cite{NS, MiMo, ACDKV, MP0, DMP, MP, HMO0, HMO, HMO2, HMMO, KaMa, HW}, leads to exact formulae for the spectral traces of the quantum operators in terms of the enumerative invariants of the CY, and it provides a non-perturbative realization of the topological string on this background.
We refer to~\cite{Geo1, Geo2, Geo3, Geo4} for an introduction to toric geometry and mirror symmetry and to~\cite{Al, Kl} for an introduction to topological string theory.

\subsection{Geometric setup and local mirror symmetry} \label{sec: geometry}
Let $X$ be a toric CY threefold and $\bm{t} = (t_1, \dots, t_s)$, where $s= b_2(X)$, be the complexified K\"ahler moduli of $X$. 
Local mirror symmetry pairs $X$ with a mirror CY threefold $\hat{X}$ in such a way that the theory of variations of complex structures of the mirror $\hat{X}$ is encoded in an algebraic equation of the form 
\be \label{eq: MC}
W(\re^x, \re^y)=0 \, , 
\ee
which describes a Riemann surface $\Sigma$ embedded in $\IC^* \times \IC^*$, called the mirror curve to $X$, and determines the B-model topological string theory on $\hat{X}$~\cite{M, BKMP}. 
We denote the genus of $\Sigma$ by $g_{\Sigma}$.
The complex deformation parameters of $\hat{X}$ can be divided into $g_{\Sigma}$ true moduli of the geometry, denoted by $\bm{\kappa} = (\kappa_1, \dots, \kappa_{g_{\Sigma}})$, and $r_{\Sigma} = s - g_{\Sigma}$ mass parameters, denoted by $\bm{\xi} = (\xi_1, \dots, \xi_{r_{\Sigma}})$~\cite{HKP, HKRS}. They are related to the Batyrev coordinates $\bm{z} = (z_1, \dots, z_s)$ of $\hat{X}$ by\footnote{We can choose the Batyrev coordinates in such a way that the first $g_{\Sigma}$ correspond to true moduli and the remaining $r_{\Sigma}$ correspond to mass parameters.}
\be \label{eq: Batyrev}
- \log z_i = \sum_{j=1}^{g_{\Sigma}} C_{i j} \mu_j + \sum_{k=1}^{r_{\Sigma}} \alpha_{i k} \log \xi_k \, , \quad i = 1, \dots, s \, , 
\ee
where the constant coefficients $C_{i j}, \alpha_{i k}$ are determined by the toric data of $X$, and the chemical potentials $\mu_j$ are defined by $\kappa_j = \re^{\mu_j}$, $j=1, \dots, g_{\Sigma}$.
The mirror curve $\Sigma$ can be identified with the family of equivalent canonical forms
\be \label{eq: canonical}
\CO_j(x,y) + \kappa_j = 0 \, , \quad j=1, \dots, g_{\Sigma} \, ,
\ee
where $\CO_j(x,y)$ is a polynomial in the variables $\re^{x}, \, \re^{y}$. Different canonical forms are related by $SL(2, \IZ)$-transformations and global translations in $x,\, y \in \IC$.

The complex moduli of the mirror $\hat{X}$ are related to the K\"ahler parameters of $X$ via the mirror map 
\be \label{eq: MM}
- t_i(\bm{z}) = \log z_i + \tilde{\Pi}_i(\bm{z}) \, , \quad i=1, \dots, s \, , 
\ee
where $\tilde{\Pi}_i(\bm{z})$ is a power series in $\bm{z}$ with finite radius of convergence. Together with Eq.~\eqref{eq: Batyrev}, it implies that
\be \label{eq: MM2}
t_i(\bm{\mu}, \bm{\xi}) = \sum_{j=1}^{g_{\Sigma}} C_{i j} \mu_j + \sum_{k=1}^{r_{\Sigma}} \alpha_{i k} \log \xi_k + \CO(\re^{-\mu_j}) \, , \quad i=1, \dots, s \, .
\ee
Following a choice of symplectic basis $A_i, B_i$, $i=1, \dots, s$, of one-cycles on the spectral curve $\Sigma$, the classical periods of the meromorphic differential one-form $\lambda = y(x)dx$, where the function $y(x)$ is locally defined by Eq.~\eqref{eq: MC}, satisfy
\be \label{eq: periods}
t_i(\bm{z}) \propto \oint_{A_i} \lambda \, , \quad \partial_{t_i} F_0(\bm{z}) \propto \oint_{B_i} \lambda \, , \quad i =1, \dots, s \, ,
\ee
where the function $F_0(\bm{z})$ is the classical prepotential of the geometry~\cite{CGPO}, which represents the genus zero amplitude of the B-model topological string on $\hat{X}$, that is, the generating functional of the genus zero Gromow--Witten invariants of $X$ convoluted with the mirror map. 

Following~\cite{GHM, CGM2}, the mirror curve in Eq.~\eqref{eq: MC} can be quantized by making an appropriate choice of reality conditions for the variables $x,y \in \IC$, promoting $x,y$ to self-adjoint Heisenberg operators $\mx, \my$ on the real line satisfying the commutation relation $[\mx, \my] = \text{i} \hbar$, and applying the standard Weyl prescription for ordering ambiguities. 
Thus, the functions $\CO_j(x, y)$ appearing in the canonical forms in Eq.~\eqref{eq: canonical} are uniquely associated to $g_{\Sigma}$ different Hermitian quantum-mechanical operators $\mO_j$, $j=1, \dots, g_{\Sigma}$, acting on $L^2(\IR)$. The mass parameters $\bm{\xi}$ become parameters of the operators $\mO_j$, and a Planck constant $\hbar \in \IR_{>0}$ is introduced as a quantum deformation parameter. 
We define the inverse operators as 
\be \label{eq: rho}
\rho_j = \mO_j^{-1} \, , j=1, \dots, g_{\Sigma} \, , 
\ee
acting on $L^2(\IR)$. The classical mirror map $t_i(\bm{z})$ in Eq.~\eqref{eq: MM} is consequently promoted to a quantum mirror map $t_i(\bm{z}, \hbar)$ given by
\be \label{eq: qMM}
- t_i(\bm{z}, \hbar) = \log z_i + \tilde{\Pi}_i(\bm{z}, \hbar), , \quad i=1, \dots, s \, , 
\ee
which reproduces the conventional mirror map in the semiclassical limit $\hbar \rightarrow 0$, and it is determined as an A-period of a quantum-corrected version of the differential $\lambda$ obtained via the all-orders, perturbative WKB approximation~\cite{MiMo, ACDKV}.

For simplicity, in the rest of this paper, we will often consider the case of toric (almost) del Pezzo CY threefolds, which are defined as the total space of the canonical line bundle on a toric (almost) del Pezzo surface $S$, that is,
\be \label{eq: delPezzo}
X = \mathcal{O}(K_S) \rightarrow S \, ,
\ee
also called local $S$. 
Examples of toric del Pezzo surfaces are the projective space $\IP^2$, the Hirzebruch surfaces $\IF_n$ for $n=0,1,2$, and the blowups of $\IP^2$ at $n$ points, denoted by $\mathcal{B}_n$, for $n=1,2,3$.
In this case, the mirror curve $\Sigma$ has genus one and, correspondingly, there are one true complex modulus $\kappa$, which is written in terms of the chemical potential $\mu$ as $\kappa = \re^{\mu}$, and $s-1$ mass parameters $\xi_k$, $k=1, \dots, s-1$. 
At leading order in the limit $\mu \rightarrow \infty$, the classical mirror map in Eq.~\eqref{eq: MM2} has the form
\be
t_i = c_i \mu + \sum_{k=1}^{s-1} \alpha_{i k} \log \xi_k + \CO(\re^{-\mu}) \, , \quad i = 1, \dots, s \, ,
\ee
where $c_i = C_{i 1}$, and the mirror curve in Eq.~\eqref{eq: MC} admits a single canonical parametrization
\be \label{eq: MC2}
\mathcal{O}_S(x, y) + \kappa = 0 \, .
\ee
We observe that, when appropriate symmetry conditions are applied to the mass parameters, the relation between the single Batyrev coordinate $z$ of $X$ and its true modulus $\kappa$ simplifies to $z = 1/\kappa^r$, where the value of $r$ is determined by the geometry. For example, such symmetry restrictions trivially apply to local $\IP^2$, which has $r=3$, while they correspond to imposing $\xi = 1$ in the case of local $\IF_0$, which has then $r=2$.
The canonical Weyl quantization scheme of~\cite{GHM} applied to the mirror curve in Eq.~\eqref{eq: MC2} produces a single Hermitian differential operator $\mO_S$ acting on $L^2(\IR)$, whose inverse is denoted by $\rho_S = \mO_S^{-1}$.
We stress that, in what follows, $X$ will be a generic toric CY threefold, and the simplified genus one case will only be considered when explicitly stated.

\subsection{Standard and NS topological strings from the refinement} \label{sec: strings}
The total free energy of the A-model conventional topological string with target $X$ is formally given by the generating series\footnote{The superscript WS in Eq.~\eqref{eq: WS} stands for worldsheet.}
\be \label{eq: WS}
F^{\text{WS}}(\bm{t}, g_s) = \sum_{g \ge 0} F_g(\bm{t}) \, g_s^{2g-2} \, ,
\ee
where the variable $g_s$ is the topological string coupling constant, and $F_g(\bm{t})$ is the free energy at fixed worldsheet genus $g \ge 0$.
In the so-called large radius limit $\Re(t_i) \gg 1$ of the moduli space of $X$, Eq.~\eqref{eq: WS} has the expansion~\cite{GV}
\be \label{eq: WS2}
F^{\text{WS}}(\bm{t}, g_s) = \frac{1}{6 g_s^2} \sum_{i,j,k=1}^s a_{i j k} t_i t_j t_k + \sum_{i=1}^s b_i t_i + \sum_{g \ge 2} C_g g_s^{2g-2} + F^{\text{GV}}(\bm{t}, g_s) \, ,
\ee
where $\bm{d} = (d_1, \dots, d_s)$ is a vector of non-negative integers representing a class in the two-homology group $\text{H}_2(X, \IZ)$, called vector of degrees, the coefficients $a_{i j k}, \, b_i$ are cubic and linear couplings characterizing the perturbative genus zero and genus one topological amplitudes, the constant $C_g$ is the so-called constant map contribution~\cite{BCOV1, BCOV2}, and $F^{\text{GV}}(\bm{t}, g_s)$ is given by the formal power series
\be \label{eq: GV}
F^{\text{GV}}(\bm{t}, g_s) = \sum_{g \ge 0} \sum_{\bm{d}} \sum_{w=1}^{\infty} n_g^{\bm{d}} \, \frac{1}{w} \left( 2 \sin\frac{w g_s}{2} \right)^{2g-2} \, \re^{-w \bm{d} \cdot \bm{t}} \, ,
\ee
where $n_g^{\bm{d}} \in \IZ$ is the Gopakumar--Vafa enumerative invariant~\cite{PT} of $X$ at genus $g$ and degree $\bm{d}$. 

When defined on a toric CY manifold, the topological string partition function can be engineered as a special limit of the instanton partition function of Nekrasov~\cite{N}. A more general theory, known as refined topological string theory~\cite{IKV, KW, HK}, is constructed by splitting the string coupling constant into two independent parameters as $g_s^2 = - \epsilon_1 \epsilon_2$, where $\epsilon_1, \epsilon_2$ correspond to the two equivariant rotations of the space-time $\IC^2$. Together with $g_s$, a second coupling constant $\hbar = \epsilon_1 + \epsilon_2$ is introduced and identified with the quantum deformation parameter which appears in the quantization of the classical spectral curve in Eq.~\eqref{eq: MC} in the mirror B-model.
The total free energy of the A-model refined topological string on $X$ at large radius has a double perturbative expansion in $g_s$ and $\hbar$ of the form~\cite{CB, CB2, BMS}
\be \label{eq: ref}
F(\bm{t}, \epsilon_1, \epsilon_2) = \sum_{g, n \ge 0} F_{g, n}(\bm{t}) \, g_s^{2g-2} \, \hbar^{2n} \, , 
\ee
from which the genus expansion of the standard topological string in Eq.~\eqref{eq: WS} is recovered in the limit $g_s = \epsilon_1 = - \epsilon_2$, and we have $F_g(\bm{t}) = F_{g, 0}(\bm{t})$, $g \ge 0$.
Another remarkable one-parameter specialization of the refined theory is obtained when one of the two equivariant parameters $\epsilon_1, \epsilon_2$ is sent to zero and the other is kept finite, \textit{e.g.}, $\epsilon_2 \rightarrow 0$ while $\hbar = \epsilon_1$ is fixed, which is known as the Nekrasov--Shatashvili (NS) limit~\cite{NS}. 
Since the refined total free energy in Eq.~\eqref{eq: ref} has a simple pole in this limit, the NS total free energy is defined as the one-parameter generating series\footnote{The superscript NS in Eq.~\eqref{eq: NS} stands for Nekrasov--Shatashvili.}
\be \label{eq: NS}
F^{\rm NS}(\bm{t}, \hbar) = \lim_{\epsilon_2 \rightarrow 0} -\epsilon_2 F(\bm{t}, \epsilon_1, \epsilon_2) = \sum_{n \ge 0} F_n^{\rm NS}(\bm{t}) \, \hbar^{2n-1} \, ,
\ee
where $F_n^{\rm NS}(\bm{t}) = F_{0,n}(\bm{t})$ denotes the NS topological amplitude at fixed order $n$ in $\hbar$. 
In the refined framework, the Gopakumar--Vafa invariants are generalized to a wider set of integer enumerative invariants, called the refined BPS invariants~\cite{CKK, NO}. We denote them by $N_{j_{\rm L}, j_{\rm R}}^{\bm{d}}$, where $j_{\rm L}, j_{\rm R}$ are two non-negative half-integers, and $\bm{d}$ is the degree vector. The perturbative expansion at large radius of the NS total free energy is expressed as the generating functional
\be \label{eq: NS2}
\ba
F^{\rm NS}(\bm{t}, \hbar) &= \frac{1}{6 \hbar} \sum_{i,j,k=1}^s a_{i j k} t_i t_j t_k + \hbar \sum_{i=1}^s b_i^{\rm NS} t_i \\
&+ \sum_{j_{\rm L}, j_{\rm R}} \sum_{\bm{d}} \sum_{w=1}^{\infty} N_{j_{\rm L}, j_{\rm R}}^{\bm{d}} \, \frac{\sin\frac{\hbar w}{2}(2 j_{\rm L}+1) \, \sin\frac{\hbar w}{2}(2 j_{\rm R}+1)}{2 w^2 \sin^3\frac{\hbar w}{2}} \, \re^{-w \bm{d} \cdot \bm{t}} \, ,
\ea
\ee
which reproduces Eq.~\eqref{eq: NS} when expanded in powers of $\hbar$. The coefficients $a_{i j k}$ are the same ones that appear in Eq.~\eqref{eq: GV}, while the constants $b_i^{\rm NS}$ can be obtained via mirror symmetry~\cite{KW, HK}.
Furthermore, $F_0^{\rm NS}(\bm{t}) = F_0(\bm{t})$, and the higher-order NS free energies are given by the perturbative WKB quantum corrections to the classical prepotential~\cite{MiMo, ACDKV}. Namely, the refined topological string free energy in the NS limit is the quantum prepotential associated to the quantum-deformed version of the classical B-period of the differential $\lambda$ in Eq.~\eqref{eq: periods}.

We recall that the total grand potential of topological string theory on $X$ is defined as the sum~\cite{HMMO}
\be \label{eq: J_total}
J(\bm{\mu}, \bm{\xi}, \hbar) = J^{\rm WS}(\bm{\mu}, \bm{\xi}, \hbar) + J^{\rm WKB}(\bm{\mu}, \bm{\xi}, \hbar) \, .
\ee
The worldsheet grand potential is obtained from the generating functional of Gopakumar--Vafa invariants of $X$ in Eq.~\eqref{eq: GV} as
\be \label{eq: J_WS}
J^{\rm WS}(\bm{\mu}, \bm{\xi}, \hbar) = F^{\text{GV}} \left(\frac{2 \pi}{\hbar} \bm{t}(\hbar) + \pi \ri \bm{B}, \frac{4 \pi^2}{\hbar} \right) \, ,
\ee
where $\bm{t}(\hbar)$ is the quantum mirror map in Eq.~\eqref{eq: qMM}, and $\bm{B}$ is a constant vector determined by the geometry, called B-field, whose presence has the effect to introduce a sign $(-1)^{w \bm{d} \cdot \bm{B}}$ in the series in Eq.~\eqref{eq: GV}. 
The all-genus worldsheet generating functional above encodes the non-perturbative contributions in $\hbar$ due to complex instantons, which are contained in the standard topological string.
Note that there is a strong-weak coupling duality between the spectral theory of the operators arising from the quantization of the mirror curve $\Sigma$ and the standard topological string theory on $X$. Namely, 
\be \label{eq: duality}
g_s = \frac{4 \pi^2}{\hbar} \, .
\ee
The WKB grand potential is obtained from the NS generating functional in Eq.~\eqref{eq: NS2} as
\be \label{eq: J_WKB}
\ba
J^{\rm WKB}(\bm{\mu}, \bm{\xi}, \hbar) &= \sum_{i=1}^s \frac{t_i(\hbar)}{2 \pi} \frac{\partial F^{\rm NS}(\bm{t}(\hbar), \hbar)}{\partial t_i} + \frac{\hbar^2}{2 \pi} \frac{\partial}{\partial \hbar} \left( \frac{1}{\hbar} F^{\rm NS}(\bm{t}(\hbar), \hbar) \right) \\
&+ \frac{2 \pi}{\hbar} \sum_{i=1}^s b_i t_i(\hbar) + A(\bm{\xi}, \hbar) \, ,
\ea
\ee
where the derivative with respect to $\hbar$ in the second term on the RHS does not act on the $\hbar$-dependence of the quantum mirror map $\bm{t}(\hbar)$. 
The coefficients $b_i$ are the same ones appearing in Eq.~\eqref{eq: WS2}, while the function $A(\bm{\xi}, \hbar)$ is not known in closed form for arbitrary geometries, although it has been conjectured in many examples.
The all-orders WKB generating functional above takes into account the perturbative corrections in $\hbar$ to the quantum-mechanical spectral problem associated to $X$, which are captured by the NS refined topological string.
The total grand potential of $X$ can then be expressed as a formal power series expansion in the large radius limit $t_i \rightarrow \infty$ with the structure 
\be \label{eq: J_total2}
J(\bm{\mu}, \bm{\xi}, \hbar) = \frac{1}{12 \pi \hbar} \sum_{i,j,k =1}^s a_{i j k} t_i t_j t_k + \sum_{i = 1}^s \left( \frac{2 \pi}{\hbar} b_i + \frac{\hbar}{2 \pi} b_i^{\rm NS} \right) t_i + \CO(\re^{-t_i}, \re^{- 2 \pi t_i/\hbar}) \, ,
\ee
where the infinitesimally small corrections in $\re^{-t_i}$, $\re^{- 2 \pi t_i/\hbar}$ have $\hbar$-dependent coefficients. Rigorous results on the properties of convergence of this expansion are missing. However, extensive evidence suggests that it is analytic in a neighborhood of $t_i \rightarrow \infty$ when $\hbar$ is real~\cite{HMMO, CGM}, while it appears to inherit the divergent behavior of the generating functionals in Eqs.~\eqref{eq: WS2} and~\eqref{eq: NS2} for complex $\hbar$.

\subsection{The TS/ST correspondence} \label{sec: TSST}
Recall that the quantization of the mirror curve $\Sigma$ to the toric CY threefold $X$ naturally leads to the quantum-mechanical operators $\rho_j$, $j=1, \dots, g_{\Sigma}$, acting on $L^2(\IR)$, which are defined in Eq.~\eqref{eq: rho}.
It was conjectured in~\cite{GHM, CGM2}, and rigorously proved in~\cite{KM} in several examples, that the operators $\rho_j$ are positive-definite and of trace class, therefore possessing discrete, positive spectra, provided the mass parameters $\bm{\xi}$ of the mirror CY $\hat{X}$ satisfy suitable reality and positivity conditions. 
As shown in~\cite{CGM2, CGuM}, one can define a spectral (or Fredholm) determinant $\Xi_X(\bm{\kappa}, \bm{\xi}, \hbar)$ associated to the set of operators $\rho_j$, which is an entire function on the moduli space parametrized by $\bm{\kappa}$. 
The fermionic spectral traces $Z_X(\bm{N}, \bm{\xi}, \hbar)$, where $N_j$ is a non-negative integer for $j=1, \dots, g_{\Sigma}$, are then defined by a power series expansion of the analytically continued spectral determinant $\Xi_X(\bm{\kappa}, \bm{\xi}, \hbar)$ around the point $\bm{\kappa}=0$ in the moduli space of $X$, known as the orbifold point. Namely, 
\be \label{eq: expXi}
\Xi_X(\bm{\kappa}, \bm{\xi}, \hbar) = \sum_{N_1 \ge 0} \cdots \sum_{N_{g_{\Sigma}} \ge 0} Z_X(\bm{N}, \bm{\xi}, \hbar) \, \kappa_1^{N_1} \cdots \kappa_{g_{\Sigma}}^{N_{g_{\Sigma}}} \, , 
\ee
with $Z_X(0, \dots, 0, \bm{\xi}, \hbar) = 1$. Classical results in Fredholm theory~\cite{S1, S2, Gr} provide explicit determinant expressions for the fermionic spectral traces, which can be regarded as multi-cut matrix model integrals\footnote{The connection between fermionic spectral traces and matrix models has been developed in~\cite{MZ, KMZ}.}.

Based on the previous insights of~\cite{MP, HMO, HMO2, HMMO, KaMa, HW}, a conjectural duality was recently proposed in~\cite{GHM, CGM2}, which relates the topological string theory on a toric CY manifold to the spectral theory of the quantum-mechanical operators on the real line which are obtained by quantization of the corresponding mirror curve. This is known as the TS/ST correspondence, and it is now supported by a large amount of evidence obtained in applications to concrete examples~\cite{KM, MZ, KMZ, GuKMR, OK}. We refer to the detailed review in~\cite{reviewM} and references therein.
The main conjectural statement of the TS/ST correspondence provides exact expressions for the spectral determinant and the fermionic spectral traces in terms of the standard and NS topological string amplitudes on $X$. More precisely,
\be \label{eq: GHM}
\Xi_X(\bm{\kappa}, \bm{\xi}, \hbar) = \sum_{\bm{n} \in \IZ^{g_{\Sigma}}} \exp(J(\bm{\mu} + 2 \pi \text{i} \bm{n}, \bm{\xi}, \hbar)) = \re^{J(\bm{\mu}, \bm{\xi}, \hbar)} \Theta(\bm{\mu}, \bm{\xi}, \hbar) \, , 
\ee
where the sum over $\bm{n} \in \IZ^{g_{\Sigma}}$ produces a periodic function in the chemical potentials $\mu_j$, which can be equivalently recast by factoring out a quantum-deformed Riemann theta function $\Theta(\bm{\mu}, \bm{\xi}, \hbar)$. 
It follows that the fermionic spectral traces $Z_X(\bm{N}, \bm{\xi}, \hbar)$, $N_j \ge 0$, $j=1, \dots, g_{\Sigma}$, are determined by the orbifold expansion of the topological string theory on $X$.
Note that the expression on the RHS of Eq.~\eqref{eq: GHM} can be interpreted as a well-defined large-$\mu_j$ expansion in powers of $\re^{-\mu_j}$, $\re^{-2 \pi \mu_j/\hbar}$. Indeed, the total grand potential and the quantum theta function appear to have a common region of converge in a neighborhood of the limit $\mu_j \rightarrow \infty$, which corresponds to the large radius point of moduli space. However, being the spectral determinant an entire function of $\bm{\kappa}$, the conjecture in Eq.~\eqref{eq: GHM} implies that such a product in the RHS is, indeed, entire in $\bm{\mu}$.
Moreover, Eqs.~\eqref{eq: GHM} and~\eqref{eq: expXi} lead to an integral formula for $Z_X(\bm{N}, \bm{\xi}, \hbar)$ as an appropriate residue at the origin $\bm{\kappa}=0$. Namely,~\cite{HMO, GHM, CGM2}
\be \label{eq: contour}
Z_X(\bm{N}, \bm{\xi}, \hbar) = \frac{1}{(2 \pi \text{i})^{g_{\Sigma}}} \int_{- \text{i} \infty}^{\text{i} \infty} \rd \mu_1 \cdots \int_{- \text{i} \infty}^{\text{i} \infty} \rd \mu_{g_{\Sigma}} \, \exp(J(\bm{\mu}, \bm{\xi}, \hbar) - \bm{N} \cdot \bm{\mu}) \, ,
\ee
where the integration contour along the imaginary axes can be suitably deformed to make the integral convergent. 
Because of the trace-class property of the quantum operators $\rho_j$, the fermionic spectral traces  $Z_X(\bm{N}, \bm{\xi}, \hbar)$ are well-defined functions of $\hbar \in \IR_{>0}$, and, although being initially defined for positive integer values of $N_j$, the Airy-type integral in Eq.~\eqref{eq: contour} naturally extends them to entire functions of $\bm{N} \in \IC^{g_{\Sigma}}$~\cite{CGM}. 
In what follows, for simplicity, we will drop from our notation the explicit dependence on $\bm{\xi}$ of $\Xi_X(\bm{\kappa}, \bm{\xi}, \hbar)$ and $Z_X(\bm{N}, \bm{\xi}, \hbar)$.

\sectiono{Stokes constants in topological string theory} \label{sec: main}
In this section, we review how the resurgent analysis of formal power series with factorial growth unveils a universal mathematical structure, which involves a set of numerical data called Stokes constants. 
Following the recent works of~\cite{GrGuM, GGuM, GGuM2, GuM, GuM2, GuM3}, we apply the theory of resurgence to the asymptotic series that arise naturally as appropriate perturbative expansions in a strongly-coupled limit of the topological string on a toric CY threefold, and we make a general proposal on the resurgent structure of these series. 
See~\cite{ABS, MS, Dorigoni} for a formal introduction to the resurgence of asymptotic expansions and~\cite{lecturesM, bookM} for its application to gauge and string theories. The resurgent structure of topological string theory on the special geometry of the resolved conifold has been studied in~\cite{PS, ASTT, GHN, AHT}.

\subsection{Notions from the theory of resurgence} \label{sec: resurgence}
Let $\phi(z)$ be a factorially divergent formal power series of the form
\be \label{eq: phi}
\phi(z) = z^{-\alpha} \sum_{n=0}^{\infty} a_n z^n \in z^{-\alpha} \IC[\![z]\!] \, , \quad a_n \sim A^{-n} n! \quad n \gg 1 \, ,
\ee
for some constants $\alpha \in \IR \backslash \IZ_{+}$ and $A \in \IR$, which is a Gevrey-1 asymptotic series. Its Borel transform
\be \label{eq: phihat}
\hat{\phi}(\zeta) = \sum_{k=0}^{\infty} \frac{a_k}{\Gamma(k-\alpha+1)} \zeta^{k-\alpha}
\ee
is a holomorphic function in an open neighborhood of $\zeta=0$ of radius $|A|$. 
When extended to the complex $\zeta$-plane, also known as Borel plane, $\hat{\phi}(\zeta)$ will show a (possibly infinite) set of singularities $\zeta_{\omega} \in \IC$, which we label by the index $\omega \in \Omega$. 
A ray in the Borel plane of the form 
\be \label{eq: ray}
\CC_{\theta_{\omega}} = \re^{\text{i} \theta_{\omega}} \IR_{+} \, , \quad \theta_{\omega} = \arg (\zeta_{\omega}) \, ,
\ee 
which starts at the origin and passes through the singularity $\zeta_{\omega}$, is called a Stokes ray. 
The Borel plane is partitioned into sectors which are bounded by the Stokes rays. In each of these sectors, the Borel transform converges to a generally different holomorphic function.

We recall that a Gevrey-1 asymptotic series is called resurgent if its Borel transform has finitely many singularities on every finite line issuing from the origin and if there exists a path circumventing these singularities along which it can be analytically continued.
If, additionally, its Borel transform has only logarithmic singularities and simple poles, then it is called simple resurgent.
We will assume here that all formal power series are resurgent.
If the singularity $\zeta_{\omega}$ is a logarithmic branch cut, the local expansion of the Borel transform in Eq.~\eqref{eq: phihat} around it has the form
\be \label{eq: Stokes0}
\hat{\phi}(\zeta) = - \frac{S_{\omega}}{2 \pi \text{i}} \log(\zeta - \zeta_{\omega}) \hat{\phi}_{\omega}(\zeta - \zeta_{\omega}) + \dots \, ,
\ee
where the dots denote regular terms in $\zeta - \zeta_{\omega}$, and $S_{\omega} \in \IC$ is the Stokes constant at $\zeta_{\omega}$. If we introduce the variable $\xi = \zeta - \zeta_{\omega}$, the function
\be \label{eq: phihat2}
\hat{\phi}_{\omega}(\xi) = \sum_{k=0}^{\infty} \hat{a}_{k, \omega} \xi^{k-\beta} \, ,
\ee
where $\beta \in \IR \backslash \IZ_{+}$, is locally analytic at $\xi = 0$, and it can be regarded as the Borel transform of the Gevrey-1 asymptotic series 
\be \label{eq: phi2}
\phi_{\omega}(z) = z^{-\beta} \sum_{n=0}^{\infty} a_{n, \omega} z^n \in z^{-\beta} \IC[\![z]\!] , \, \quad a_{n, \omega} = \Gamma(n-\beta+1) \, \hat{a}_{n, \omega} \, .
\ee
Note that the value of the Stokes constant $S_{\omega}$ depends on a choice of normalization of the series $\phi_{\omega}(z)$.
If the analytically continued Borel transform $\hat{\phi}(\zeta)$ in Eq.~\eqref{eq: phihat} does not grow too fast at infinity, its Laplace transform at an arbitrary angle $\theta$ in the Borel plane is given by\footnote{Roughly, we require that the Borel transform grows at most exponentially in an open sector of the Borel plane containing the angle $\theta$.}
\be \label{eq: Laplace}
s_{\theta}(\phi)(z) = \int_0^{\re^{\text{i} \theta} \infty} \re^{-\zeta} \hat{\phi}(\zeta z) \, \rd \zeta = z^{-1}  \int_0^{\re^{\text{i} \theta} \infty} \re^{-\zeta/z} \hat{\phi}(\zeta) \, \rd \zeta \, , 
\ee
and its asymptotics near the origin reconstructs the original, divergent formal power series $\phi(z)$. 
If the Laplace integral in Eq.~\eqref{eq: Laplace}, for some choice of angle $\theta$, exists in some region of the complex $z$-plane, we say that the series $\phi(z)$ is Borel summable, and we call $s_{\theta}(\phi)(z)$ the Borel resummation of $\phi(z)$ along the direction $\theta$.
Note that the Borel resummation inherits the sectorial structure of the Borel transform. It is a locally analytic function with discontinuities at the special rays identified by 
\be
\arg(z)=\arg(\zeta_{\omega}) \, , \quad \omega \in \Omega \, .
\ee 
The discontinuity across $\theta$ is the difference between the Borel resummations along two rays in the complex $\zeta$-plane which lie slightly above and slightly below $\CC_{\theta}$. Namely,
\be \label{eq: disc}
\text{disc}_{\theta}\phi(z) = s_{\theta_+}(\phi)(z) - s_{\theta_-}(\phi)(z) = \int_{\mathcal{C}_{\theta_+} - \, \mathcal{C}_{\theta_-}} \re^{-\zeta} \hat{\phi}(\zeta z) \,  \rd \zeta \, , 
\ee
where $\theta_{\pm}= \theta \pm \epsilon$, for some small positive angle $\epsilon$, and $\mathcal{C}_{\theta_{\pm}}$ are the corresponding rays.
A standard contour deformation argument shows that the two lateral Borel resummations differ by exponentially small terms. More precisely, 
\be \label{eq: Stokes1}
\text{disc}_{\theta}\phi(z) = \sum_{\omega  \in \Omega_{\theta}} S_{\omega} \re^{-\zeta_{\omega}/z} s_{\theta_-}(\phi_{\omega})(z) \, ,
\ee
where the index $\omega$ labels the singularities $\zeta_{\omega}$ such that $\arg (\zeta_{\omega}) = \theta$, while $\phi_{\omega}(z)$ is the formal power series in Eq.~\eqref{eq: phi2}, and the complex numbers $S_{\omega}$ are the same Stokes constants which appear in Eq.~\eqref{eq: Stokes0}. 
If we regard the lateral Borel resummations as operators, the Stokes automorphism $\mathfrak{S}_{\theta}$ is defined by the convolution
\be \label{eq: autStokes}
s_{\theta_+} = s_{\theta_-} \circ \mathfrak{S}_{\theta} \, ,
\ee
and the discontinuity formula in Eq.~\eqref{eq: Stokes1} has the equivalent, more compact form
\be
\mathfrak{S}_{\theta}(\phi) = \phi +  \sum_{\omega  \in \Omega_{\theta}} S_{\omega} \re^{-\zeta_{\omega}/z} \phi_{\omega} \, .
\ee
Moreover, the Stokes automorphism can be written as
\be \label{eq: alienStokes}
\mathfrak{S}_{\theta} = \exp \left( \sum_{\omega \in \Omega_{\theta}} \re^{-\zeta_{\omega}/z} \Delta_{\zeta_{\omega}} \right) \, ,
\ee
where $\Delta_{\zeta_{\omega}}$ is the alien derivative associated to the singularity $\zeta_{\omega}$, $\omega \in \Omega_{\theta}$. A short introduction to alien calculus is provided in Appendix~\ref{app: alien}.

As we have shown, moving from one formal power series to a new one encoded in the singularity structure of its Borel transform, we eventually build a whole family of asymptotic series from the single input in Eq.~\eqref{eq: phi}. We can then repeat the procedure with each new series obtained in this way. We denote by $S_{\omega \omega'} \in \IC$ the Stokes constants of the asymptotic series $\phi_{\omega}(z)$. Let us define the basic trans-series
\be
\Phi_{\omega}(z) = \re^{-\zeta_{\omega}/z} \phi_{\omega}(z) \, ,
\ee
for each $\omega \in \Omega$, such that its Borel resummation along $\theta$ is given by
\be
s_{\theta}(\Phi_{\omega})(z) = \re^{-\zeta_{\omega}/z} s_{\theta}(\phi_{\omega})(z) \, ,
\ee
and the corresponding Stokes automorphism acts as
\be
\mathfrak{S}_{\theta}(\Phi_{\omega}) = \Phi_{\omega} + \sum_{\omega' \in \Omega_{\theta}} S_{\omega \omega'} \Phi_{\omega'} \, .
\ee
The minimal resurgent structure associated to $\phi(z)$ is defined as the smallest set of basic trans-series which resurge from it and which forms a closed set under Stokes automorphisms. We denote it by~\cite{GuM}
\be
\mathfrak{B}_{\phi} = \left\{ \Phi_{\omega}(z) \right\}_{\omega \in \bar{\Omega}} \, , \quad 
\ee
where $\bar{\Omega} \subseteq \Omega$. We observe that the minimal resurgent structure does not necessarily include all the basic trans-series arising from $\phi(z)$. 
As pointed out in~\cite{GuM}, an example of this situation is provided by complex Chern--Simons theory on the complement of a hyperbolic knot.
In this paper, we will focus on the (possibly infinite-dimensional) matrix of Stokes constants indexed by the distinct basic trans-series in the minimal resurgent structure of $\phi(z)$, as they incorporate information about the non-analytic content of the original asymptotic series. Namely,
\be
\mathcal{S}_{\phi} = \{S_{\omega \omega'} \}_{\omega, \omega' \in \bar{\Omega}} \, .
\ee

\subsection{The resurgent structure of topological strings} \label{sec: resurgent_strings}
We want to apply the machinery described above to understand the resurgent structure of the asymptotic series which arise naturally from the perturbative expansion of the refined topological string on a toric CY threefold $X$ in a specific scaling limit of the coupling constants. Under the assumption that the given series are resurgent, this will give us access to the hidden sectors of the topological string which are invisible in perturbation theory.
Let us go back to the definition of the fermionic spectral traces $Z_X(\bm{N}, \hbar)$, $N_j \ge 0$,  $j = 1, \dots, g_{\Sigma}$, in Eq.~\eqref{eq: expXi}. 
Building on numerical evidence obtained in some concrete genus-one examples, it was conjectured in~\cite{GuM} that the Stokes constants appearing in the resurgent structure of these objects, when perturbatively expanded in the limit $\hbar \rightarrow \infty$, with $\bm{N}$ fixed, are non-trivial integer invariants of the geometry related to the counting of BPS states. 
In this paper, we will illustrate how the same resurgent machine advocated in~\cite{GuM} can be applied to the asymptotic series that emerge in the dual semiclassical limit $\hbar \rightarrow 0$ of the spectral theory. Note that this corresponds to the strongly-coupled regime $g_s \rightarrow \infty$ of the topological string theory via the TS/ST correspondence. We will explore, in practice, the best known examples of toric CY threefolds with one true complex modulus.

Along the lines of~\cite{GuM}, let us describe a conjectural proposal for the resurgent structure of the topological string in the limit $\hbar \rightarrow 0$, which is supported by the concrete results obtained in the examples of local $\IP^2$ and local $\IF_0$ in Sections~\ref{sec: localP2} and~\ref{sec: localF0}, respectively.
We consider the semiclassical perturbative expansion of the fermionic spectral trace $Z_X(\bm{N}, \hbar)$, at fixed $\bm{N}$. The corresponding family of asymptotic series 
\be \label{eq: psiN}
\psi_{\bm{N}}(\hbar) = Z_X(\bm{N}, \hbar \rightarrow 0) \, ,
\ee
indexed by the set of non-negative integers $\bm{N}$, will be the main object of study in this paper. 
We will comment in Section~\ref{sec: dual_limit} how these asymptotic expansions can be independently defined on the topological strings side of the TS/ST correspondence via the integral in Eq.~\eqref{eq: contour}. 
We denote
\be \label{eq: phiN}
\phi_{\bm{N}}(\hbar) = \log \psi_{\bm{N}}(\hbar) \, ,
\ee
for each choice of $\bm{N}$.
\begin{remark}
Note that, in order to perform a resurgent analysis of the fermionic spectral traces, it will be necessary to consider the case of complex $\hbar$.
The issue of the complexification of $\hbar$, or, equivalently, of $g_s$, in the context of the TS/ST correspondence has been addressed in various studies~\cite{Ka, KS1, KS2, GM2}. In this paper, we will assume that the TS/ST correspondence can be extended to $\hbar \in \IC' = \IC \backslash \IR_{\le 0}$ in such a way that the quantum operators $\rho_j$ remain of trace class, and the fermionic spectral traces $Z_X(\bm{N}, \hbar)$ are analytically continued to $ \hbar \in \IC'$. This assumption will be explicitly tested in the examples considered in this paper.
\end{remark}

\begin{figure}[htb!]
\center
 \includegraphics[width=0.4\textwidth]{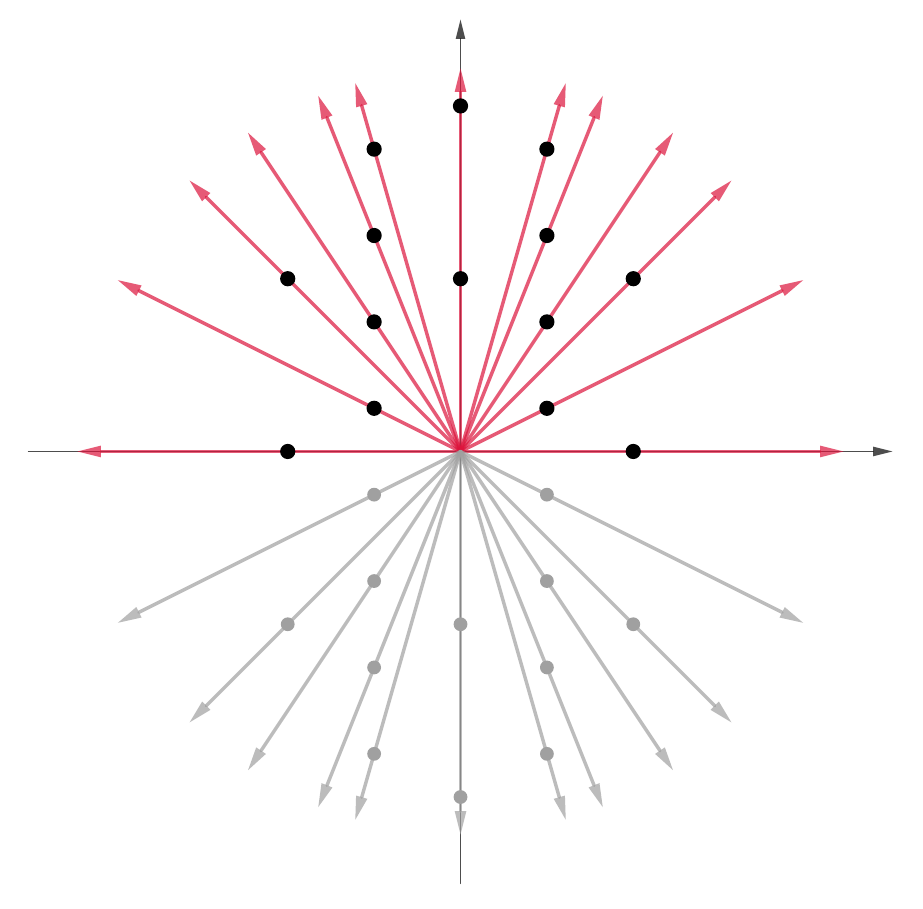}
 \caption{Infinite towers of singularities and the peacock arrangement of Stokes rays in the complex Borel plane.}
 \label{fig: peacock}
\end{figure}
We will verify in examples that the series $\phi_{\bm{N}}(\hbar)$ are Gevrey-1 and (simple) resurgent, and we assume that this is the case in general. 
Each of the formal power series in Eq.~\eqref{eq: phiN} is associated with a minimal resurgent structure $\mathfrak{B}_{\phi_{\bm{N}}}$ and a corresponding matrix of Stokes constants $\mathcal{S}_{\phi_{\bm{N}}}$. 
For fixed $\bm{N}$, we observe a finite number of Gevrey-1 asymptotic series 
\be \label{eq: seriesR}
\phi_{\sigma ; \bm{N}}(\hbar) \, , \quad \sigma \in \{ 0 , \dots , l \} \, ,
\ee
which resurge from the original perturbative expansion $\phi_{\bm{N}}(\hbar) = \phi_{0 ; \bm{N}}(\hbar)$, where the positive integer $l$ depends on $\bm{N}$ and on the CY geometry. 
For each value of $\sigma$, the singularities of the Borel transform $\hat{\phi}_{\sigma ; \bm{N}}(\zeta)$ are located along infinite towers in the Borel $\zeta$-plane, and every two singularities in the same tower are spaced by an integer multiple of some constant $\mathcal{A} \in \IC$, which depends on $\bm{N}$ and on the CY geometry. 
Such a global arrangement is known as a peacock pattern. See Fig.~\ref{fig: peacock} for a schematic illustration. It was recently conjectured in~\cite{GuM} that peacock patterns are typical of theories controlled by a quantum curve in exponentiated variables\footnote{Peacock patterns have been previously observed in complex Chern--Simons theory on the complement of a hyperbolic knot~\cite{GGuM, GGuM2} and in the weakly-coupled topological string theory on a toric CY threefold~\cite{C-SMS, GuM}.}.
Each asymptotic series $\phi_{\sigma ; \bm{N}}(\hbar)$ gives rise to an infinite family of basic trans-series, labelled by a non-negative integer $n$, that is,
\be \label{eq: bts}
\Phi_{\sigma, n ; \bm{N}}(\hbar) = \phi_{\sigma ; \bm{N}}(\hbar) \re^{- n \frac{\mathcal{A}}{\hbar}} \, , \quad n \in \IN \, ,
\ee
and the minimal resurgent structure has the particular form
\be
\mathfrak{B}_{\phi_{\bm{N}}} = \{ \Phi_{\sigma, n ; \bm{N}}(\hbar) \}_{\sigma = 0, \dots, l , n \in \IN} \, .
\ee
For fixed $\bm{N}$, the Stokes constants are labelled by two indices $\sigma, \sigma' = 0 , \dots , l$ and by the integer $n \in \IN$. Let us denote them as $S_{\sigma \sigma' , n ; \bm{N}}$. As we will find in examples, we expect the Stokes constants to be rational numbers, after choosing a canonical normalization of the asymptotic series in Eq.~\eqref{eq: seriesR}, and to be closely related to non-trivial sequences of integer constants. Moreover, we conjecture that they can be naturally organized as coefficients of generating functions in the form of $q$-series, which are determined by the original perturbative expansion in a unique way. Schematically, 
\be
S_{\sigma \sigma' ; \bm{N}}(q) = \sum_{n \in \IN}  S_{\sigma \sigma' , n ; \bm{N}} \, q^n \, ,
\ee
which we further expect to be intimately related to a non-trivial collection of topological invariants of the theory. We stress that, analogously to~\cite{GuM}, we do not yet have a direct, physical or geometrical interpretation of the proposed enumerative invariants. However, the exact solution to the resurgent structure of the first fermionic spectral trace of the local $\IP^2$ geometry, which is presented in Section~\ref{sec: localP2} for both limits $\hbar \rightarrow 0$ and $\hbar \rightarrow \infty$, shows that, when looking at the logarithm of the fermionic spectral trace, the Stokes constants have a manifest and strikingly simple arithmetic meaning as divisor sum functions. Moreover, the perturbative coefficients are encoded in $L$-functions which factorize explicitly as products of zeta functions and Dirichlet $L$-functions, while the duality between the weakly- and strongly-coupled scaling regimes emerges in anumber-theoretic form.
On the other hand, the Stokes constants for the exponentiated series in Eq.~\eqref{eq: psiN} appear to be generally complex numbers, and they can be expressed in terms of the Stokes constants of the series $\phi_{\bm{N}}(\hbar)$ by means of a closed partition-theoretic formula.
Let us stress that we do not yet have a clear understanding of the possible generalization of the arithmetic construction presented in Section~\ref{sec: localP2} to arbitrary toric CY geometries. The case of local $\IF_0$, which is analyzed partially in Section~\ref{sec: localF0} via numerical methods, is significantly more complex as the first resurgent asymptotic series shows a leading-order behavior of logarithmic type. 

\sectiono{The example of local \texorpdfstring{$\IP^2$}{P2}} \label{sec: localP2} 
The simplest example of a toric del Pezzo CY threefold is the total space of the canonical bundle over $\IP^2$, that is, $\CO(-3) \rightarrow \IP^2$, known as the local $\IP^2$ geometry. 
It has one true complex modulus $\kappa$ and no mass parameters. Its moduli space is identified with the one-parameter family of mirror curves described by the equation
\be \label{eq: mcP2}
\re^x + \re^y + \re^{-x-y} + \kappa = 0 \, , \quad x,y \in \IC \, , 
\ee
and the Batyrev coordinate $z$ is given by $z = \frac{1}{\kappa^3}$.
The large radius point, the maximal conifold point, and the orbifold point of the moduli space of local $\IP^2$ correspond to $z=0$, $z=-1/27$, and $z= \infty$, respectively.
The quantization of the mirror curve in Eq.~\eqref{eq: mcP2} gives the quantum operator
\be \label{eq: opP2}
\mO_{\IP^2}(\mx, \my) = \re^{\mx} + \re^{\my} + \re^{-\mx-\my} \, ,
\ee
acting on $L^2(\IR)$, where $\mx, \, \my$ are self-adjoint Heisenberg operators satisfying $[ \mx , \, \my ] = \text{i} \hbar$.
It was proven in~\cite{KM} that the inverse operator 
\be
\rho_{\IP^2} = \mO_{\IP^2}^{-1}
\ee 
is positive-definite and of trace class. The fermionic spectral traces of $\rho_{\IP^2}$ are well-defined and can be computed explicitly~\cite{MZ}.
In this section, we will study the resurgent structure of the first fermionic spectral trace 
\be
Z_{\IP^2}(1, \hbar) = \text{Tr}(\mrho_{\IP^2})
\ee 
in the semiclassical limit $\hbar \rightarrow 0$ and in the dual strongly-coupled limit $\hbar \rightarrow \infty$.

\subsection{Computing the perturbative series} \label{sec: WignerP2}
Let us apply the phase-space formulation of quantum mechanics to obtain the WKB expansion of the trace of the inverse operator $\rho_{\IP^2} $ at next-to-leading order (NLO) in $\hbar \rightarrow 0$, starting from the explicit expression of the operator $\mO_{\IP^2}$ in Eq.~\eqref{eq: opP2}, and following Appendix~\ref{app: Wigner}.
For simplicity, we denote by $O_W, \, \rho_W$ the Wigner transforms of the operators $\mO_{\IP^2}, \, \rho_{\IP^2}$, respectively. 
The Wigner transform of $\mO_{\IP^2}$ is obtained by performing the integration in Eq.~\eqref{eq: OW} directly. As we show in Example~\ref{ex: exampleWigner}, this simply gives the classical function
\be
O_{\rm W} = \re^{x} + \re^{y} + \re^{-x-y} \, .
\ee
Substituting it into Eqs.~\eqref{eq: G2} and~\eqref{eq: G3}, we have
\begin{subequations}
\begin{align}
\CG_2 &= -\frac{\hbar^2}{4} \left[ \re^{x+y} + \re^{-x} + \re^{-y} \right] + \CO(\hbar^4) \, , \\
\CG_3 &= -\frac{\hbar^2}{4} \re^{-2(x+y)} \left[ -3 \re^{2x+2y} + \re^{x+3y} + \re^{3x+4y} + \re^x + \left( x \leftrightarrow y \right) \right] + \CO(\hbar^4) \, , 
\end{align}
\end{subequations}
where $\left( x \leftrightarrow y \right)$ indicates the symmetric expression after exchanging the variables $x$ and $y$.
It follows from Eq.~\eqref{eq: rhoW2} that the Wigner transform of $\rho_{\IP^2}$, up to order $\hbar^2$, is then given by
\be \label{eq: rhoW2_P2}
\rho_{\rm W} = \frac{1}{O_{\rm W}} - \frac{9 \hbar^2}{4} \frac{1}{O_{\rm W}^4} + \CO(\hbar^4) \, .
\ee
We note that the same result can be obtained by solving Eq.~\eqref{eq: cosL2} order by order in powers of $\hbar^2$.
Integrating Eq.~\eqref{eq: rhoW2_P2} over phase space, as in Eq.~\eqref{eq: tracePS}, we obtain the NLO perturbative expansion in $\hbar$ of the trace, that is, 
\be
\text{Tr}(\rho_{\IP^2}) = \frac{1}{2 \pi \hbar} \int_{\IR^2} \rho_{\rm W} \, \rd x \rd y =  \frac{1}{2 \pi \hbar} \int_{\IR^2} \frac{1}{O_{\rm W}} \, \rd x \rd y - \frac{9 \hbar}{8 \pi} \int_{\IR^2} \frac{1}{O_{\rm W}^4} \, \rd x \rd y + \CO(\hbar^4)  \, ,
\ee
and evaluating the integrals explicitly, we find
\be \label{eq: WWfinalP2}
\text{Tr}(\rho_{\IP^2}) = \frac{\Gamma\left(\frac{1}{3}\right)^3}{6 \pi \hbar} \left\{ 1 - \frac{\hbar^2}{72} + \CO(\hbar^4) \right\} \, ,
\ee
where $\Gamma(z)$ denotes the Gamma function.
We stress that the phase-space formalism adopted above provides, in principle, the perturbative expansion of $\text{Tr}(\rho_{\IP^2})$ at all orders in $\hbar$ by systematically extending all intermediate computations beyond order $\hbar^2$. It is not, however, the most practical path. 

The integral kernel for the operator $\rho_{\IP^2}$ is given by~\cite{MZ}
\be \label{eq: P2kernel}
\rho_{\IP^2}(x_1, \, x_2) = \frac{\re^{\pi \mb (x_1+x_2) /3}}{2 \mb \cosh(\pi (x_1-x_2)/\mb + \ri \pi /6)} \frac{\Phi_{\mb} (x_2 + \ri \mb /3) }{\Phi_{\mb} (x_1 - \ri \mb/3)} \, ,
\ee
where $\mb$ is related to $\hbar$ by 
\be 
2 \pi \mb^2 = 3 \hbar \, , 
\ee
and $\Phi_{\mb}$ denotes Faddeev's quantum dilogarithm. Note that the integral kernel in Eq.~\eqref{eq: P2kernel} is well-defined for $\hbar \in \IC'$, since $\Phi_{\mb}$ can be analytically continued to all values of $\mb$ such that $\mb^2 \notin \IR_{\le 0}$. 
A summary of the properties of this function is provided in Appendix~\ref{app: Faddeev}. In what follows, we will assume that $\Re (\mb) >0$.
The first spectral trace has the integral representation~\cite{MZ}
\be \label{eq: P2int}
\text{Tr}(\rho_{\IP^2}) = \frac{1}{\sqrt{3} \mb} \int_{\IR} \re^{2 \pi \mb x /3} \frac{\Phi_{\mb} (x + \ri \mb/3)}{\Phi_{\mb} (x - \ri \mb/3)} \, \rd x \, ,
\ee
which is an analytic function of $\hbar \in \IC'$. As pointed out in~\cite{GuM}, the analytic continuation of the first spectral trace obtained in this way matches the natural analytic continuation of the total grand potential of local $\IP^2$ to complex values of $\hbar$ such that $\Re (\hbar) >0$. The TS/ST correspondence is, then, still applicable.
The integral in Eq.~\eqref{eq: P2int} can be evaluated by using the integral Ramanujan formula, or by analytically continuing $x$ to the complex domain, completing the integration contour from above, and summing over residues, yielding the closed formula~\cite{KM}
\be \label{eq: P2close}
\text{Tr}(\rho_{\IP^2}) = \frac{1}{\sqrt{3} \mb} \re^{- \frac{\pi \ri}{36} (12 c_{\mb}^2+4 \mb^2-3)} \frac{\Phi_{\mb} \left( c_{\mb} - \frac{\ri \mb}{3} \right)^2}{\Phi_{\mb} \left( c_{\mb} - \frac{2 \ri \mb}{3} \right)} = \frac{1}{\mb} \left| \Phi_{\mb} \left( c_{\mb} - \frac{\ri \mb}{3} \right)\right|^3 \, ,
\ee
where $c_{\mb} = \ri (\mb + \mb^{-1})/2$. 
Moreover, the expression in Eq.~\eqref{eq: P2close} can be factorized into a product of $q$- and $\tilde{q}$-series by applying the infinite product representation in Eq.~\eqref{eq: seriesPhib}. 
Namely, we have that 
\be
\Phi_{\mb} \left( c_{\mb} - \frac{\ri \mb}{3} \right) =\frac{(q^{2/3} ; \, q)_{\infty}}{(w^{-1} ; \, \tilde{q})_{\infty}}  \, , \quad \Phi_{\mb} \left( c_{\mb} - \frac{2 \ri \mb}{3} \right) = \frac{(q^{1/3} ; \, q)_{\infty}}{(w ; \, \tilde{q})_{\infty}} \, ,
\ee
where $(x q^{\alpha}; \, q)_{\infty}$ is the quantum dilogarithm defined in Eq.~\eqref{eq: dilog}, and therefore
\be \label{eq: P2fact}
\text{Tr}(\rho_{\IP^2}) = \frac{1}{\sqrt{3} \mb} \re^{- \frac{\pi \ri}{36} \mb^2 + \frac{\pi \ri}{12} \mb^{-2} + \frac{\pi \ri}{4}}  \frac{(q^{2/3} ; \, q)_{\infty}^2}{(q^{1/3} ; \, q)_{\infty}} \frac{(w ; \, \tilde{q})_{\infty}}{(w^{-1} ; \, \tilde{q})_{\infty}^2} \, ,
\ee
where $q = \re^{2 \pi \text{i} \mb^2}$, $\tilde{q} = \re^{- 2 \pi \text{i} \mb^{-2}}$, and $w = \re^{2 \pi \ri /3}$. Note that the factorization in Eq.~\eqref{eq: P2fact} is not symmetric in $q, \, \tilde{q}$.
We assume that $\Im(\mb^2) > 0$, which implies $|q|, |\tilde{q}| < 1$, so that the $q$- and $\tilde{q}$-series converge.

Let us consider the formula in Eq.~\eqref{eq: P2fact} and derive its all-orders perturbative expansion in the limit $\hbar \rightarrow 0$. 
The anti-holomorphic blocks contribute the constant factor
\be \label{eq: antiblockP2}
\frac{(w ; \, \tilde{q})_{\infty}}{(w^{-1} ; \, \tilde{q})_{\infty}^2} \sim \frac{1- w}{(1-w^{-1})^2} = \frac{- \ri}{\sqrt{3}} \, .
\ee
Applying the known asymptotic expansion formula for the quantum dilogarithm in Eq.~\eqref{eq: logPhiK}, with the choice of $\alpha=1/3, 2/3$, and recalling the identities~\cite{Gamma}
\be
\Gamma(2/3) = \frac{2 \pi}{\sqrt{3}} \Gamma(1/3)^{-1} \, , \quad B_{2n+1}= 0, \, \quad B_n(1/3) = (-1)^n B_n(2/3)  \, ,
\ee
where $n \in \IN$, we have that
\be \label{eq: blockP2}
\ba
2 \log (q^{2/3} ; \, q)_{\infty} - \log (q^{1/3} ; \, q)_{\infty} = &- \frac{\pi \ri}{12} \mb^{-2} - \frac{1}{2} \log(- 2 \pi \ri \mb^2) + \log \left( 3 \frac{ \Gamma(1/3)^3}{(2 \pi)^{3/2}} \right) \\
&+ \frac{\pi \ri}{36} \mb^2 - 3 \sum_{n=1}^{\infty} (2 \pi \ri \mb^2)^{2n} \frac{B_{2n} B_{2n+1}(2/3)}{2n (2n+1)!} \, ,
\ea
\ee
where $B_n(z)$ is the $n$-th Bernoulli polynomial, $B_n = B_n(0)$ is the $n$-th Bernoulli number, and $\Gamma(z)$ is the gamma function.
We note that the terms of order $\mb^2$ and $\mb^{-2}$ cancel with the opposite contributions from the exponential in Eq.~\eqref{eq: P2fact}, so that there is no global exponential pre-factor. However, the logarithmic term in $\mb^2$ gives a global pre-factor of the form $1/\mb^2$ after the exponential expansion.
Substituting Eqs.~\eqref{eq: antiblockP2} and~\eqref{eq: blockP2} into Eq.~\eqref{eq: P2fact}, and using $2 \pi \mb^2 = 3 \hbar$, we obtain the all-orders semiclassical expansion of the first spectral trace of local $\IP^2$ in the form\footnote{The formula in Eq.~\eqref{eq: expP2} has also been obtained in~\cite{unpublished_Hatsuda}.}
\be \label{eq: expP2}
\text{Tr}(\rho_{\IP^2}) = \frac{\Gamma\left(\frac{1}{3}\right)^3}{6 \pi \hbar} \exp \left( 3 \sum_{n = 1}^{\infty} (-1)^{n-1} \frac{B_{2n} B_{2n+1}(2/3)}{2n (2n+1)!} (3 \hbar)^{2n} \right) \, ,
\ee
which has coefficients in $\IQ$ of alternating sign up to the global pre-factor. We comment that $\omega_2 = \Gamma(1/3)^3/4 \pi$ is the real half-period of the Weierstrass elliptic function in the equianharmonic case, which corresponds to the elliptic invariants $g_2 = 0$ and $g_3 = 1$, while the other half-period is $\omega_1 = \re^{\pi \ri/3} \omega_2$~\cite{Constants}.
The formula in Eq.~\eqref{eq: expP2} allows us to compute the coefficients of the perturbative series for $Z_{\mathbb{P}^2}(1, \hbar \rightarrow 0)$ at arbitrarily high order. The first few terms are
\be \label{eq: first-terms-P2-zero}
1 - \frac{\hbar^2}{72} + \frac{23  \hbar^4}{51840} - \frac{491 \hbar^6}{11197440} + \frac{1119703 \hbar^8}{112870195200} -\frac{71569373 \hbar^{10}}{17878638919680} + \text{O}(\hbar^{12}) \, ,
\ee
multiplied by the global pre-factor in Eq.~\eqref{eq: expP2}, which confirms our analytic calculation at NLO in Eq.~\eqref{eq: WWfinalP2}.

\subsection{Exact solution to the resurgent structure for \texorpdfstring{$\hbar \rightarrow 0$}{hbar going to zero}} \label{sec: P2zero}

\subsubsection{Resumming the Borel transform} \label{sec: borel-zero}
Let us denote by $\phi(\hbar)$ the formal power series appearing in the exponent in Eq.~\eqref{eq: expP2}. Namely, 
\be \label{eq: phiP2}
\phi(\hbar) = \sum_{n=1}^{\infty} a_{2n} \hbar^{2n} \in \IQ[\![\hbar]\!] \, , \quad a_{2n} = (-1)^{n-1} \frac{B_{2n} B_{2n+1}(2/3)}{2n (2n+1)!} 3^{2n+1}  \quad n \ge 1 \, ,
\ee
which is simply related to the perturbative expansion in the limit $\hbar \rightarrow 0$ of the logarithm of the first spectral trace of local $\IP^2$ by 
\be
\log \text{Tr}(\rho_{\IP^2}) = \phi(\hbar) + 3 \log \Gamma(1/3) - \log (6 \pi \hbar) \, .
\ee
We recall that the Bernoulli polynomials have the asymptotic behavior~\cite{Dilcher}
\be \label{eq: asymBernoulli}
B_{2n}(z) \sim 2 (-1)^{n-1} \cos(2 \pi z) \frac{(2n)!}{(2 \pi)^{2n}} \, , \quad  B_{2n+1}(z) \sim 2 (-1)^{n-1} \sin(2 \pi z) \frac{(2n+1)!}{(2 \pi)^{2n+1}} \, , 
\ee
for $n \gg 1$. It follows that the coefficients of $\phi(\hbar)$ satisfy the factorial growth
\be
a_{2n} \sim (-1)^n (2n)! \left(\frac{4 \pi^2}{3} \right)^{-2n}  \quad n \gg 1 \, ,
\ee
and $\phi(\hbar)$ is a Gevrey-1 asymptotic series. 
Its Borel transform is given by
\be \label{eq: hatphizeta1}
\hat{\phi}(\zeta) = 3 \sum_{n=1}^{\infty} (-1)^{n-1} \frac{B_{2n} B_{2n+1}(2/3)}{2n (2n)! (2n+1)!} (3\zeta)^{2n} \in \IQ[\![\zeta]\!] \, ,
\ee
and it is the germ of an analytic function in the complex $\zeta$-plane.  
\begin{proposition}
Using the definition in Eq.~\eqref{eq: prodHa}, we can interpret the Borel transform $\hat{\phi}(\zeta)$ in Eq.~\eqref{eq: hatphizeta1} as the Hadamard product
\be \label{eq: P2Ha}
\hat{\phi}(\zeta) = (f \diamond g)(\zeta) \, ,
\ee
where the formal power series $f(\zeta)$ and $g(\zeta)$ have finite radius of convergence at the origin $\zeta = 0$, and they can be resummed explicitly as\footnote{We impose that $f(0)=g(0)=0$ in order to eliminate the removable singularities of $f(\zeta), g(\zeta)$ at the origin.}
\begin{subequations}
\begin{align}
f(\zeta) &= \sum_{n=1}^{\infty} \frac{B_{2n+1}(2/3)}{(2n+1)!} \zeta^{2n} =  \frac{1}{2+4 \cosh (\zeta/3)} -\frac{1}{6} \, , \quad |\zeta| < 2\pi \, , \label{eq: fP2} \\
g(\zeta) &= 3 \sum_{n=1}^{\infty} (-1)^{n-1} \frac{B_{2n}}{2n (2n)!} (3 \zeta)^{2n} = -3 \log \left(\frac{2}{3 \zeta} \sin \left( \frac{3 \zeta}{2} \right) \right) \, , \quad |\zeta| < 2\pi/3 \label{eq: gP2} \, .
\end{align}
\end{subequations}
\end{proposition}
\begin{proof}
The Bernoulli polynomials with argument $2/3$ are defined by the generating function
\be \label{eq: fP2-origin}
\sum_{n=0}^{\infty} \frac{B_{n}(2/3)}{n!} \zeta^{n} = \frac{\zeta \re^{2 \zeta/3}}{\re^{\zeta}-1} \, , \quad |\zeta| < 2\pi \, .
\ee 
We apply the hyperbolic identities
\be
\frac{\re^{\zeta}-1}{2 \re^{\zeta/2}} = \sinh(\zeta/2) \, , \quad \re^{\zeta/6} = \sinh (\zeta/6) + \cosh (\zeta/6) \, ,
\ee
and we take the odd part of both sides of Eq.~\eqref{eq: fP2-origin}. We obtain in this way that
\be \label{eq: fP2-origin2}
\sum_{n=0}^{\infty} \frac{B_{2n+1}(2/3)}{(2n+1)!} \zeta^{2n+1} = \frac{\zeta}{2} \frac{\sinh(\zeta/6)}{\sinh(\zeta/2)} \, , \quad |\zeta| < 2\pi \, .
\ee
Using the sum-of-arguments and the half-argument identities for $\sinh(\zeta/3 + \zeta/6)$ and $\coth(\zeta/6)$, respectively, the formula in Eq.~\eqref{eq: fP2-origin2} yields the statement in Eq.~\eqref{eq: fP2}.
Let us consider Eq.~\eqref{eq: zetaseriesEven} for $a=1$ and apply the identity in Eq.~\eqref{eq: ztoBeven} for $\zeta(2n, 1)= \zeta(2n)$, $n \ge 1$. We find that
\be
\sum_{n=1}^{\infty} (-1)^{n-1} \frac{B_{2n}}{2n (2n)!} (3 \zeta)^{2n} = \log \left( \Gamma\left( 1 + \frac{3\zeta}{2\pi} \right) \Gamma\left( 1 - \frac{3\zeta}{2\pi}  \right) \right) \, , \quad |\zeta| < 2 \pi/3 \, ,
\ee
and the statement in Eq.~\eqref{eq: gP2} then follows from Euler's reflection formula for the gamma function, that is, 
\be \label{eq: reflectionGamma}
\Gamma\left( 1 + x \right) \Gamma\left( 1 - x \right) = \frac{\pi x}{\sin \left( \pi x \right) } \, , \quad x \in \IC \backslash \IZ \, ,
\ee
with the choice $x = 3 \zeta / 2 \pi$.
\end{proof}
After being analytically continued to the whole complex plane, the function $f(\zeta)$ has poles of order one along the imaginary axis at
\be
\mu^{\pm}_k = 2 \pi \ri (\pm 1 + 3 k) \, , \quad k \in \IZ \, ,
\ee
while the function $g(\zeta)$ has logarithmic branch points along the real axis at
\be
\nu_m = \frac{2 \pi}{3} m \, , \quad m \in \IZ_{\ne 0} \, .
\ee
We illustrate the singularities of $f(\zeta), g(\zeta)$ in the complex $\zeta$-plane in Fig.~\ref{fig: plotHadamard} on the left. 
\begin{figure}[htb!]
\center
 \includegraphics[width=0.8\textwidth]{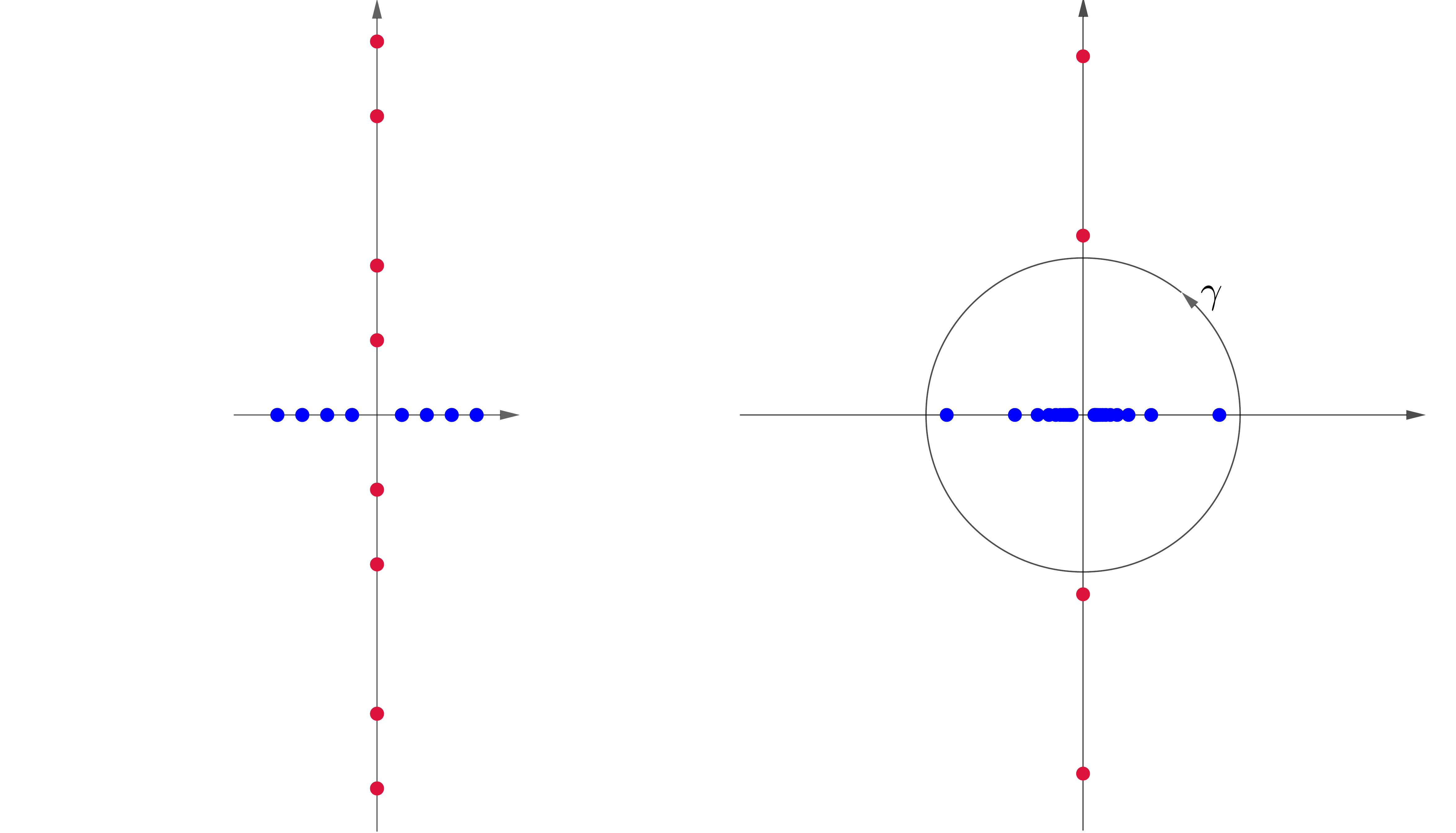}
 \caption{On the left, the first few singularities of $f(\zeta)$ (in red) and $g(\zeta)$ (in blue) in the complex $\zeta$-plane. On the right, the contour $\gamma$ and the first few singularities of $f(s)$ (in red) and $g(\zeta/s)$ (in blue) in the complex $s$-plane with reference values $r=5.5$ and $\zeta=10$.}
 \label{fig: plotHadamard}
\end{figure}

\begin{proposition}
The Borel transform $\hat{\phi}(\zeta)$ in Eq.~\eqref{eq: hatphizeta1} can be expressed as
\be \label{eq: intBorel2}
\hat{\phi}(\zeta) = - \frac{3 \sqrt{3}}{2 \pi} \sum_{k \in \IZ} \frac{1}{1+3k} \log \left( \frac{4 \pi \ri (1+3k)}{3 \zeta} \sin \left(\frac{3 \zeta}{4 \pi \ri (1+3k)} \right) \right) \, ,
\ee
which is a well-defined, exact function of $\zeta$.
\end{proposition}
\begin{proof}
We will now apply Hadamard's multiplication theorem~\cite{Hadamard, Hadamard2}. We refer to Appendix~\ref{app: Hadamard} for a short introduction. 
Let $\gamma$ be a circle in the complex $s$-plane centered at the origin $s=0$ with radius $0 < r < 2 \pi$. As a consequence of the Hadamard decomposition in Eq.~\eqref{eq: P2Ha}, the Borel transform can be written as the integral
\be \label{eq: intBorel}
\hat{\phi}(\zeta) = \frac{1}{2 \pi \ri} \int_{\gamma} f(s) g(\zeta/s) \frac{ \rd s}{s} = -\frac{3}{4 \pi \ri} \int_{\gamma} \left( \frac{1}{1+2 \cosh (s/3)} -\frac{1}{3} \right) \log \left(\frac{2 s}{3 \zeta} \sin \left( \frac{3 \zeta}{2 s} \right) \right) \frac{ \rd s}{s} \, ,
\ee
for $|\zeta| < 2 \pi r /3$. 
We note that, for such values of $\zeta$, the function $s \mapsto g(\zeta/s)$ has logarithmic branch points at $s = \zeta/\nu_m$, $m \in \IZ_{\ne 0}$, which sit inside the contour of integration $\gamma$ and accumulate at the origin, and no singularities for $|s| > r$. The function $f(s)$ has simple poles at the points $s= \mu^{\pm}_k$ with residues
\be
\underset{s=2 \pi \ri (\pm 1 + 3 k)}{\text{Res}} f(s) = \mp \frac{\sqrt{3} \ri}{2} \, , \quad k \in \IZ \, .
\ee
We illustrate the singularities of $f(s), g(\zeta/s)$ in the complex $s$-plane in Fig.~\ref{fig: plotHadamard} on the right. 
By Cauchy's residue theorem, the integral in Eq.~\eqref{eq: intBorel} can be evaluated by summing the residues at the poles of the integrand which lie outside $\gamma$, allowing us to express the Borel transform as an exact function of $\zeta$. More precisely, we find the desired analytic formula
\be \label{eq: intBorel2-proof}
\ba
\hat{\phi}(\zeta) &= - \sum_{k \in \IZ} \underset{s=2 \pi \ri (1 + 3 k)}{\text{Res}} f(s) g(\zeta/s) \frac{1}{s} - \sum_{k \in \IZ} \underset{s=2 \pi \ri (- 1 + 3 k)}{\text{Res}} f(s) g(\zeta/s) \frac{1}{s} = \\
&= - \frac{3 \sqrt{3}}{2 \pi} \sum_{k \in \IZ} \frac{1}{1+3k} \log \left( \frac{4 \pi \ri (1+3k)}{3 \zeta} \sin \left(\frac{3 \zeta}{4 \pi \ri (1+3k)} \right) \right) \, .
\ea
\ee
The convergence of the infinite sum in the second line of Eq.~\eqref{eq: intBorel2-proof} can be easily verified by, \textit{e.g.}, the limit comparison test. 
\end{proof}
\begin{corollary}
The singularities of the Borel transform $\hat{\phi}(\zeta)$ in Eq.~\eqref{eq: intBorel2} are logarithmic branch points located along the imaginary axis at
\be \label{eq: zetakm}
\zeta_{k,m} = \mu^{+}_k \nu_m = \frac{4 \pi^2 \ri}{3} (1+3k) m \, , \quad k \in \IZ \, , \quad m \in \IZ_{\ne 0} \, ,
\ee
which we write equivalently as
\be \label{eq: zetan}
\zeta_n = \frac{4 \pi^2 \ri}{3} n \, , \quad n \in \IZ_{\ne 0} \, ,
\ee
that is, the branch points lie at all non-zero integer multiples of the two complex conjugate dominant singularities at $\pm 4 \pi^2 \ri/3$, as illustrated in Fig.~\ref{fig: peacockP2zero}. 
\end{corollary}
This is the simplest occurrence of the peacock pattern of singularities described in Section~\ref{sec: resurgent_strings}. There are two Stokes lines at the angles $\pm \pi/2$.
\begin{figure}[htb!]
\center
 \includegraphics[width=0.2\textwidth]{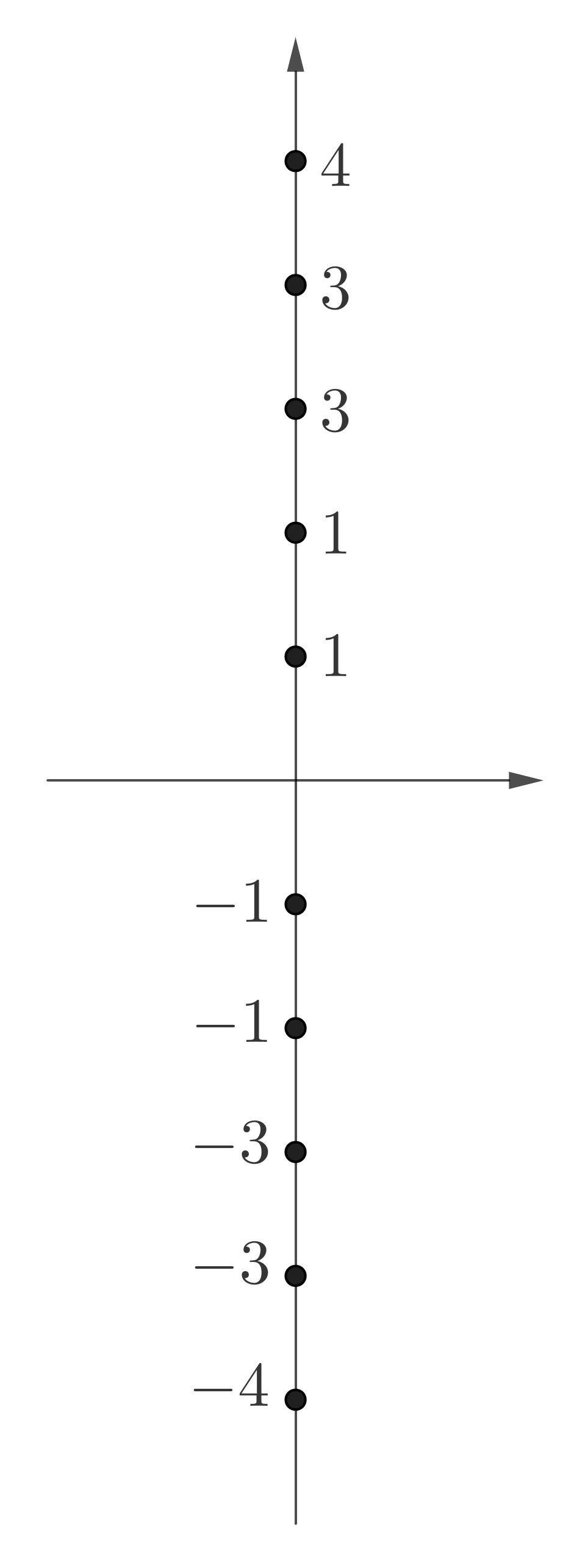}
 \caption{The first few singularities of the Borel transform of the asymptotic series $\phi(\hbar)$, defined in Eq.~\eqref{eq: phiP2}, and the associated integer constants $\alpha_n \in \IZ_{\ne 0}$, defined in Eq.~\eqref{eq: intStokesP2zero}.}
 \label{fig: peacockP2zero}
\end{figure}
Note that the analytic expression for the Borel transform in Eq.~\eqref{eq: intBorel2} is explicitly simple resurgent, and thus we expect its local singular behavior to be of the form in Eq.~\eqref{eq: Stokes0}.
\begin{corollary} \label{cor: localexp}
The local expansion of the Borel transform $\hat{\phi}(\zeta)$ in Eq.~\eqref{eq: intBorel2} at $\zeta = \zeta_n$, $n \in \IZ_{\ne 0}$, is given by
\be \label{eq: explocalP2}
\hat{\phi}(\zeta) =  - \frac{S_n}{2 \pi \text{i}} \log(\zeta - \zeta_n) + \dots \, ,
\ee
where $S_n \in \IC$ is the Stokes constant.
\end{corollary}
\begin{proof}
The local expansion of $\hat{\phi}(\zeta)$ around the logarithmic singularity $\zeta = \zeta_n$ is obtained by summing the contributions from all pairs $(k,m) \in \IZ \times \IZ_{\ne 0}$ such that $n = (1+3k)m$. There is finitely many such pairs of integers, and we collect them into a set $I_n$ for each $n \in \IZ_{\ne 0}$.
For a fixed value of $k \in \IZ$, we denote the corresponding term in the sum in Eq.~\eqref{eq: intBorel2} by 
\be \label{eq: fk}
f_k(\zeta) = - \frac{3 \sqrt{3}}{2 \pi (1+3k)} \log \left( \frac{4 \pi \ri (1+3k)}{3 \zeta} \sin \left(\frac{3 \zeta}{4 \pi \ri (1+3k)} \right) \right) \, .
\ee 
We expand it locally around $\zeta= \zeta_{k,m}$ for every choice of $m \in \IZ_{\ne 0}$, and we obtain
\be
f_k(\zeta) = -\frac{s_{k,m}}{2 \pi \text{i}} \log(\zeta - \zeta_{k,m}) + \dots \, ,
\ee
where the dots denote regular terms in $\zeta - \zeta_{k,m}$, and $s_{k,m}$ is a complex number.
Since $\zeta_n = \zeta_{k,m}$ for all $(k,m) \in I_n$, it follows that the local expansion of $\hat{\phi}(\zeta)$ at $\zeta = \zeta_n$ is given by
\be \label{eq: explocalP2-proof}
\hat{\phi}(\zeta) = \sum_{k \in \IZ} f_k(\zeta) = - \frac{S_n}{2 \pi \text{i}} \log(\zeta - \zeta_n) + \dots \, ,
\ee
where the Stokes constant $S_n$ is the finite sum 
\be \label{eq: stokesSum}
S_n = \sum_{(k,m) \in I_n} s_{k,m} \, .
\ee
\end{proof}
It follows from Corollary~\ref{cor: localexp} that the locally analytic function that resurges at $\zeta=\zeta_n$ is trivially $\hat{\phi}_n(\zeta-\zeta_n) = 1$, $n \in \IZ_{\ne 0}$. 
We observe that the Laplace transform in Eq.~\eqref{eq: Laplace} acts trivially on constants, and thus we also have that $\phi_n(\hbar) = 1$, $n \in \IZ_{\ne 0}$, that is, there are no perturbative contributions coming from the higher-order instanton sectors.
Moreover, the procedure above allows us to derive analytically all the Stokes constants. 
After being suitably normalized, the Stokes constants $S_n$ are rational numbers, and they are simply related to an interesting sequence of integers $\alpha_n$,  $n \in \IZ_{\ne 0}$. In particular, we find that 
\begin{subequations} \label{eq: intStokesP2zero}
\be
 S_1 = 3 \sqrt{3} \ri \, , \quad S_n = S_1 \frac{\alpha_n}{n} \quad n \in \IZ_{\ne 0,1} \, , 
\ee
\be
\alpha_n = -\alpha_{-n} \, , \quad \alpha_n \in \IZ_{> 0}  \quad n \in \IZ_{> 0} \, .
\ee
\end{subequations}
Explicitly, the first several integer constants $\alpha_n$, $n>0$, are
\be
1, 1, 3, 3, 4, 3, 8, 5, 9, 4, 10, 9, 14, 8, 12, 11, 16, 9, 20, 12, \dots \, .
\ee
The pattern of singularities in the Borel plane and the associated $\alpha_n \in \IZ_{\ne 0}$ are shown in Fig.~\ref{fig: peacockP2zero}. 

\subsubsection{Closed formulae for the Stokes constants} \label{sec: closed-zero}
We will now present and prove a series of exact arithmetic formulae for the Stokes constants $S_n$ of the asymptotic series $\phi(\hbar)$, defined in Eq.~\eqref{eq: phiP2}, and the related integer constants $\alpha_n$, defined in Eq.~\eqref{eq: intStokesP2zero}, for $n \in \IZ_{\ne 0}$. Let us start by showing that both sequences $S_n$ and $\alpha_n$ define explicit divisor sum functions.
\begin{proposition} \label{prop: formulaS1} 
The normalized Stokes constant $S_n/S_1$, where $S_1 = 3 \sqrt{3} \ri$, is determined by the positive integer divisors of $n \in \IZ_{\ne 0}$ according to the closed formula
\be \label{eq: formulaS1}
\frac{S_n}{S_1} = \sum_{\substack{d | n \\ d \equiv_3 1}} \frac{1}{d} - \sum_{\substack{d | n \\ d \equiv_3 2}} \frac{1}{d} \, , 
\ee
which implies that $S_n = S_{-n}$ and $S_n/S_1 \in \IQ_{>0}$.
\end{proposition}
\begin{proof}
Let us denote by $D_n$ the set of positive integer divisors of $n$. We recall that $n$ satisfies the factorization property $n = (1+3k)m$ for $k \in \IZ$ and $m \in \IZ_{\ne 0}$. It follows that either 
\be \label{eq: pairsPlus}
m = \frac{n}{d} \, , \quad k = \frac{d-1}{3} \, ,
\ee
where $d \in D_n$ such that $d-1$ is divisible by $3$, or 
\be \label{eq: pairsMinus}
m = -\frac{n}{d} \, , \quad k = -\frac{d+1}{3} \, ,
\ee
where $d \in D_n$ such that $d+1$ is divisible by $3$. 
In the first case of $d \equiv_3 1$, substituting the values of $k,m$ from Eq.~\eqref{eq: pairsPlus} into Eqs.~\eqref{eq: fk} and~\eqref{eq: zetakm}, we find that the contribution to the Stokes constant $S_n$ coming from the local expansion of $f_k(\zeta)$ around $\zeta_{k,m}$ is simply $s_{k,m} = 3\sqrt{3} \ri /d$. 
In the second case of $d \equiv_3 2$, substituting the values of $k,m$ from Eq.~\eqref{eq: pairsMinus} into Eqs.~\eqref{eq: fk} and~\eqref{eq: zetakm}, we find that the contribution to the Stokes constant $S_n$ coming from the local expansion of $f_k(\zeta)$ around $\zeta_{k,m}$ is simply $s_{k,m} = -3 \sqrt{3} \ri/d$.
Finally, for any divisor $d \in D_n$ which is a multiple of $3$, neither $d-1$ or $d+1$ are divisible by $3$, which implies that the choice $m = \pm n/d$ is not allowed, and the corresponding contribution is $s_{k,m}=0$.
Putting everything together and using Eq.~\eqref{eq: stokesSum}, we find the desired statement.
\end{proof}
We note that the arithmetic formula for the Stokes constants in Eq.~\eqref{eq: formulaS1} can be written equivalently as a closed expression for the integer constants $\alpha_n$,  $n \in \IZ_{\ne 0}$. Namely,
\be \label{eq: formulaA1}
\alpha_n = \sum_{\substack{d | n \\ \frac{n}{d} \equiv_3 1}} d - \sum_{\substack{d | n \\ \frac{n}{d} \equiv_3 2}} d \, , 
\ee
which implies that $\alpha_n = - \alpha_{-n}$, and $\alpha_n \in \IZ_{>0}$ for all $n>0$. 
Two corollaries then follow straightforwardly from Proposition~\ref{prop: formulaS1}. 
\begin{corollary} \label{prop: formulaA1mult} 
The positive integer constants $\alpha_n$, $n \in \IZ_{>0}$, satisfy the closed formulae
\be
\alpha_{p_1^{e_1}} = \frac{p_1^{e_1+1}-1}{p_1-1} \, , \quad \alpha_{p_2^{e_2}} = \frac{p_2^{e_2+1}+(-1)^{e_2}}{p_2+1} \, , \quad \alpha_{p_3^{e_3}} = p_3^{e_3} \, , 
\ee
where $e_i \in \IN$, and $p_i \in \IP$ are prime numbers such that $p_i \equiv_3 i$ for $i=1,2,3$. Moreover, they obey the multiplicative property
\be \label{eq: formulaA1mult}
\alpha_n = \prod_{p \in \IP} \alpha_{p^e} \, , \quad n = \prod_{p \in \IP} p^e \, , \quad e \in \IN \, .
\ee
\end{corollary}
\begin{proof}
The three closed formulae follow directly from Eq.~\eqref{eq: formulaA1}. Explicitly, let $n=p^e$ with $p \in \IP$ and $e \in \IN$. We have that
\begin{subequations}
\begin{align}
\sum_{\substack{d | n \\ \frac{n}{d} \equiv_3 1}} d = p^e \, , \quad  \sum_{\substack{d | n \\ \frac{n}{d} \equiv_3 2}} d = 0 \, , \quad &\mbox{if} \quad p \equiv 0 \mod 3 \, , \\
\sum_{\substack{d | n \\ \frac{n}{d} \equiv_3 1}} d = \sum_{i=0}^e p^i = \frac{p^{e+1}-1}{p-1} \, , \quad  \sum_{\substack{d | n \\ \frac{n}{d} \equiv_3 2}} d = 0 \, , \quad &\mbox{if} \quad p \equiv 1 \mod 3 \, , \\
\sum_{\substack{d | n \\ \frac{n}{d} \equiv_3 1}} d = \sum_{i=0}^{\lfloor e/2 \rfloor} p^{e-2i} \, , \quad  \sum_{\substack{d | n \\ \frac{n}{d} \equiv_3 2}} d = \sum_{i=0}^{\lfloor e/2 \rfloor} p^{e-(2i+1)} \, , \quad &\mbox{if} \quad p \equiv 2 \mod 3 \, .
\end{align}
\end{subequations}
Let us now prove the multiplicity property. We will prove a slightly stronger statement. We write $n = p q$ for $p,q \in \IZ_{> 0}$ coprimes. We choose a positive integer divisor $d | n$, and we write $d = s t$ where $s | p$ and $t | q$. Consider two cases:
\begin{itemize}
\item[(1)] Suppose that $n/d \equiv_3 1$. Then, either $p/s \equiv_3 q/t \equiv_3 1$, or $p/s \equiv_3 q/t \equiv_3 2$, and therefore
\be \label{eq: fact1formulaA1}
\sum_{\substack{d | n \\ \frac{n}{d} \equiv_3 1}} d = \sum_{\substack{s | p \\ \frac{p}{s} \equiv_3 1}} s  \sum_{\substack{t | q \\ \frac{q}{t} \equiv_3 1}} t + \sum_{\substack{s | p \\ \frac{p}{s} \equiv_3 2}} s  \sum_{\substack{t | q \\ \frac{q}{t} \equiv_3 2}} t \, .
\ee
\item[(2)] Suppose that $n/d \equiv_3 2$. Then, either $p/s \equiv_3 1$ and $q/t \equiv_3 2$, or $p/s \equiv_3 2$ and $q/t \equiv_3 1$, and therefore
\be \label{eq: fact2formulaA1}
\sum_{\substack{d | n \\ \frac{n}{d} \equiv_3 2}} d = \sum_{\substack{s | p \\ \frac{p}{s} \equiv_3 1}} s  \sum_{\substack{t | q \\ \frac{q}{t} \equiv_3 2}} t + \sum_{\substack{s | p \\ \frac{p}{s} \equiv_3 2}} s  \sum_{\substack{t | q \\ \frac{q}{t} \equiv_3 1}} t \, .
\ee
\end{itemize}
Substituting Eqs.~\eqref{eq: fact1formulaA1} and~\eqref{eq: fact2formulaA1} into Eq.~\eqref{eq: formulaA1}, we find that
\be
\alpha_n = \left( \sum_{\substack{s | p \\ \frac{p}{s} \equiv_3 1}} s - \sum_{\substack{s | p \\ \frac{p}{s} \equiv_3 2}} s \right) \left( \sum_{\substack{t | q \\ \frac{q}{t} \equiv_3 1}} t - \sum_{\substack{t | q \\ \frac{q}{t} \equiv_3 2}} t \right) = \alpha_p \alpha_q  \, , 
\ee
which proves that the sequence $\alpha_n$, $n \in \IZ_{>0}$, defines a multiplicative arithmetic function. Note that the proof breaks if $p,q$ are not coprimes, since the formulae above lead in general to overcounting the contributions coming from common factors. Therefore, the sequence $\alpha_n$ is not totally multiplicative. Note that the sequence of normalized Stokes constants $S_n/S_1$, $n \in \IZ_{>0}$, is also a multiplicative arithmetic function. 
\end{proof}
\begin{corollary} \label{prop: formulaA2} 
The positive integer constants $\alpha_n$, $n  \in \IZ_{>0}$, are encoded in the generating function 
\be \label{eq: formulaA2}
\sum_{n =1}^{\infty} \alpha_n x^n = \sum_{m =1}^{\infty} \frac{m x^m}{1+x^m+x^{2m}} \, .
\ee
\end{corollary}
\begin{proof}
We denote by $f(x)$ the generating function in the RHS of Eq.~\eqref{eq: formulaA2}. We note that 
\be
f(x) = f_1(x) - f_2(x) \, ,
\ee
where the functions $f_1(x), f_2(x)$ are defined by
\be
f_1(x) = \sum_{m \in \IN_{\ne 0}} \frac{m x^m}{1-x^{3m}} \, , \quad f_2(x) = \sum_{m \in \IN_{\ne 0}} \frac{m x^{2m}}{1-x^{3m}} \, .
\ee
The formula in Eq.~\eqref{eq: formulaA2} follows from the stronger statement
\be
\sum_{\substack{d | n \\ \frac{n}{d} \equiv_3 1}} d = \frac{1}{n!} \frac{\rd^n f_1(0)}{\rd x^n} \, , \quad  \sum_{\substack{d | n \\ \frac{n}{d} \equiv_3 2}} d = \frac{1}{n!} \frac{\rd^n f_2(0)}{\rd x^n} \, , \quad n  \in \IZ_{>0} \, .
\ee
We will now prove this claim for the function $f_1(x)$. The case of $f_2(x)$ is proven analogously.
Let us denote by 
\be
f_{1,m}(x) = \frac{m x^m}{1-x^{3m}} \, , \quad m \in \IN_{\ne 0} \, , 
\ee
and consider the derivative $\rd^n f_{1,m}(x)/\rd x^n$ for fixed $m$. We want to determine its contributions to $\rd^n f_1(0)/\rd x^n$. Since we are interested in those terms that survive after taking $x=0$, we look for the monomials of order $x^{dm-n}$, where $d|n$, in the numerator of $\rd^n f_{1,m}(x)/\rd x^n$, and we take $m = n/d$. More precisely, deriving $a$-times the factor $m x^m$ and $(n-a)$-times the factor $(1-x^{3m})^{-1}$, we have the term
\be  \label{eq: monomials0}
\binom{n}{a} \frac{\rd^a (m x^m)}{\rd x^a} \frac{\rd^{n-a} (1-x^{3m})^{-1}}{\rd x^{n-a}} \, , \quad a \in \IN_{\ne 0} \, .
\ee
Recall that the generalized binomial theorem for the geometric series yields
\be \label{eq: geomseries}
\frac{\rd^{n-a} (1-x^{3m})^{-1}}{\rd x^{n-a}} = \sum_{k=0}^{\infty} \frac{(3mk)!}{(3mk-n+a)!} x^{3mk-n+a} \, .
\ee
Substituting Eq.~\eqref{eq: geomseries} into Eq.~\eqref{eq: monomials0} and performing the derivation, we have
\be
n! \sum_{k=0}^{\infty} m \binom{m}{m-a} \binom{3mk}{3mk-n+a} x^{(1+3k)m-n} \, . 
\ee
It follows then that the only non-zero term at fixed $m \in \IN_{\ne 0}$ comes from the values of $k \in \IN$ and $a \in \IN_{\ne 0}$ such that $(1+3k)m = n$ and $a=m$, which implies in turn that $m|n$ and $n/m \equiv_3 1$. Finally, summing the non-trivial contributions over $m$ gives precisely
\be
\frac{\rd^n f_1(0)}{\rd x^n} = \sum_{\substack{m|n \\ \frac{n}{m} \equiv_3 1}} n! m \binom{m}{0} \binom{n-m}{0} = n! \sum_{\substack{m|n \\ \frac{n}{m} \equiv_3 1}} m \, .
\ee
\end{proof}
A third notable consequence of Proposition~\ref{prop: formulaS1} is that the Stokes constants $S_n$, $n \in \IZ_{>0}$, can be naturally organized as coefficients of an exact generating function given by quantum dilogarithms.
\begin{corollary} \label{prop: formulaS2} 
The Stokes constants $S_n$, $n  \in \IZ_{>0}$, are encoded in the generating function 
\be \label{eq: formulaS2}
\sum_{n =1}^{\infty} S_n x^n = -\ri \pi - 3 \log \frac{(w ; \, x)_{\infty}}{(w^{-1} ; \, x)_{\infty}} \, , \quad |x|<1 \, ,
\ee
where $w = \re^{2 \pi \ri /3}$.
\end{corollary}
\begin{proof}
We apply the definition of the quantum dilogarithm in Eq.~\eqref{eq: dilog} and Taylor expand the logarithm function for $|x| <1$. We obtain in this way that 
\be
\log (w ; \, x)_{\infty} =  \sum_{m=0}^{\infty} \log(1-w x^m) = -\sum_{m=0}^{\infty} \sum_{k=1}^{\infty} \frac{x^{mk}}{k} w^k \, ,
\ee
and therefore also
\be \label{eq: logfrac}
\log \frac{(w ; \, x)_{\infty}}{(w^{-1} ; \, x)_{\infty}} = -\sum_{m=0}^{\infty} \sum_{k=1}^{\infty} \frac{x^{mk}}{k} (w^k-w^{-k}) \, .
\ee
We observe that 
\be \label{eq: cases}
w^k-w^{-k} =  \re^{2 \pi \ri k /3} -  \re^{- 2 \pi \ri k/3} = 
\begin{cases}
         0 & \mbox{for} \quad k \equiv 0 \mod 3 \\ 
         \ri \sqrt{3}  & \mbox{for} \quad k \equiv 1 \mod 3 \\
         - \ri \sqrt{3} & \mbox{for} \quad k \equiv 2 \mod 3
\end{cases} \, .
\ee
Substituting Eq.~\eqref{eq: cases} into Eq.~\eqref{eq: logfrac}, and performing the change of variable $n=mk$, we find
\be \label{eq: logfrac2}
\log \frac{(w ; \, x)_{\infty}}{(w^{-1} ; \, x)_{\infty}} = - \ri \sqrt{3} \sum_{n =1}^{\infty} x^n  \left( \sum_{\substack{d|n \\ d \equiv_3 1}} \frac{1}{d} - \sum_{\substack{d|n \\ d \equiv_3 2}} \frac{1}{d}  \right) - \sum_{k=1}^{\infty} \frac{1}{k} (w^k-w^{-k}) \, ,
\ee
where the last term in the RHS is simply resummed to
\be
- \sum_{k=1}^{\infty} \frac{1}{k} (w^k-w^{-k}) = \log (1-w) - \log(1-w^{-1}) = \log(-w) = -\frac{\pi \ri}{3} \, .
\ee
Substituting the arithmetic formula for the Stokes constants in Eq.~\eqref{eq: formulaS1} into the expression in Eq.~\eqref{eq: logfrac2}, we obtain the desired statement.
\end{proof}
We note that, choosing $x = \tilde{q} = \re^{-4 \pi^2 \ri /3 \hbar}$, Corollary~\ref{prop: formulaS2} directly implies an exact expression in terms of $\tilde{q}$-series for the discontinuity of the asymptotic series $\phi(\hbar)$ across the positive imaginary axis, which borders the only two distinct Stokes sectors in the upper half of the Borel plane. Namely, following the definition in Eq.~\eqref{eq: Stokes1}, and recalling that the formal power series that resurges from the Borel singularity $\zeta_n = n 4 \pi^2 \ri /3$, $n \in \IZ_{>0}$, is trivially $\phi_n(\hbar)=1$, we have that
\be \label{eq: discP2series1}
\text{disc}_{\pi/2}\phi(\hbar) = s_{+}(\phi)(\hbar) - s_{-}(\phi)(\hbar) = \sum_{n=1}^{\infty} S_n \re^{-n 4 \pi^2 \ri /3 \hbar} \, ,
\ee
where $s_{\pm}(\phi)(\hbar)$ are the lateral Borel resummations at the angles $\pi/2 \pm \epsilon$ with $\epsilon \ll 1$, which lie slightly above and slightly below the Stokes line along the positive imaginary axis, respectively. Substituting Eq.~\eqref{eq: formulaS2} into Eq.~\eqref{eq: discP2series1}, we obtain the exact formula
\be \label{eq: discP2series2}
\text{disc}_{\pi/2}\phi(\hbar) = -\ri \pi - 3 \log (w ; \, \tilde{q})_{\infty} + 3 \log (w^{-1} ; \, \tilde{q})_{\infty} \, .
\ee
We stress that $(w ; \, \tilde{q})_{\infty}$ and $(w^{-1} ; \, \tilde{q})_{\infty}$ are the same $\tilde{q}$-series which appear as the anti-holomophic block of the first spectral trace of local $\IP^2$ in Eq.~\eqref{eq: P2fact}.

\subsubsection{Exact large-order relations} \label{sec: zeta-zero}
We provide here an alternative closed formula for the perturbative coefficients $a_{2n}$, $n \in \IN_{\ne 0}$, of the asymptotic series $\phi(\hbar)$ in Eq.~\eqref{eq: phi2}, which highlights an interesting link to analytic number theory. We recall that the large-$n$ asymptotics of the coefficients $a_{2n}$ is controlled at leading order by the singular behavior of the Borel transform $\hat{\phi}(\zeta)$ in the neighbourhood of its dominant complex conjugate singularities $\zeta_{\pm 1}$, which is encoded in the local expansion in Eq.~\eqref{eq: explocalP2}. We have the standard formula
\be
a_{2n} \sim \frac{(-1)^n}{\pi \ri} \frac{\Gamma(2n)}{A^{2n}} S_1 \, , \quad  n \gg 1 \, ,
\ee
where $A= 4 \pi^2/3$ and $S_1= 3 \sqrt{3} \ri$. By systematically including the contributions from all sub-dominant singularities in the Borel plane, the leading asymptotics can be upgraded to an exact large-order relation, which is\footnote{The formula in Eq.~\eqref{eq: exactlarge} has also been obtained in~\cite{unpublished_Gu}.}
\be \label{eq: exactlarge}
a_{2n} = \frac{(-1)^n}{\pi \ri} \frac{\Gamma(2n)}{A^{2n}} \sum_{m=1}^{\infty} \frac{S_m}{m^{2n}} \, , \quad  n \in \IN_{\ne 0} \, ,
\ee
where the Stokes constant $S_m$ is given explicitly in Eq.~\eqref{eq: intStokesP2zero}.
\begin{proposition} \label{prop: formulaS3} 
The Stokes constants $S_m$, $m \in \IZ_{> 0}$, satisfy the exact relations
\be \label{eq: formulaS3}
\sum_{m=1}^{\infty} \frac{S_m}{m^{2n}} = 3 \sqrt{3} \ri \frac{\zeta(2n)}{3^{2n+1}} \left(\zeta\left(2n+1,\frac{1}{3}\right) - \zeta\left(2n+1, \frac{2}{3} \right) \right) \, , \quad n \in \IN_{\ne 0} \, ,
\ee
where $\zeta(z)$ denotes the Riemann zeta function, and $\zeta(z, a)$ denotes the Hurwitz zeta function.
\end{proposition}
\begin{proof}
Substituting the original expression for the perturbative coefficients $a_{2n}$, $n \in \IN_{\ne 0}$, in Eq.~\eqref{eq: phiP2} into the exact large-order relation in Eq.~\eqref{eq: exactlarge}, we have that
\be \label{eq: proofS3-1}
\sum_{m=1}^{\infty} \frac{S_m}{m^{2n}} = - 3\pi \ri (2 \pi)^{4n} \frac{B_{2n} B_{2n+1}(2/3)}{(2n)! (2n+1)!} \, .
\ee
Using the known identities
\begin{subequations}
\be \label{eq: ztoBeven}
\zeta(2n) = (-1)^{n+1} \frac{(2 \pi)^{2n} B_{2n}}{2 (2n)!} \, ,
\ee
\be
B_{2n+1}(2/3) = - B_{2n+1}(1/3) \, , \quad B_{2n+1}(z) = -(2n+1) \zeta(-2n, z) \, , \label{eq: betaID}
\ee
\end{subequations}
the formula in Eq.~\eqref{eq: proofS3-1} becomes
\be \label{eq: proofS3-2}
\sum_{m=1}^{\infty} \frac{S_m}{m^{2n}} = - 3\pi \ri (-1)^n (2 \pi)^{2n} \frac{\zeta(2n)}{(2n)!} \left(\zeta\left(-2n,\frac{2}{3}\right) - \zeta\left(-2n, \frac{1}{3} \right) \right) \, .
\ee
We recall that the Hurwitz zeta function satisfies the functional equation
\be \label{eq: functZeta}
\zeta\left(1-z, \frac{a}{b}\right) = \frac{2 \Gamma(z)}{(2 \pi b)^z} \sum_{j=1}^b \zeta\left( z, \frac{j}{b} \right) \cos\left(\frac{\pi z}{2} - \frac{2 \pi j a}{b} \right) \, ,
\ee
for integers $1 \le a \le b$, which gives in particular
\be \label{eq: proofS3-zeta}
\zeta\left(-2n,\frac{2}{3}\right) - \zeta\left(-2n, \frac{1}{3}\right) = (-1)^n \frac{2 \sqrt{3} (2n)!}{(6 \pi)^{2n+1}} \left( \zeta\left(2n+1,\frac{2}{3}\right) - \zeta\left(2n+1, \frac{1}{3}\right) \right) \, .
\ee
Substituting Eq.~\eqref{eq: proofS3-zeta} into Eq.~\eqref{eq: proofS3-2}, we obtain the desired statement.
\end{proof}
\begin{remark}
We note that the exact expression in Eq.~\eqref{eq: formulaS3} can be written equivalently in terms of the integer constants $\alpha_m$,  $m \in \IZ_{\ne 0}$. Namely, 
\be \label{eq: formulaA3}
\sum_{m=1}^{\infty} \frac{\alpha_m}{m^{2n+1}} = \frac{\zeta(2n)}{3^{2n+1}} \left(\zeta\left(2n+1,\frac{1}{3}\right) - \zeta\left(2n+1, \frac{2}{3} \right) \right) \, , \quad n \in \IN_{\ne 0} \, ,
\ee
which hints at a fascinating connection to the analytic theory of $L$-functions. More precisely, let us point out that the series in the LHS of Eq.~\eqref{eq: formulaA3} belongs to the family of Dirichlet series~\cite{Dirichlet}.
As a consequence of Corollary~\ref{prop: formulaA1mult}, the sequence of integers $\alpha_m$, $m \in \IZ_{\ne 0}$, defines a bounded multiplicative arithmetic function. Therefore, the corresponding Dirichlet series satisfies an expansion as an Euler product indexed by the set of prime numbers $\IP$, that is, 
\be \label{eq: L-series}
\sum_{m=1}^{\infty} \frac{\alpha_m}{m^{2n+1}} = \prod_{p \in \IP} \sum_{e=0}^{\infty} \frac{\alpha_{p^e}}{p^{e(2n+1)}}  \, , \quad n \in \IN_{\ne 0} \, ,
\ee
which proves that the given Dirichlet series is, indeed, an $L$-series. We will further explore in this direction in Section~\ref{sec: dirichlet}.
\end{remark}

\subsubsection{Exponentiating with alien calculus} \label{sec: exp-zero}
We will now translate our analytic solution to the resurgent structure of the asymptotic series $\phi(\hbar)$ in Eq.~\eqref{eq: phiP2} into results on the original, exponentiated perturbative series in Eq.~\eqref{eq: expP2}, which we denote by
\be \label{eq: psiP2}
\psi(\hbar) = \re^{\phi(\hbar)} = \exp \left( 3 \sum_{n = 1}^{\infty} (-1)^{n-1} \frac{B_{2n} B_{2n+1}(2/3)}{2n (2n+1)!} (3 \hbar)^{2n} \right)  \in \IQ[\![\hbar]\!] \, ,
\ee
and which is also a Gevrey-1 asymptotic series. Its Borel transform $\hat{\psi}(\zeta)$ inherits from $\hat{\phi}(\zeta)$ the same pattern of singularities in Eq.~\eqref{eq: zetan}. Namely, there are infinitely many and discrete logarithmic branch points located along the imaginary axis of the complex $\zeta$-plane at $\zeta_n = n 4 \pi^2 \ri/3$, $n \in \IZ_{\ne 0}$.
We denote by $s_{\pm}(\psi)(\hbar)$ the lateral Borel resummations at the angles $\pi/2 \pm \epsilon$ with $\epsilon \ll 1$, which lie slightly above and slightly below the Stokes line along the positive imaginary axis, respectively. 
Let us apply Eqs.~\eqref{eq: autStokes} and~\eqref{eq: alienStokes} and expand the exponential operator defining the Stokes automorphism. We find that
\be \label{eq: splus1}
\ba
s_{+}(\psi)(\hbar) &= s_{-} \circ \mathfrak{S}_{\pi/2}(\psi)(\hbar) = s_{-} \circ \exp \left( \sum_{n=1}^{\infty} \re^{-\zeta_n/\hbar} \Delta_{\zeta_n} \right)(\psi)(\hbar) = \\ 
&= s_{-}(\psi)(\hbar) + \sum_{r=1}^{\infty} \frac{1}{r!} \sum_{n_1, \dots, n_r =1}^{\infty} \re^{-(\zeta_{n_1} + \dots + \zeta_{n_r})/\hbar} s_{-} \left( \Delta_{\zeta_{n_1}} \cdots  \Delta_{\zeta_{n_r}} \psi \right)(\hbar) = \\
&= s_{-}(\psi)(\hbar) + \sum_{k=1}^{\infty} \re^{-\zeta_k/\hbar} \sum_{p \in \CP(k)} \frac{1}{r!} \binom{r}{N_1, \dots , N_k} s_{-} \left( \Delta_{\zeta_{n_1}} \cdots  \Delta_{\zeta_{n_r}} \psi \right)(\hbar) \, ,
\ea
\ee
where $\Delta_{\zeta_n}$ is the alien derivative associated to the singularity $\zeta_n$, $n \in \IZ_{>0}$, whose definition and basic properties are summarized in Appendix~\ref{app: alien}, and $\CP(k)$ is the set of all partitions of the positive integer $k$. A partition $p \in \CP(k)$ of length $|p|=r \in \IN_{\ne 0}$ has the form $p = (n_1, \dots , n_r)$ with $1 \le n_1 \le \dots \le n_r \le k$ such that $n_1 + \dots + n_r = k$. We denote by $N_i \in \IN$ the number of times that the positive integer $i \in \IZ_{> 0}$ is repeated in the partition $p$. Note that $\sum_{i=1}^k N_i = r$.

Since Eq.~\eqref{eq: exampleA} directly applies to the asymptotic series $\phi(\hbar)$, the action of the alien derivative $\Delta_{\zeta_n}$ on $\phi(\hbar)$ simplifies to give precisely the Stokes constant at the singularity $\zeta_n$, that is,\footnote{Let us stress that the output of alien derivation on a formal power series has, in general, a more complex dependence on the Stokes constants. For more details, see Appendix~\ref{app: alien}.}
\be
\Delta_{\zeta_n} \phi(\hbar) = S_n \, , 
\ee
where $S_n$ is written explicitly in Eq.~\eqref{eq: intStokesP2zero}, while the formula in Eq.~\eqref{eq: alienExp} becomes
\be
\Delta_{\zeta_n} \psi(\hbar) = \Delta_{\zeta_n} \re^{\phi(\hbar)} = S_n  \psi(\hbar) \, ,
\ee
and therefore we have
\be \label{eq: psiDs}
\Delta_{\zeta_{n_1}} \cdots  \Delta_{\zeta_{n_r}} \psi(\hbar) = S_{n_1} \cdots S_{n_r}  \psi(\hbar) \, , \quad r \in \IN_{\ne 0} \, .
\ee
Substituting Eq.~\eqref{eq: psiDs} into the last line of Eq.~\eqref{eq: splus1}, we obtain that
\be \label{eq: splus2}
s_{+}(\psi)(\hbar) = s_{-}(\psi)(\hbar) + s_{-} (\psi)(\hbar) \sum_{k=1}^{\infty} \re^{-\zeta_k/\hbar} \bar{S}_k  \, , 
\ee
where the asymptotic series $\psi_k(\hbar)$, which resurges from $\psi(\hbar)$ at the singularity $\zeta_k$, is simply
\be \label{eq: psik}
\psi_k(\hbar) = \psi(\hbar) \, , \quad k \in \IZ_{>0} \, ,
\ee
and the Stokes constant $\bar{S}_k \in \IC$ of $\psi(\hbar)$ at the singularity $\zeta_k$ is fully determined by the Stokes constants of $\phi(\hbar)$ via the closed formula
\be \label{eq: newStokesP2}
\bar{S}_k =  \sum_{p \in \CP(k)} \frac{1}{r!} \binom{r}{N_1, \dots , N_k} S_{n_1} \cdots  S_{n_r} \, , \quad k \in \IZ_{>0} \, .
\ee
We stress that the sum over partitions in Eq.~\eqref{eq: newStokesP2} is finite, and thus all the Stokes constants of the original perturbative series $\psi(\hbar)$ are known exactly. More precisely, the discontinuity formula in Eq.~\eqref{eq: splus2} solves analytically the resurgent structure of $\psi(\hbar)$. Note that the instanton sectors associated to the symmetric singularities along the negative imaginary axis are analytically derived from the resurgent structure of $\phi(\hbar)$ by applying the same computations above to the discontinuity of $\psi(\hbar)$ across the angle $3 \pi/2$. We find straightforwardly that, if we define $\CP(k) = \CP(|k|)$ when $k<0$, the formulae in Eqs.~\eqref{eq: psik} and~\eqref{eq: newStokesP2} hold for all values of $k \in \IZ_{\ne 0}$. In particular, we have that $\bar{S}_{k} = \bar{S}_{-k}$. 

Let us point out that the Stokes constants $\bar{S}_k$, $k \in \IZ_{\ne 0}$, are generally complex numbers. However, we can say something more. The discontinuity formula in Eq.~\eqref{eq: discP2series2} can be directly exponentiated to give an exact generating function in terms of known $\tilde{q}$-series for the Stokes constants $\bar{S}_k$. Namely, we find that
\be \label{eq: quotient}
\sum_{k=1}^{\infty} \bar{S}_k \tilde{q}^{k} = \re^{-\ri \pi}  \frac{(w^{-1} ; \, \tilde{q})_{\infty}^3}{(w ; \, \tilde{q})_{\infty}^3} \, ,
\ee
where $w= \re^{2 \pi \ri /3}$. As a consequence of the $q$-binomial theorem, the quotient of $\tilde{q}$-series in the RHS of Eq.~\eqref{eq: quotient} can be expanded in powers of $\tilde{q}$, and the resulting numerical coefficients are combinations of integers, related to the enumerative combinatorics of counting partitions, and complex numbers, arising as integer powers of the complex constants $w, w^{-1}$. Explicitly, the first several Stokes constants $\bar{S}_k$, $k>0$, are
\be
3 \sqrt{3} \ri , -\frac{27}{2} + \frac{3 \sqrt{3} \ri}{2}, -\frac{27}{2} - \frac{21 \sqrt{3} \ri}{2}, -18 \sqrt{3} \ri, 27 - 30 \sqrt{3} \ri, \frac{189}{2} - \frac{51 \sqrt{3} \ri }{2}, 162 + 15 \sqrt{3} \ri,  \dots \, .
\ee
We stress that a special kind of simplification occurs when factoring out the contribution from $S_1=3 \sqrt{3} \ri$. More precisely, if we divide the discontinuity formula in Eq.~\eqref{eq: discP2series2} by $S_1$ and take the exponential of both sides, we find a new generating series, that is, 
\be \label{eq: quotient2}
\sum_{k=1}^{\infty} \bar{S}'_k \tilde{q}^{k} = \re^{-\frac{\pi}{3\sqrt{3}}} \left( \frac{(w ; \, \tilde{q})_{\infty}}{(w^{-1} ; \, \tilde{q})_{\infty}} \right)^{\frac{\ri}{\sqrt{3}}} \, ,
\ee
where the new constants $\bar{S}'_k$ are, notably, rational numbers.
Let us remark that these rational Stokes constants $\bar{S}'_k$ appear naturally in the resurgent study of the normalized perturbative series $\phi'(\hbar)= \phi(\hbar)/3\sqrt{3}\ri$ after exponentiation. Explicitly, the first several values of $\bar{S}'_k$, $k>0$, are
\be
1, 1, \frac{5}{3}, \frac{13}{6}, \frac{83}{30}, \frac{299}{90}, \frac{419}{90}, \frac{409}{72}, \frac{23137}{3240} , \frac{138761}{16200} , \frac{1894921}{178200}, \frac{14008261}{1069200}, \frac{3406991}{213840}, \dots \, .
\ee
Finally, we note that the numbers $k! \bar{S}'_k$, $k \in \IZ_{> 0}$, define a sequence of positive integers, which is
\be
1, 2, 10, 52, 332, 2392, 23464, 229040, 2591344, 31082464, 424462304, 6275700928 , \dots \, .
\ee

\subsection{Exact solution to the resurgent structure for \texorpdfstring{$\hbar \rightarrow \infty$}{hbar going to infinity}} \label{sec: P2infty}
Let us go back to the exact formula for the first spectral trace of local $\IP^2$ in Eq.~\eqref{eq: P2fact} and derive its all-orders perturbative expansion in the dual limit $\hbar \rightarrow \infty$. In the strong-weak coupling duality of Eq.~\eqref{eq: duality} between the spectral theory of the operator $\rho_{\IP^2}$ and the standard topological string theory on local $\IP^2$, this regime corresponds to the weakly-coupled limit $g_s \rightarrow 0$ of the topological string. The resurgent structure of the perturbative series $Z_{\IP^2}(1, \hbar \rightarrow \infty)$ has been studied numerically in~\cite{GuM}. We will show here how the same procedure that we have presented in Section~\ref{sec: P2zero} for the semiclassical limit $\hbar \rightarrow 0$ can be straightforwardly applied to the dual case. We obtain in this way a fully analytic solution to the resurgent structure of $\text{Tr}(\rho_{\IP^2})$ for $\hbar \rightarrow \infty$.

Let us start by applying the known asymptotic expansion formula for the quantum dilogarithm in Eq.~\eqref{eq: logPhiNC} to the anti-holomorphic blocks in Eq.~\eqref{eq: P2fact} and explicitly evaluate the special functions that appear.
We recall that
\begin{subequations} \label{eq: dilog0}
\begin{align}
\log(1-w) - 2 \log(1-w^{-1}) &= -\frac{\pi \ri}{2} - \frac{1}{2} \log(3) \, , \\
\text{Li}_2(w)-2\text{Li}_2(w^{-1}) &= \frac{\pi^2}{18} + \ri V \, , \\
\text{Li}_0(w)-2\text{Li}_0(w^{-1}) &= \frac{1}{2} + 3\sqrt{3} \ri B_1(2/3) \, ,
\end{align}
\end{subequations}
where we have defined $V = 2 \Im \left( \text{Li}_2(\re^{\pi \ri/3}) \right)$ and $w=\re^{2 \pi \ri/3}$, as before. For integer $n \ge 2$, the dilogarithm functions give
\be \label{eq: dilog1}
\ba
&\text{Li}_{2-2n}(w) - 2 \text{Li}_{2-2n}(w^{-1}) = \sum_{s=1}^{\infty} \frac{1}{s^{2-2n}}(w^s - 2 w^{-s}) = \\
&= - 3^{2n-2} \left[ \zeta(2-2n) + \left( \frac{1}{2} + \frac{3 \sqrt{3} \ri}{2} \right) \zeta\left(2-2n , \frac{1}{3} \right) + \left( \frac{1}{2} - \frac{3 \sqrt{3} \ri}{2} \right) \zeta\left(2-2n , \frac{2}{3} \right) \right] \, .
\ea
\ee
Using the identity
\be
\zeta\left(2-2n , \frac{1}{3} \right) + \zeta\left(2-2n , \frac{2}{3} \right) \propto \zeta(2-2n) = 0 \, , \quad n \in \IZ_{>1} \, , 
\ee
and the formulae in Eq.~\eqref{eq: betaID}, the expression in Eq.~\eqref{eq: dilog1} simplifies to
\be \label{eq: dilog2}
\text{Li}_{2-2n}(w) - 2 \text{Li}_{2-2n}(w^{-1}) = 3^{2n-1} \sqrt{3}\ri  \frac{B_{2n-1}(2/3)}{2n-1} \, .
\ee
Substituting Eqs.~\eqref{eq: dilog0} and~\eqref{eq: dilog2} into the asymptotic expansion formula in Eq.~\eqref{eq: logPhiNC}, we find that the anti-holomorphic blocks contribute in the limit $\hbar \rightarrow \infty$ as
\be \label{eq: antiholoblockP2}
\ba
\log(w; \, \tilde{q})_{\infty} - 2 \log(w^{-1}; \, \tilde{q})_{\infty} &= -\frac{\pi \ri}{12} \mb^{-2} - \frac{\pi \ri}{4} - \frac{1}{4} \log(3) + \left(\frac{\pi \ri}{36}- \frac{V}{2 \pi} \right) \mb^2 \\
& - \sqrt{3} \ri \sum_{n=1}^{\infty} (6 \pi \ri \mb^{-2})^{2n-1} \frac{B_{2n}B_{2n-1}(2/3)}{(2n-1)(2n)!} \, , 
\ea
\ee
while the holomorphic blocks contribute trivially as
\be \label{eq: holoblockP2}
\frac{(q^{2/3} ; \, q)_{\infty}^2}{(q^{1/3} ; \, q)_{\infty}} \sim 1 \, .
\ee
\begin{remark}
We note that the terms of order $\mb^2$ and $\mb^{-2}$ in Eq.~\eqref{eq: antiholoblockP2} only partially cancel with the opposite contributions from the exponential in Eq.~\eqref{eq: P2fact}, leaving the exponential factor
\be \label{eq: pre-factorExp}
\text{Tr}(\rho_{\IP^2}) \sim  \exp\left(-\frac{3 V}{g_s} \right) \, , \quad g_s \rightarrow 0 \, ,
\ee
which proves the statement of the conifold volume conjecture\footnote{The conifold volume conjecture for toric CY manifolds has been tested in examples of genus one and two in~\cite{CGM2, CGuM, MZ, KMZ, Volume}.} in the special case of local $\IP^2$. 
A dominant exponential of the form in Eq.~\eqref{eq: pre-factorExp} was already found numerically in~\cite{GuM} for both local $\IP^2$ and local $\IF_0$ in the weakly-coupled limit $g_s \rightarrow 0$. However, we show in this paper that, for the same geometries, the perturbative expansion in the limit $\hbar \rightarrow 0$ of the first fermionic spectral trace does not have such a global exponential pre-factor, being dominated by a leading term of order $\hbar^{-1}$. This suggests that there is no analogue of the conifold volume conjecture in the semiclassical regime.
\end{remark}

Substituting Eqs.~\eqref{eq: antiholoblockP2} and~\eqref{eq: holoblockP2} into Eq.~\eqref{eq: P2fact}, and using $2 \pi \mb^2 = 3 \hbar$, we obtain the all-orders perturbative expansion for $\hbar \rightarrow \infty$ of the first spectral trace of local $\IP^2$ in the form
\be \label{eq: expP2infty}
\text{Tr}(\rho_{\IP^2}) = \sqrt{\frac{2 \pi}{3^{5/2} \hbar}} \re^{-\frac{3 V}{4 \pi^2} \hbar} \exp \left( \sqrt{3} \sum_{n = 1}^{\infty} (-1)^{n-1} \frac{B_{2n} B_{2n-1}(2/3)}{(2n)! (2n-1)} \left( \frac{4 \pi^2}{\hbar} \right)^{2n-1} \right) \, ,
\ee
which has coefficients in $\IQ[\pi, \sqrt{3}]$ up to the global pre-factor. 
The formula in Eq.~\eqref{eq: expP2infty} allows us to compute the coefficients of the perturbative series for $Z_{\mathbb{P}^2}(1, \hbar \rightarrow \infty)$ at arbitrarily high order. The first few terms are
\be \label{eq: first-terms-P2-infty}
1 + \frac{\pi^2}{6 \sqrt{3} \hbar} + \frac{\pi^4}{216 \hbar^2} - \frac{59 \pi^6}{19440 \sqrt{3}\hbar^3}- \frac{251 \pi^8}{1399680\hbar^4} + \frac{23687 \pi^{10}}{58786560 \sqrt{3} \hbar^5} + \text{O}(\hbar^{-6}) \, ,
\ee
multiplied by the global pre-factor in Eq.~\eqref{eq: expP2infty}.

\subsubsection{Resumming the Borel transform} \label{sec: borel-infty}
Let us introduce the parameter $\tau = -1/\mb^2 = - 2 \pi /3\hbar$ and denote by $\phi(\tau)$ the formal power series appearing in the exponent in Eq.~\eqref{eq: expP2infty}. Namely, 
\be \label{eq: phiP2infty}
\phi(\tau) = \sum_{n=1}^{\infty} a_{2n} \tau^{2n-1} \in \IQ[\pi, \sqrt{3}] [\![\tau]\!] \, , \quad a_{2n} = (-1)^{n} \sqrt{3} \frac{B_{2n} B_{2n-1}(2/3)}{(2n)! (2n-1)} (6 \pi)^{2n-1}  \quad n \ge 1 \, ,
\ee
which is simply related to the perturbative expansion in the limit $\hbar \rightarrow \infty$ of the logarithm of the first spectral trace of local $\IP^2$ by 
\be
\log \text{Tr}(\rho_{\IP^2}) = \phi(\tau) + \frac{V}{2 \pi \tau} +\frac{1}{2} \log(\tau) -\frac{3}{4} \log(3) + \frac{\pi \ri}{2}  \, .
\ee
As a consequence of the known asymptotic behavior of the Bernoulli polynomials in Eq.~\eqref{eq: asymBernoulli}, we obtain that the coefficients of $\phi(\tau)$ satisfy the expected factorial growth
\be
a_{2n} \sim (-1)^n (2n)! \left(\frac{2 \pi}{3} \right)^{-2n}  \quad n \gg 1 \, ,
\ee
and $\phi(\tau)$ is a Gevrey-1 asymptotic series. 
Its Borel transform is given by
\be \label{eq: hatphizeta2}
\hat{\phi}(\zeta) = \sqrt{3} \sum_{n=1}^{\infty} (-1)^{n} \frac{B_{2n} B_{2n-1}(2/3)}{(2n)! (2n-1) (2n-1)!} (6 \pi \zeta)^{2n-1} \in \IQ[\pi, \sqrt{3}] [\![\zeta]\!] \, ,
\ee
and it is the germ of an analytic function in the complex $\zeta$-plane.  
\begin{proposition}
Using the definition in Eq.~\eqref{eq: prodHa}, we can interpret the Borel transform $\hat{\phi}(\zeta)$ in Eq.~\eqref{eq: hatphizeta2} as the Hadamard product
\be
\hat{\phi}(\zeta) = ( f \diamond g )(\zeta) \, ,
\ee
where the formal power series $f(\zeta)$ and $g(\zeta)$ have finite radius of convergence at the origin $\zeta=0$, and they can be resummed explicitly as\footnote{We impose that $f(0)=g(0)=0$ in order to eliminate the removable singularities of $f(\zeta), g(\zeta)$ at the origin.}
\begin{subequations}
\begin{align}
f(\zeta) &= \sum_{n=1}^{\infty} \frac{B_{2n}}{(2n)!} \zeta^{2n-1} =  -\frac{1}{2 \zeta} \left( 2-\zeta \coth \left(\frac{\zeta}{2} \right) \right) \, , \quad |\zeta| < 2\pi \, , \label{eq: fP2infty} \\
g(\zeta) &= \sum_{n=1}^{\infty} (-1)^n \frac{B_{2n-1}(2/3)}{(2n-1) (2n-1)!} \sqrt{3} (6 \pi)^{2n-1} \zeta^{2n-1} = \frac{3}{2} \log \left(\frac{ \cos(\pi/6 + \pi \zeta) }{ \cos(\pi/6 - \pi \zeta) } \right) \, , \quad |\zeta| < 1/3 \label{eq: gP2infty} \, .
\end{align}
\end{subequations}
\end{proposition}
\begin{proof}
The Bernoulli numbers are defined by the generating function
\be \label{eq: fP2infty-origin}
\sum_{n=0}^{\infty} \frac{B_{n}}{n!} \zeta^{n} = \frac{\zeta}{2} \left( \coth \left(\frac{\zeta}{2} \right) - 1 \right) \, , \quad |\zeta| < 2\pi \, .
\ee 
Taking the even part of both sides of Eq.~\eqref{eq: fP2infty-origin}, and multiplying by $1/\zeta$, we obtain the statement in Eq.~\eqref{eq: fP2infty}.
Let us apply the second identity in Eq.~\eqref{eq: BernoulliToZeta} to the power series in the LHS of Eq.~\eqref{eq: gP2infty} and use the functional equation for the Hurwitz zeta function in Eq.~\eqref{eq: functZeta} for $\zeta(2-2n, 2/3)$, $n \ge 2$.
We find that
\be
g(\zeta) = - \pi \sqrt{3} \zeta + 3 \sum_{n=2}^{\infty} \left(\zeta \left( 2n-1, \frac{2}{3} \right) - \zeta \left( 2n-1, \frac{1}{3} \right) \right) \frac{\zeta^{2n-1}}{2n-1} \, .
\ee
Let us now use the formula in Eq.~\eqref{eq: zetaseriesOdd} for $a=2/3, 1/3$ and recall the known identity
\be
\Psi (2/3) - \Psi(1/3) = \frac{\pi}{\sqrt{3}} \, ,
\ee
where $\Psi(a)$ denotes the digamma function. We obtain in this way that
\be \label{eq: gP2infty-origin}
g(\zeta) =  \frac{3}{2} \log \left( \frac{\Gamma(2/3-\zeta) \Gamma(1/3+\zeta)}{\Gamma(2/3+\zeta) \Gamma(1/3-\zeta)} \right) \, , \quad |\zeta| < 1/3 \, .
\ee
After we apply Euler's reflection formula in Eq.~\eqref{eq: reflectionGamma} with $x=1/3+\zeta, 1/3 - \zeta$ and use the trigonometric identities 
\be
\cos(\pi/6 \pm \pi \zeta) = \sin (\pi/3 \mp \pi \zeta) \, ,
\ee
the formula in Eq.~\eqref{eq: gP2infty-origin} then yields the statement in Eq.~\eqref{eq: gP2infty}.
\end{proof}
After being analytically continued to the whole complex plane, the function $f(\zeta)$ has poles of order one along the imaginary axis at
\be
\mu_m = 2 \pi \ri m \, , \quad m \in \IZ_{\ne 0} \, ,
\ee
while the function $g(\zeta)$ has logarithmic branch points along the real axis at
\be
\nu^{-}_k = -\frac{1}{3} + 2k \, , \quad \nu^{+}_k = \frac{2}{3} + 2k \, , \quad k \in \IZ \, .
\ee

\begin{proposition}
The Borel transform $\hat{\phi}(\zeta)$ in Eq.~\eqref{eq: hatphizeta2} can be expressed as\footnote{We remark that each of the infinite sums giving the Borel transforms in Eqs.~\eqref{eq: intBorel2} and~\eqref{eq: intBorel2infty} can be straightforwardly written as the logarithm of an infinite product.}
\be \label{eq: intBorel2infty}
\hat{\phi}(\zeta) = - \frac{3}{2 \pi \ri} \sum_{m \in \IZ_{\ne 0}} \frac{1}{m} \log \left( \cos \left(\frac{\pi}{6} + \frac{\zeta}{2 \ri m }\right) \right) \, ,
\ee
which is a well-defined, exact function of $\zeta$.
\end{proposition}
\begin{proof}
We consider a circle $\gamma$ in the complex $s$-plane with center $s=0$ and radius $0 < r < 2 \pi$ and apply Theorem~\ref{theo: HadamardT}. The Borel transform can be written as the integral
\be \label{eq: intBorelinfty}
\hat{\phi}(\zeta) = \frac{1}{2 \pi \ri} \int_{\gamma} f(s) g(\zeta/s) \frac{ \rd s}{s} = -\frac{3}{4 \pi \ri} \int_{\gamma} \frac{2-s \coth(s/2)}{s} \log \left(\frac{ \cos(\pi/6 + \pi \zeta/s) }{ \cos(\pi/6 - \pi \zeta/s) } \right) \frac{ \rd s}{s} \, ,
\ee
for $|\zeta| < r /3$. 
We note that, for such values of $\zeta$, the function $s \mapsto g(\zeta/s)$ has logarithmic branch points at $s = \zeta/\nu^{\pm}_k$, $k \in \IZ$, which sit inside the contour of integration $\gamma$ and accumulate at the origin, and no singularities for $|s| > r$. The function $f(s)$ has simple poles at the points $s= \mu_m$ with residues
\be
\underset{s=2 \pi \ri m}{\text{Res}} f(s) = 1 \, , \quad m \in \IZ_{\ne 0} \, . 
\ee
By Cauchy's residue theorem, the integral in Eq.~\eqref{eq: intBorelinfty} can be evaluated by summing the residues at the poles of the integrand which lie outside $\gamma$, allowing us to express the Borel transform as an exact function of $\zeta$. More precisely, we find the desired analytic formula
\be \label{eq: intBorel2infty-proof}
\hat{\phi}(\zeta) = - \sum_{m \in \IZ_{\ne 0}} \underset{s=2 \pi \ri m}{\text{Res}} f(s) g(\zeta/s) \frac{1}{s} = - \frac{3}{2 \pi \ri} \sum_{m \in \IZ_{\ne 0}} \frac{1}{m} \log \left( \cos \left(\frac{\pi}{6} + \frac{\zeta}{2 \ri m }\right) \right) \, .
\ee
The convergence of the infinite sum in the RHS of Eq.~\eqref{eq: intBorel2infty-proof} can be easily verified by, \textit{e.g.}, the limit comparison test. 
\end{proof}
\begin{corollary}
The singularities of the Borel transform $\hat{\phi}(\zeta)$ in Eq.~\eqref{eq: intBorel2infty} are logarithmic branch points located along the imaginary axis at
\be \label{eq: zetakm-infty}
\zeta^{-}_{k,m} = \nu^{-}_{-k} \mu_{-m} = \frac{2 \pi \ri}{3} (1+6k) m \, , \quad \zeta^{+}_{k,m} = \nu^{+}_{-k} \mu_{-m} = \frac{2 \pi \ri}{3} (-2+6k) m \, , \quad k \in \IZ \, , \quad m \in \IZ_{\ne 0} \, ,
\ee
which we write equivalently as
\be \label{eq: zetan-infty}
\zeta_n = \frac{2 \pi \ri}{3} n \, , \quad n \in \IZ_{\ne 0} \, ,
\ee
that is, the branch points lie at all non-zero integer multiples of the two complex conjugate dominant singularities at $\pm 2 \pi \ri/3$, as illustrated in Fig.~\ref{fig: peacockP2infty}. 
\end{corollary}
Analogously to the dual case of $\hbar \rightarrow 0$, there are only two Stokes lines at the angles $\pm \pi/2$.
\begin{figure}[htb!]
\center
 \includegraphics[width=0.2\textwidth]{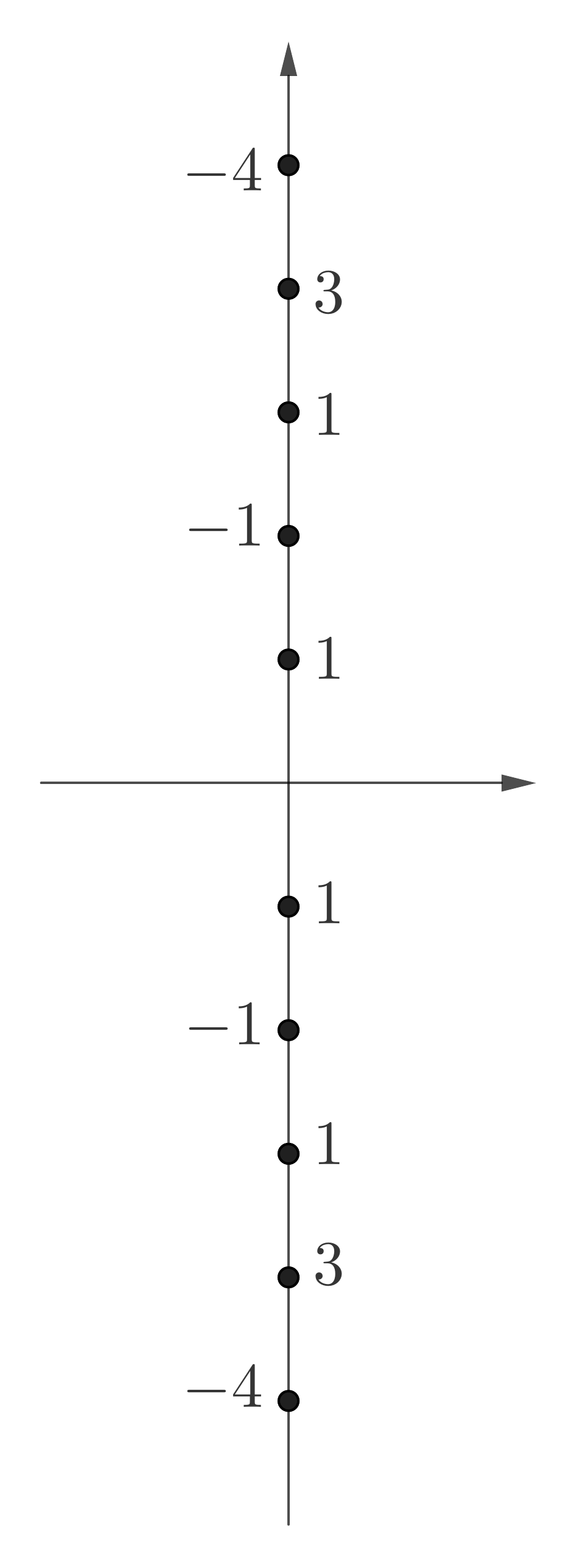}
 \caption{The first few singularities of the Borel transform of the asymptotic series $\phi(\tau)$, defined in Eq.~\eqref{eq: phiP2infty}, and the associated integer constants $\beta_n \in \IZ_{\ne 0}$, defined in Eq.~\eqref{eq: intStokesP2infty}.}
 \label{fig: peacockP2infty}
\end{figure}
Moreover, the analytic expression in Eq.~\eqref{eq: intBorel2infty} is explicitly simple resurgent.
\begin{corollary} \label{cor: localexp-infty}
The local expansion of the Borel transform $\hat{\phi}(\zeta)$ in Eq.~\eqref{eq: intBorel2infty} at $\zeta = \zeta_n$, $n \in \IZ_{\ne 0}$, is given by
\be \label{eq: explocalP2infty}
\hat{\phi}(\zeta) = - \frac{R_n}{2 \pi \text{i}} \log(\zeta - \zeta_n) + \dots \, ,
\ee
where $R_n \in \IC$ is the Stokes constant.
\end{corollary}
\begin{proof}
The local expansion around the logarithmic singularity $\zeta = \zeta_n$ is obtained by summing the contributions from all pairs $(k,m) \in \IZ \times \IZ_{\ne 0}$ such that $n = (-1+6k)m$ or $n=(2+6k)m$. Let us collect such pairs of integers into the finite sets $I^+_n, I^-_n$ for each $n \in \IZ_{\ne 0}$, accordingly.
For a fixed value of $m \in \IZ_{\ne 0}$, we denote the corresponding term in the sum in Eq.~\eqref{eq: intBorel2infty} by 
\be \label{eq: fm-infty}
f_m(\zeta) = - \frac{3}{2 \pi \ri m} \log \left( \cos \left(\frac{\pi}{6} + \frac{\zeta}{2 \ri m }\right) \right) \, ,
\ee 
whose expansion around $\zeta= \zeta^{\pm}_{k,m}$, for fixed $k \in \IZ$, is given by
\be
f_m(\zeta) = -\frac{s^{\pm}_{k,m}}{2 \pi \text{i}} \log(\zeta - \zeta^{\pm}_{k,m}) + \dots \, ,
\ee
where the dots denote regular terms in $\zeta - \zeta^{\pm}_{k,m}$, and $s^{\pm}_{k,m}$ is a complex number.
Since $\zeta_n = \zeta^{\pm}_{k,m}$ for all $(k,m) \in I^{\pm}_n$, it follows that the local expansion of $\hat{\phi}(\zeta)$ at $\zeta = \zeta_n$ is again given by
\be \label{eq: explocalP2infty-proof}
\hat{\phi}(\zeta) = \sum_{m \in \IZ_{\ne 0}} f_m(\zeta) = - \frac{R_n}{2 \pi \text{i}} \log(\zeta - \zeta_n) + \dots \, ,
\ee
where the Stokes constant $R_n$ is now the finite sum 
\be \label{eq: stokesSuminfty}
R_n = \sum_{(k,m) \in I^+_n} s^+_{k,m} + \sum_{(k,m) \in I^-_n} s^-_{k,m} \, .
\ee
\end{proof}
It follows from Corollary~\ref{cor: localexp-infty} that the locally analytic function that resurges at $\zeta=\zeta_n$ is trivially $\hat{\phi}_n(\zeta-\zeta_n) = 1$, which also implies that $\phi_n(\tau) = 1$, $n \in \IZ_{\ne 0}$.
Once again, the procedure above allows us to derive analytically all the Stokes constants. 
In the limit $\hbar \rightarrow \infty$, the Stokes constants $R_n$ are rational numbers, and they too are simply related to an interesting sequence of integers $\beta_n$,  $n \in \IZ_{\ne 0}$. In particular, we find that 
\begin{subequations} \label{eq: intStokesP2infty}
\be
R_1 = 3 \, , \quad R_n = R_1 \frac{\beta_n}{n} \quad n \in \IZ_{\ne 0,1} \, , 
\ee
\be
\beta_n = \beta_{-n} \, , \quad \beta_n \in \IZ_{\ne 0}  \quad n \in \IZ_{> 0} \, .
\ee
\end{subequations}
Explicitly, the first several integer constants $\beta_n$, $n>0$, are
\be
1, -1, 1, 3, -4, -1, 8, -5, 1, 4, -10, 3, 14, -8, -4, 11, -16, -1, 20, -12, \dots \, .
\ee
The pattern of singularities in the Borel plane and the associated $\beta_n \in \IZ_{\ne 0}$ are shown in Fig.~\ref{fig: peacockP2infty}. 

\subsubsection{Closed formulae for the Stokes constants} \label{sec: closed-infty}
The exact resummation of the Borel transform in Eq.~\eqref{eq: intBorel2infty} allows us to obtain and prove a series of exact arithmetic formulae for the Stokes constants $R_n$ of the asymptotic series $\phi(\tau)$, defined in Eq.~\eqref{eq: phiP2infty}, and the related integer constants $\beta_n$, defined in Eq.~\eqref{eq: intStokesP2infty}, for $n \in \IZ_{\ne 0}$. These new arithmetic statements are manifestly dual to the analogous formulae that we have presented in Section~\ref{sec: closed-zero} for the semiclassical limit of $\hbar \rightarrow 0$. Besides, their proofs follow very similar arguments. Let us start by showing that both sequences $R_n, \beta_n$ are number-theoretic divisor sum functions.
\begin{proposition} \label{prop: formulaS1infty} 
The normalized Stokes constant $R_n/R_1$, where $R_1 = 3$, is determined by the positive integer divisors of $n \in \IZ_{\ne 0}$ according to the closed formula
\be \label{eq: formulaS1infty}
\frac{R_n}{R_1} = \sum_{\substack{d | n \\ d \equiv_3 1}} \frac{d}{n} - \sum_{\substack{d | n \\ d \equiv_3 2}} \frac{d}{n} \, , 
\ee
which implies that $R_n = -R_{-n}$ and $R_n/R_1 \in \IQ_{\ne 0}$.
\end{proposition}
\begin{proof}
Let us denote by $D_n$ the set of positive integer divisors of $n$. We recall that $n$ satisfies one of the two factorization properties $n = (1+6k)m$ or $n=(-2+6k)m$ for $k \in \IZ$ and $m \in \IZ_{\ne 0}$. It follows that one of four cases apply. Namely,
\begin{subequations}
\begin{align}
m = \frac{n}{d} \, , \quad k = \frac{d-1}{6} \, , \quad &\mbox{if} \quad d \equiv 1 \mod 6 \, , \label{eq: pairs1} \\
m = -\frac{n}{d} \, , \quad k = -\frac{d-2}{6} \, , \quad &\mbox{if} \quad d \equiv 2 \mod 6 \, , \label{eq: pairs2} \\
m = \frac{n}{d} \, , \quad k = \frac{d+2}{6} \, , \quad &\mbox{if} \quad d \equiv 4 \mod 6 \, , \label{eq: pairs4} \\
m = -\frac{n}{d} \, , \quad k = -\frac{d+1}{6} \, , \quad &\mbox{if} \quad d \equiv 5 \mod 6 \, ,  \label{eq: pairs5}
\end{align}
\end{subequations}
where $d \in D_n$.
In both cases of Eqs.~\eqref{eq: pairs1} and~\eqref{eq: pairs4}, which represent together the congruence class of $d \equiv_3 1$, substituting the given values of $k,m$ into Eqs.~\eqref{eq: fm-infty} and~\eqref{eq: zetakm-infty}, we find that the contribution to the Stokes constant $R_n$ coming from the local expansion of $f_m(\zeta)$ around $\zeta^{\pm}_{k,m}$ is simply $s^{\pm}_{k,m} = 3d/n$. 
Furthermore, in both cases of Eqs.~\eqref{eq: pairs2} and~\eqref{eq: pairs5}, which populate the congruence class of $d \equiv_3 2$, substituting the given values of $k,m$ into Eqs.~\eqref{eq: fm-infty} and~\eqref{eq: zetakm-infty}, we find that the contribution to the Stokes constant $R_n$ coming from the local expansion of $f_m(\zeta)$ around $\zeta^{\pm}_{k,m}$ is simply $s^{\pm}_{k,m} = -3d/n$.
Finally, for any divisor $d \in D_n$ which is a multiple of $3$, neither $d\pm1$ or $d\pm2$ are divisible by $6$, which implies that the choice $m = \pm n/d$ is not allowed, and the corresponding contribution is $s^{\pm}_{k,m}=0$.
Putting everything together and using Eq.~\eqref{eq: stokesSuminfty}, we find the desired statement.
\end{proof}
We note that the arithmetic formula for the Stokes constants in Eq.~\eqref{eq: formulaS1infty} can be written equivalently as
\be \label{eq: formulaA1infty}
\beta_n = \sum_{\substack{d | n \\ d \equiv_3 1}} d - \sum_{\substack{d | n \\ d \equiv_3 2}} d \, , 
\ee
which implies that $\beta_n = \beta_{-n}$ and $\beta_n \in \IZ_{\ne 0}$ for all $n>0$, as expected. We highlight the simple symmetry between the formulae in Eqs.~\eqref{eq: formulaS1},~\eqref{eq: formulaA1} and the formulae in Eqs.~\eqref{eq: formulaS1infty},~\eqref{eq: formulaA1infty}. More precisely, the Stokes constants $R_n$ in the limit $\hbar \rightarrow \infty$ are obtained from the Stokes constants $S_n$ in the semiclassical limit $\hbar \rightarrow 0$ via the simple exchange of divisors $d \mapsto n/d$ in the arguments of the sums.
As before, two corollaries follow straightforwardly from Proposition~\ref{prop: formulaS1infty}.
\begin{corollary} \label{prop: formulaA1mult-infty} 
The integer constants $\beta_n$, $n \in \IZ_{>0}$, satisfy the closed formulae
\be
\beta_{p_1^{e_1}} = \frac{p_1^{e_1+1}-1}{p_1-1} \, , \quad \beta_{p_2^{e_2}} = \frac{(-1)^{e_2} p_2^{e_2+1}+1}{p_2+1} \, , \quad \beta_{p_3^{e_3}} = 1 \, , 
\ee
where $e_i \in \IN$, and $p_i \in \IP$ are prime numbers such that $p_i \equiv_3 i$ for $i=1,2,3$. Moreover, they obey the multiplicative property
\be \label{eq: formulaA1mult-infty}
\beta_n = \prod_{p \in \IP} \beta_{p^e} \, , \quad n = \prod_{p \in \IP} p^e \, , \quad e \in \IN \, .
\ee
\end{corollary}
\begin{proof}
The three closed formulae follow directly from Eq.~\eqref{eq: formulaA1infty}. Explicitly, let $n=p^e$ with $p \in \IP$ and $e \in \IN$. We have that
\begin{subequations}
\begin{align}
\sum_{\substack{d | n \\ d \equiv_3 1}} d = 1 \, , \quad  \sum_{\substack{d | n \\ d \equiv_3 2}} d = 0 \, , \quad &\mbox{if} \quad p \equiv 0 \mod 3 \, , \\
\sum_{\substack{d | n \\ d \equiv_3 1}} d = \sum_{i=0}^e p^i = \frac{p^{e+1}-1}{p-1} \, , \quad  \sum_{\substack{d | n \\ d \equiv_3 2}} d = 0 \, , \quad &\mbox{if} \quad p \equiv 1 \mod 3 \, , \\
\sum_{\substack{d | n \\ d \equiv_3 1}} d = \sum_{i=0}^{\lfloor e/2 \rfloor} p^{2i} \, , \quad  \sum_{\substack{d | n \\ d \equiv_3 2}} d = \sum_{i=0}^{\lfloor e/2 \rfloor} p^{2i+1} \, , \quad &\mbox{if} \quad p \equiv 2 \mod 3 \, .
\end{align}
\end{subequations}
Let us now prove the multiplicity property. We will prove a slightly stronger statement. We write $n = p q$ for $p,q \in \IZ_{> 0}$ coprimes. We choose a positive integer divisor $d | n$, and we write $d = s t$ where $s | p$ and $t | q$. Consider two cases:
\begin{itemize}
\item[(1)] Suppose that $d \equiv_3 1$. Then, either $s \equiv_3 t \equiv_3 1$, or $s \equiv_3 t \equiv_3 2$, and therefore
\be \label{eq: fact1formulaA1-infty}
\sum_{\substack{d | n \\ d \equiv_3 1}} d = \sum_{\substack{s | p \\ s \equiv_3 1}} s  \sum_{\substack{t | q \\ t \equiv_3 1}} t + \sum_{\substack{s | p \\ s \equiv_3 2}} s  \sum_{\substack{t | q \\ t \equiv_3 2}} t \, .
\ee
\item[(2)] Suppose that $d \equiv_3 2$. Then, either $p/s \equiv_3 1$ and $q/t \equiv_3 2$, or $p/s \equiv_3 2$ and $q/t \equiv_3 1$, and therefore
\be \label{eq: fact2formulaA1-infty}
\sum_{\substack{d | n \\ d \equiv_3 2}} d = \sum_{\substack{s | p \\ s \equiv_3 1}} s  \sum_{\substack{t | q \\ t \equiv_3 2}} t + \sum_{\substack{s | p \\ s \equiv_3 2}} s  \sum_{\substack{t | q \\ t \equiv_3 1}} t \, .
\ee
\end{itemize}
Substituting Eqs.~\eqref{eq: fact1formulaA1-infty} and~\eqref{eq: fact2formulaA1-infty} into Eq.~\eqref{eq: formulaA1infty}, we find that
\be
\beta_n = \left( \sum_{\substack{s | p \\ s \equiv_3 1}} s - \sum_{\substack{s | p \\ s \equiv_3 2}} s \right) \left( \sum_{\substack{t | q \\ t \equiv_3 1}} t - \sum_{\substack{t | q \\ t \equiv_3 2}} t \right) = \beta_p \beta_q  \, , 
\ee
which proves that the sequence $\beta_n$, $n \in \IZ_{>0}$, defines a multiplicative arithmetic function. Note that the proof breaks if $p,q$ are not coprimes, since the formulae above lead in general to overcounting the contributions coming from common factors. Therefore, the sequence $\beta_n$ is not totally multiplicative. Note that the sequence of normalized Stokes constants $R_n/R_1$, $n \in \IZ_{>0}$, is also a multiplicative arithmetic function.
\end{proof}
\begin{corollary} \label{prop: formulaA2infty} 
The integer constants $\beta_n$, $n  \in \IZ_{>0}$, are encoded in the generating function 
\be \label{eq: formulaA2infty}
\sum_{n =1}^{\infty} \beta_n x^n = \sum_{m =1}^{\infty} \frac{x^m(1-x^{2m})}{(1+x^m+x^{2m})^2} \, .
\ee
\end{corollary}
\begin{proof}
We denote by $f(x)$ the generating function in the RHS of Eq.~\eqref{eq: formulaA2infty}. We note that 
\be
f(x) = f_1(x) - f_2(x) \, ,
\ee
where the functions $f_1(x), f_2(x)$ are defined by
\be
f_1(x) = \sum_{m \in \IN} \frac{(3m+1) x^{3m+1}}{1-x^{3m+1}} \, , \quad f_2(x) = \sum_{m \in \IN} \frac{(3m+2) x^{3m+2}}{1-x^{3m+2}} \, .
\ee
The formula in Eq.~\eqref{eq: formulaA2infty} follows from the stronger statement
\be
\sum_{\substack{d | n \\ d \equiv_3 1}} d = \frac{1}{n!} \frac{\rd^n f_1(0)}{\rd x^n} \, , \quad  \sum_{\substack{d | n \\ d \equiv_3 2}} d = \frac{1}{n!} \frac{\rd^n f_2(0)}{\rd x^n} \, , \quad n  \in \IZ_{>0} \, .
\ee
We will now prove this claim for the function $f_1(x)$. The case of $f_2(x)$ is proven analogously.
Let us denote by 
\be
f_{1,m}(x) = \frac{(3m+1) x^{3m+1}}{1-x^{3m+1}} \, , \quad m \in \IN \, , 
\ee
and consider the derivative $\rd^n f_{1,m}(x)/\rd x^n$ for fixed $m$. We want to determine its contributions to $\rd^n f_1(0)/\rd x^n$. Since we are interested in those terms that survive after taking $x=0$, we look for the monomials of order $x^{d(3m+1)-n}$, where $d|n$, in the numerator of $\rd^n f_{1,m}(x)/\rd x^n$, and we take $3m+1 = n/d$. More precisely, let us introduce the parameter $q = 3m+1$ for simplicity of notation. Deriving $a$-times the factor $q x^q$ and $(n-a)$-times the factor $(1-x^q)^{-1}$, we have the term
\be  \label{eq: monomials0-infty}
\binom{n}{a} \frac{\rd^a (q x^q)}{\rd x^a} \frac{\rd^{n-a} (1-x^q)^{-1}}{\rd x^{n-a}} \, , \quad a \in \IN_{\ne 0} \, .
\ee
Recall that the generalized binomial theorem for the geometric series yields
\be \label{eq: geomseries-infty}
\frac{\rd^{n-a} (1-x^q)^{-1}}{\rd x^{n-a}} = \sum_{k=0}^{\infty} \frac{(qk)!}{(qk-n+a)!} x^{qk-n+a} \, .
\ee
Substituting Eq.~\eqref{eq: geomseries-infty} into Eq.~\eqref{eq: monomials0-infty} and performing the derivation, we find
\be
n! \sum_{k=0}^{\infty} q \binom{q}{q-a} \binom{qk}{qk-n+a} x^{(1+k)q-n} \, . 
\ee
It follows then that the only non-zero term at fixed $m \in \IN$ comes from the values of $k \in \IN$ and $a \in \IN_{\ne 0}$ such that $(1+k)q = n$ and $a=q=3m+1$, which implies in turn that $q|n$ with $q \equiv_3 1$. Finally, summing the non-trivial contributions over $q$ gives precisely
\be
\frac{\rd^n f_1(0)}{\rd x^n} = \sum_{\substack{q|n \\ q \equiv_3 1}} n! q \binom{q}{0} \binom{n-q}{0} = n! \sum_{\substack{q|n \\ q \equiv_3 1}} q \, .
\ee
\end{proof}
Finally, Proposition~\ref{prop: formulaS1infty} implies that the Stokes constants $R_n$, $n \in \IZ_{>0}$, can be naturally organized as coefficients of an exact generating function given by quantum dilogarithms.
\begin{corollary} \label{prop: formulaS2infty} 
The Stokes constants $R_n$, $n  \in \IZ_{>0}$, are encoded in the generating function 
\be \label{eq: formulaS2infty}
\sum_{n =1}^{\infty} R_n x^{n/3} = 3 \log \frac{(x^{2/3} ; \, x)_{\infty}}{(x^{1/3} ; \, x)_{\infty}} \, , \quad |x|<1 \, .
\ee
\end{corollary}
\begin{proof}
We apply the definition of the quantum dilogarithm in Eq.~\eqref{eq: dilog} and Taylor expand the logarithm function for $|x| <1$. We obtain in this way that 
\be
\log (x^{2/3} ; \, x)_{\infty} = \sum_{k=0}^{\infty} \log(1-x^{2/3+k}) = -\sum_{k=0}^{\infty} \sum_{m=1}^{\infty} \frac{x^{2m/3+mk}}{m} \, ,
\ee
and therefore also
\be \label{eq: logfrac-infty}
\log \frac{(x^{2/3} ; \, x)_{\infty}}{(x^{1/3} ; \, x)_{\infty}} = \sum_{k=0}^{\infty} \sum_{m=1}^{\infty} \frac{1}{m} x^{(1+3k)m/3}  - \sum_{k=0}^{\infty} \sum_{m=1}^{\infty} \frac{1}{m} x^{(2+3k)m/3} \, .
\ee
Renaming $n=(1+3k)m$ and $n=(2+3k)m$ in the first and second terms of the RHS, respectively, we find
\be \label{eq: logfrac2-infty}
\log \frac{(x^{2/3} ; \, x)_{\infty}}{(x^{1/3} ; \, x)_{\infty}} = \sum_{n =1}^{\infty} x^{n/3}  \left( \sum_{\substack{d|n \\ d \equiv_3 1}} \frac{d}{n} - \sum_{\substack{d|n \\ d \equiv_3 2}} \frac{d}{n}  \right) \, .
\ee
Substituting the arithmetic formula for the Stokes constants in Eq.~\eqref{eq: formulaS1infty} into the expression in Eq.~\eqref{eq: logfrac2-infty}, we obtain the desired statement.
\end{proof}
We note that, choosing $x = q = \re^{2 \pi \ri \mb^2} = \re^{-2 \pi \ri \tau^{-1}}$, Corollary~\ref{prop: formulaS2infty} directly provides an exact $q$-series expression for the discontinuity of the asymptotic series $\phi(\tau)$ across the positive imaginary axis, which borders the only two distinct Stokes sectors in the upper half of the Borel plane. Namely, recalling that $\phi_n(\tau)=1$, $n \in \IZ_{>0}$, we have that
\be \label{eq: discP2series2infty}
\text{disc}_{\pi/2}\phi(\tau) = \sum_{n=1}^{\infty} R_n \re^{-n 2 \pi \ri /3 \tau} = 3 \log (q^{2/3} ; \, q)_{\infty} - 3 \log (q^{1/3} ; \, q)_{\infty} \, ,
\ee
which is dual to the discontinuity formula in Eq.~\eqref{eq: discP2series2}. We stress that the $q$-series $(q^{2/3} ; \, q)_{\infty}$ and $(q^{1/3} ; \, q)_{\infty}$ occur as the holomophic block of the first spectral trace of local $\IP^2$ in Eq.~\eqref{eq: P2fact}. 

\subsubsection{Exact large-order relations} \label{sec: zeta-infty}
Following the same arguments of Section~\ref{sec: zeta-zero}, we provide here a number-theoretic characterization of the perturbative coefficients $a_{2n}$, $n \in \IN_{\ne 0}$, of the asymptotic series $\phi(\tau)$ in Eq.~\eqref{eq: phiP2infty}. 
Once again, we upgrade the large-$n$ asymptotics of the coefficients $a_{2n}$ by systematically including the contributions from all sub-dominant singularities in the Borel plane and obtain in this way an exact large-order relation, that is, 
\be \label{eq: exactlarge-infty}
a_{2n} = \frac{(-1)^n}{\pi} \frac{\Gamma(2n-1)}{A^{2n-1}} \sum_{m=1}^{\infty} \frac{R_m}{m^{2n-1}} \, , \quad  n \in \IN_{\ne 0} \, ,
\ee
where the Stokes constant $R_m$ is given explicitly in Eq.~\eqref{eq: intStokesP2infty}, and we have defined $A=2\pi/3$.
\begin{proposition} \label{prop: formulaS3infty} 
The Stokes constants $R_m$, $m \in \IZ_{> 0}$, satisfy the exact relations
\be \label{eq: formulaS3infty}
\sum_{m=1}^{\infty} \frac{R_m}{m^{2n-1}} = 3 \frac{\zeta(2n)}{3^{2n-1}} \left(\zeta\left(2n-1,\frac{1}{3}\right) - \zeta\left(2n-1, \frac{2}{3} \right) \right) \, , \quad n \in \IN_{\ne 0} \, ,
\ee
where $\zeta(z)$ denotes the Riemann zeta function, and $\zeta(z, a)$ denotes the Hurwitz zeta function.
\end{proposition}
\begin{proof}
Substituting the original expression for the perturbative coefficients $a_{2n}$, $n \in \IN_{\ne 0}$, in Eq.~\eqref{eq: phiP2infty} into the exact large-order relation in Eq.~\eqref{eq: exactlarge-infty}, we have that
\be \label{eq: proofS3-1infty}
\sum_{m=1}^{\infty} \frac{R_m}{m^{2n-1}} =  \pi \sqrt{3} (2 \pi)^{4n-2} \frac{B_{2n} B_{2n-1}(2/3)}{(2n)! (2n-1)!} \, .
\ee
Using the known identities
\begin{subequations}
\be
\zeta(2n) = (-1)^{n+1} \frac{(2 \pi)^{2n} B_{2n}}{2 (2n)!} \, ,
\ee
\be \label{eq: BernoulliToZeta}
B_{2n-1}(2/3) = - B_{2n-1}(1/3) \, , \quad B_{2n-1}(z) = -(2n-1) \zeta(2-2n, z) \, ,
\ee
\end{subequations}
the formula in Eq.~\eqref{eq: proofS3-1infty} becomes
\be \label{eq: proofS3-2infty}
\sum_{m=1}^{\infty} \frac{R_m}{m^{2n-1}} = \pi \sqrt{3} (-1)^n (2 \pi)^{2n-2} \frac{\zeta(2n)}{(2n-2)!} \left(\zeta\left(2-2n,\frac{2}{3}\right) - \zeta\left(2-2n, \frac{1}{3} \right) \right) \, .
\ee
The functional identity for the Hurwitz zeta function in Eq.~\eqref{eq: functZeta} yields
\be \label{eq: proofS3-zetainfty}
\zeta\left(2-2n,\frac{2}{3}\right) - \zeta\left(2-2n, \frac{1}{3}\right) = \frac{2 \sqrt{3} (2n-2)!}{(-1)^{n-1} (6 \pi)^{2n-1}} \left( \zeta\left(2n-1,\frac{2}{3}\right) - \zeta\left(2n-1, \frac{1}{3}\right) \right) \, ,
\ee
and substituting this into Eq.~\eqref{eq: proofS3-2infty}, we obtain the desired statement.
\end{proof}
\begin{remark}
We note that the exact expression in Eq.~\eqref{eq: formulaS3infty} can be written equivalently in terms of the integer constants $\beta_m$,  $m \in \IZ_{\ne 0}$. Namely, 
\be \label{eq: formulaA3infty}
\sum_{m=1}^{\infty} \frac{\beta_m}{m^{2n}} = \frac{\zeta(2n)}{3^{2n-1}} \left(\zeta\left(2n-1,\frac{1}{3}\right) - \zeta\left(2n-1, \frac{2}{3} \right) \right) \, , \quad n \in \IN_{\ne 0} \, ,
\ee
which is dual to the formula in Eq.~\eqref{eq: formulaA3}. As before, let us point out that the series in the LHS of Eq.~\eqref{eq: formulaA3infty} belongs to the family of Dirichlet series~\cite{Dirichlet}.
As a consequence of Corollary~\ref{prop: formulaA1mult-infty}, the sequence of integers $\beta_m$, $m \in \IZ_{\ne 0}$, defines a bounded multiplicative arithmetic function, and the corresponding Dirichlet series satisfies an expansion as an Euler product indexed by the set of prime numbers $\IP$, that is, 
\be \label{eq: L-series-infty}
\sum_{m=1}^{\infty} \frac{\beta_m}{m^{2n}} = \prod_{p \in \IP} \sum_{e=0}^{\infty} \frac{\beta_{p^e}}{p^{e(2n)}} \, , \quad n \in \IN_{\ne 0} \, ,
\ee
which proves that the given Dirichlet series is, indeed, an $L$-series. We will further explore in this direction in Section~\ref{sec: dirichlet}.
\end{remark}

\subsubsection{Exponentiating with alien calculus} \label{sec: exp-infty}
Let us now translate our analytic solution to the resurgent structure of the asymptotic series $\phi(\tau)$ in Eq.~\eqref{eq: phiP2infty} into results on the original, exponentiated perturbative series in Eq.~\eqref{eq: expP2infty}, which we denote by
\be \label{eq: psiP2infty}
\psi(\tau) = \re^{\phi(\tau)} = \exp \left( \sqrt{3} \sum_{n = 1}^{\infty} (-1)^{n} \frac{B_{2n} B_{2n-1}(2/3)}{(2n)! (2n-1)} (6 \pi \tau)^{2n-1} \right) \in \IQ[\pi, \sqrt{3}] [\![\tau]\!] \, ,
\ee
and which is also a Gevrey-1 asymptotic series. Its Borel transform $\hat{\psi}(\zeta)$ inherits from $\hat{\phi}(\zeta)$ the same pattern of singularities in Eq.~\eqref{eq: zetan-infty}. Namely, there are infinitely many and discrete logarithmic branch points located along the imaginary axis of the complex $\zeta$-plane at $\zeta_n = n 2 \pi \ri/3$, $n \in \IZ_{\ne 0}$.
Let us denote by $s_{\pm}(\psi)(\tau)$ the lateral Borel resummations at the angles $\pi/2 \pm \epsilon$ with $\epsilon \ll 1$, which lie slightly above and slightly below the Stokes line along the positive imaginary axis, respectively. The same arguments developed in Section~\ref{sec: exp-zero} with the use of alien derivation apply here as well. In particular, we find that
\be \label{eq: splus2infty}
s_{+}(\psi)(\tau) = s_{-}(\psi)(\tau) + s_{-} (\psi)(\tau) \sum_{k=1}^{\infty} \re^{-\zeta_k/\tau} \bar{R}_k  \, ,
\ee
where the asymptotic series $\psi_k(\tau)$, which resurges from $\psi(\tau)$ at the singularity $\zeta_k$, is simply
\be \label{eq: psikinfty}
\psi_k(\tau) = \psi(\tau) \, , \quad k \in \IZ_{>0} \, ,
\ee
and the Stokes constant $\bar{R}_k \in \IC$ of $\psi(\tau)$ at the singularity $\zeta_k$ is fully determined by the Stokes constants of $\phi(\tau)$ via the closed combinatorial formula
\be \label{eq: newStokesP2infty}
\bar{R}_k =  \sum_{p \in \CP(k)} \frac{1}{r!} \binom{r}{N_1, \dots , N_k} R_{n_1} \cdots R_{n_r}  \, , \quad k \in \IZ_{>0} \, ,
\ee
where $\CP(k)$ is the set of all partitions $p = (n_1, \dots , n_r)$ of the positive integer $k$, $r=|p|$ denotes the length of the partition, and $N_i \in \IN$ is the number of times that the positive integer $i \in \IZ_{> 0}$ is repeated in the partition $p$. Note that $\sum_{i=1}^k N_i = r$.
We stress that the sum over partitions in Eq.~\eqref{eq: newStokesP2infty} is finite, and thus all the Stokes constants of the original perturbative series $\psi(\tau)$ are known exactly. More precisely, the discontinuity formula in Eq.~\eqref{eq: splus2infty} solves analytically the resurgent structure of $\psi(\tau)$. Applying the same computations above to the discontinuity of $\psi(\tau)$ across the angle $3 \pi/2$, and recalling that $R_n = -R_{-n}$ for all $n \in \IZ_{\ne 0}$, we find straightforwardly that
\be \label{eq: newStokesP2infty-minus}
\bar{R}_k =  \sum_{p \in \CP(-k)} \frac{(-1)^r}{r!} \binom{r}{N_1, \dots , N_k} R_{n_1} \cdots R_{n_r}  \, , \quad k \in \IZ_{<0} \, ,
\ee
which implies that $\bar{R}_{k} \ne \pm \bar{R}_{-k}$ in general. 

Let us point out that the formulae in Eqs.~\eqref{eq: newStokesP2infty} and~\eqref{eq: newStokesP2infty-minus} immediately prove that $\bar{R}_k \in \IQ$, $k \in \IZ_{\ne 0}$. However, we can say more. The discontinuity formula in Eq.~\eqref{eq: discP2series2infty} can be directly exponentiated to give an exact generating function in terms of known $q$-series for the Stokes constants $\bar{R}_k$. Namely, we find that
\be \label{eq: quotient-infty}
\sum_{k=1}^{\infty} \bar{R}_k q^{k/3} = \frac{(q^{2/3} ; \, q)_{\infty}^3}{(q^{1/3} ; \, q)_{\infty} ^3} \, , \quad \sum_{k=1}^{\infty} \bar{R}_{-k} q^{k/3} = \frac{(q^{1/3} ; \, q)_{\infty}^3}{(q^{2/3} ; \, q)_{\infty} ^3} \, .
\ee
As a consequence of the $q$-binomial theorem, the quotients of $q$-series in Eq.~\eqref{eq: quotient-infty} can be expanded in powers of $q^{1/3}$, and the resulting numerical coefficients possess a natural interpretation in terms of counting partitions. In particular, they are integer numbers.  
Explicitly, the first several Stokes constants $\bar{R}_k$, $k>0$, are 
\be
3, 3, 1, 3, 6, 0, -3, 9, 9, -9, 0, 19, -6, -15, 27, 12, -33, 17, 33,  \dots \, ,
\ee
while the first several Stokes constants $\bar{R}_k$, $k<0$, are
\be
-3, 6, -10, 12, -9, 1, 9, -15, 8, 15, -42, 54, -36, -15, 73, -90, 39, 62, -153,  \dots \, .
\ee
We comment that our exact solution is in full agreement with the numerical investigation of~\cite{GuM}.

\subsection{A number-theoretic duality} \label{sec: dirichlet}
We will now further develop the simple arithmetic symmetry observed in the closed formulae in Sections~\ref{sec: closed-zero} and~\ref{sec: closed-infty}, and we will show how it can be reformulated into a full-fledged analytic number-theoretic duality, shedding new light on the statements of Sections~\ref{sec: zeta-zero} and~\ref{sec: zeta-infty}. We refer to~\cite{NT1, NT2, NT3} for background material on analytic number theory.
Let us recall that the Stokes constants associated to the limits $\hbar \rightarrow 0$ and $\hbar \rightarrow \infty$ are given by the explicit divisor sum functions in Eqs.~\eqref{eq: formulaS1} and~\eqref{eq: formulaS1infty}, respectively. We can write these formulae equivalently as
\begin{subequations} \label{eq: convolution}
\begin{align}
\frac{S_n}{S_1} &= \sum_{d | n} \chi_{3,2} \left( d \right) F_{-1}\left( d \right) F_0\left( \frac{n}{d} \right) = \left(\chi_{3,2} F_{-1} * F_0 \right)(n) \, , \quad n \in \IN_{\ne 0} \, ,\\
\frac{R_n}{R_1} &= \sum_{d | n} \chi_{3,2} \left( d \right) F_0\left( d \right)  F_{-1}\left( \frac{n}{d} \right) = \left(\chi_{3,2} F_0 * F_{-1} \right)(n) \, , \quad n \in \IN_{\ne 0} \, ,
\end{align}
\end{subequations}
where the product $*$ denotes the Dirichlet convolution of arithmetic functions. We have introduced $F_{\alpha}(n) = n^{\alpha}$, $\alpha \in \IR$, and $\chi_{3,2}(n)$ is the unique non-principal Dirichlet character modulo $3$, which is defined by
\be \label{eq: character32}
\chi_{3,2} (n) = 
\begin{cases}
1 & \quad \text{if} \quad n \equiv 1 \quad \text{mod} \; 3 \, , \\
0 & \quad \text{if} \quad n \equiv 0 \quad \text{mod} \; 3 \, , \\
-1 & \quad \text{if} \quad n \equiv -1 \quad \text{mod} \; 3 \, ,
\end{cases}
\ee
for $n \in \IN$. Note that, despite the arithmetic functions $\chi_{3,2}, F_0, F_{-1}$ being totally multiplicative, the convolutions $S_n/S_1, R_n/R_1$ are multiplicative only.

We have shown that the Dirichlet series associated to the Stokes constants naturally appear in the exact large-order formulae for the perturbative coefficients in Eqs.~\eqref{eq: formulaS3} and~\eqref{eq: formulaS3infty}, and that they are, in particular, $L$-series. We can say something more. Since the multiplication of Dirichlet series is compatible with the Dirichlet convolution~\cite{Dirichlet}, it follows directly from the decomposition in Eq.~\eqref{eq: convolution} that we have the formal factorization
\begin{subequations} \label{eq: convolution2}
\begin{align}
\sum_{m=1}^{\infty} \frac{S_m/S_1}{m^{s}} &= \sum_{m=1}^{\infty} \frac{\chi_{3,2}(m) F_{-1}(m)}{m^{s}} \sum_{m=1}^{\infty} \frac{F_0(m)}{m^{s}} =  L(s+1, \chi_{3,2}) \zeta(s) \, , \\
\sum_{m=1}^{\infty} \frac{R_m/R_1}{m^{s}} &= \sum_{m=1}^{\infty} \frac{\chi_{3,2}(m) F_0(m)}{m^{s}} \sum_{m=1}^{\infty} \frac{F_{-1}(m)}{m^{s}} = L(s, \chi_{3,2}) \zeta(s+1) \, ,
\end{align}
\end{subequations}
where $s \in \IC$ such that $\Re (s) >1$, and $L(s, \chi_{3,2})$ is the Dirichlet $L$-series of the primitive character $\chi_{3,2}$. We have found, in this way, that the arithmetic duality which relates the weak- and strong-coupling Stokes constants $S_n, R_n$, $n \in \IN_{\ne 0}$, in Eq.~\eqref{eq: convolution} is translated at the level of the $L$-series encoded in the perturbative coefficients into a simple unitary shift of the arguments of the factors in the RHS of Eq.~\eqref{eq: convolution2}. 
Furthermore, $L(s, \chi_{3,2})$ is absolutely convergent for $\Re (s) >1$, and it can be analytically continued to a meromorphic function on the whole complex $s$-plane, called a Dirichlet $L$-function. 
Each of the two $L$-series in the LHS of Eq.~\eqref{eq: convolution2} is, therefore, the product of two well-known $L$-functions, and such a remarkable factorization explicitly proves the convergence in the right half-plane of the complex numbers $\Re (s) >1$ and the existence of a meromorphic continuation throughout the complex $s$-plane. Namely, our $L$-series are, themselves, $L$-functions.
Let us comment that the formulae in Eqs.~\eqref{eq: formulaS3} and ~\eqref{eq: formulaS3infty} follow from Eq.~\eqref{eq: convolution2} by means of the known relation between Dirichlet $L$-functions and Hurwitz zeta functions at rational values. Specifically, we have that
\be
L(s, \chi_{3,2}) = \frac{1}{3^s} \left(\zeta\left(s,\frac{1}{3}\right) - \zeta\left(s, \frac{2}{3} \right) \right) \, , \quad \Re (s) >1 \, .
\ee
Moreover, the Dirichlet $L$-function satisfies the Euler product expansion
\be
L(s, \chi_{3,2}) = \prod_{p \in \IP} \left( 1- \frac{\chi_{3,2}(p)}{p^s} \right)^{-1} \, , \quad \Re (s) >1 \, ,
\ee
where $\IP$ is the set of prime numbers.

\subsection{Numerical tests} \label{sec: numericsP2}
Let us conclude with a parallel and independent numerical analysis which cross-checks and confirms our analytic results on the resurgent structure of the asymptotic series $\psi(\hbar)$ in Eq.~\eqref{eq: psiP2}.
We recall that this corresponds, up to a global pre-factor, to the perturbative expansion of the first fermionic spectral trace of local $\IP^2$ in the semiclassical limit $\hbar \rightarrow 0$ in Eq.~\eqref{eq: expP2}.
An analogous numerical investigation of the asymptotic series in Eq.~\eqref{eq: psiP2infty} in the dual weakly-coupled limit $g_s \rightarrow 0$ is performed in~\cite{GuM}. 
Let us truncate the sum to a very high but finite order $d \gg 1$. We denote the resulting $\IQ$-polynomial by $\psi_d(\hbar)$ and its Borel transform by $\hat{\psi}_d(\zeta)$. Namely,
\be \label{eq: psi_d}
\psi_d(\hbar) = \sum_{n=0}^d b_{2n} \hbar^{2n} \in \IQ [\hbar] \, , \quad \hat{\psi}_d(\zeta) = \sum_{n=0}^d \frac{b_{2n}}{(2n)!} \zeta^{2n} \in \IQ [\zeta] \, ,
\ee
where the coefficients $b_{2n}$, $1 \le n \le d$, are computed by Taylor expanding the exponential in the RHS of Eq.~\eqref{eq: psiP2}. The first few terms of $\psi_d(\hbar)$ are shown in Eq.~\eqref{eq: first-terms-P2-zero}. It is straightforward to verify that the perturbative coefficients satisfy the expected factorial growth
\be
b_{2n} \sim (-1)^n (2n)! \left(\frac{4 \pi^2}{3} \right)^{-2n} \quad n \gg 1 \, .
\ee
We assume $\hbar \in \IC'$ and perform a full numerical Pad\'e--Borel analysis~\cite{CD1, CD2, CD3} in the complex $\zeta$-plane. Let $d$ be even. We compute the singular points of the diagonal Pad\'e approximant of order $d/2$ of the truncated Borel expansion $\hat{\psi}_d(\zeta)$, which we denote by $\hat{\psi}^{\rm PB}_d(\zeta)$, and we observe two dominant complex conjugate branch points at $\zeta = \pm 4 \pi^2 \ri /3$ and their respective arcs of accumulating spurious poles mimicking two branch cuts along the positive and negative imaginary axis. Let us introduce $A= 4 \pi^2/3$ and $\zeta_n = n 4 \pi^2 \ri/3$, $n \in \IZ_{\ne 0}$, as before. A zoom-in to the poles of the Pad\'e approximant in the bottom part of the upper-half complex $\zeta$-plane is shown in the leftmost plot in Fig.~\ref{fig: PB-P2}.
\begin{figure}[htb!]
\centering
\includegraphics[width=1.\textwidth]{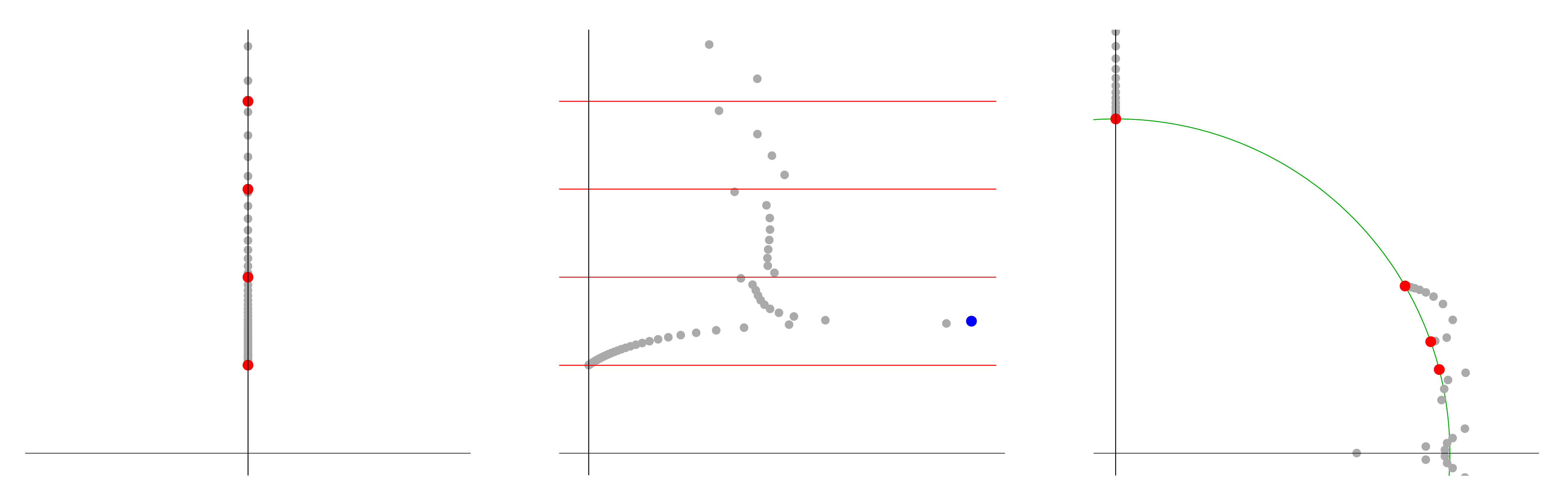}
\caption{In the leftmost plot, a zoom-in to the poles of $\hat{\psi}^{\rm PB}_d(\zeta)$ (in gray) in the bottom part of the upper-half complex $\zeta$-plane. We show the first few integer multiples of the dominant pole (in red). In the central plot, a zoom-in to the poles of $\hat{\psi}^{\rm PB}_{d, {\rm probe}}(\zeta)$ (in gray) in the bottom part of the upper-half complex $\zeta$-plane, including the test charge singularity at $\zeta = \zeta_{\text{probe}}$ (in blue). We show the horizontal lines intersecting the positive imaginary axis at the first few physical branch points (in red). In the rightmost plot, a zoom-in to the poles of $\hat{\psi}^{\rm PCB}_{d}(\eta)$ (in gray) in the upper-right quarter of the complex $\eta$-plane. We show the first few physical branch points (in red) and the unit circle (in green). The plots are obtained with $d=300$.}
\label{fig: PB-P2}
\end{figure}
To reveal the presence of subdominant singularities which are hidden by the unphysical Pad\'e poles, we apply the potential theory interpretation of Pad\'e approximation~\cite{Stahl, Saff}. More concretely, to test the presence of the suspected next-to-leading order branch point at $\zeta = \zeta_2$, we introduce by hand the singular term $\hat{\psi}_{\text{probe}}(\zeta) = \left( \zeta - \zeta_{\text{probe}} \right)^{-1/5}$, where $\zeta_{\text{probe}} = A  \left( 1/14 + 3 \ri/2 \right)$, and we compute the $(d/2)$-diagonal Pad\'e approximant of the sum
\be
\hat{\psi}^{\rm PB}_{d, {\rm probe}}(\zeta) = (\hat{\psi}_d + \hat{\psi}_{\text{probe}})^{\rm PB} (\zeta)\, . 
\ee
The resulting Pad\'e--Borel poles distribution is distorted as shown in the central plot in Fig.~\ref{fig: PB-P2}. The targeted true branch point at $\zeta =  \zeta_2$ is now clearly visible, and so it is the test charge singularity at $\zeta = \zeta_{\text{probe}}$. The other true branch points at $\zeta = \zeta_3, \zeta_4$ are outlined. We can further improve the precision of the Pad\'e extrapolation and analytic continuation of the truncated Borel series $\hat{\psi}_d(\zeta)$ with the use of conformal maps. Namely, we perform the change of variable
\be
\zeta = A \frac{2 \eta}{1-\eta^2} \, , \quad \eta \in \IC \, ,
\ee
which maps the cut Borel $\zeta$-plane $\mathbb{C} \backslash \{(-\infty, \zeta_{-1}] \cup  [\zeta_1, +\infty) \}$ into the interior of the unit disk $|\eta| < 1$. 
The dominant branch points $\zeta= \zeta_{\pm 1}$ are mapped into $\eta = \pm \text{i}$, while the point at infinity is mapped into $\eta = \pm 1$. Correspondingly, the branch cuts $(-\infty, \zeta_{-1}]$ and $[\zeta_1, +\infty)$ along the imaginary $\zeta$-axis split each one into two identical copies which lie onto the two lower-half and upper-half quarters of the unit circle in the $\eta$-plane, respectively. The inverse conformal map is explicitly given by
\be
\eta = \pm \sqrt{\frac{-1 + \sqrt{1+(\zeta/A)^2}}{1 + \sqrt{1+(\zeta/A)^2}}} \,  , \quad \zeta \in \IC \, .
\ee
We compute the $(d/2)$-diagonal Pad\'e approximant of the conformally mapped Borel expansion, which we denote by $\hat{\psi}^{\rm PCB}_{d}(\eta)$. Its singularities in the complex $\eta$-plane are shown in the rightmost plot in Fig.~\ref{fig: PB-P2}.
We observe two symmetric arcs of spurious poles emanating from the conformal map images of $\zeta = \zeta_{\pm 1}$ along the imaginary axis in opposite directions. The smaller arcs of poles jumping along the unit circle towards the real axis represent the Pad\'e boundary of convergence joining the conformal map images of the repeated singularities at $\zeta = \zeta_n$, $n \in \IZ_{\ne 0}$. 
Thus, the numerical analysis confirms the pattern of singularities previously found analytically.

Let us now test the expected resurgent structure of the asymptotic series $\psi(\hbar)$. We recall that the local expansions of the Borel transform $\hat{\psi}(\zeta)$ in the neighbourhoods of its dominant singularities at $\zeta = \zeta_{\pm 1}$ are governed by the one-instanton perturbative corrections $\psi_{\pm 1}(\hbar)$ and the first Stokes constants $\bar{S}_{\pm 1}$. Let us consider the standard functional ansatz 
\be \label{eq: ansatzP2}
\psi_1(\hbar) = \sum_{k=0}^{\infty} c_k \hbar^{k-b}  \in \hbar^{-b} \IC [\![\hbar]\!] \, ,
\ee
where the coefficient $c_k \in \IC$ can be interpreted as the $(k+1)$-loop contribution around the one-instanton configuration in the upper-half Borel $\zeta$-plane, and $b \in \IR \backslash \IZ_+$ is the so-called characteristic exponent. 
We fix the normalization condition $c_0=1$, and we further assume that $\psi_{-1}(\hbar)=\psi_1(\hbar)$ and $\bar{S}_{-1} = \bar{S}_1$. By Cauchy's integral theorem, the large-$n$ asymptotics of the perturbative coefficients $b_{2n}$ in Eq.~\eqref{eq: psi_d} is controlled at leading order by 
\be \label{eq: large_order_standard} 
b_{2n} \sim \frac{(-1)^n  \bar{S}_1}{\pi \ri } \frac{\Gamma(2n+b)}{A^{2n+b}} \sum_{k=0}^{\infty} \frac{c_k A^k}{\prod_{j=1}^k (2n+b-j)} \quad n \gg 1 \, .
\ee
An independent numerical estimate of $A$ is obtained from the convergence of the sequences
\begin{subequations} 
\begin{align}
4 n^2 \frac{b_{2n}}{b_{2n+2}} &= |A|^2 + \CO \left(1/n\right) \, , \\
\frac{|A|^2 b_{2n+2}}{4 n^2 b_{2n}} + \frac{4 n^2 b_{2n-2}}{|A|^2 b_{2n}} &= 2 \cos(2 \theta_A) + \mathcal{O}\left(1/n\right) \, ,
\end{align}
\end{subequations}
which give the absolute value $|A| \approx 4 \pi^2/3$ and the phase $\theta_A \approx 0$, as expected. Analogously, a numerical estimate of the characteristic exponent $b$ is obtained from the convergence of the sequence
\be
\left( \frac{A^2}{4 n^2} \frac{b_{2n+2}}{b_{2n}} - 1 \right) 2n = 1+ 2b + \mathcal{O}\left(1/n\right) \, ,
\ee
which gives $b \approx 0$. Finally, we estimate the Stokes constant $\bar{S}_1$ as the large-$n$ limit of the sequence\footnote{Using the first 300 perturbative coefficients, the numerical estimate for the first Stokes constant agrees with the exact value up to 32 digits.} 
\begin{equation}
\pi \ri (-1)^n \frac{A^{2n}}{\Gamma(2n)}  b_{2n} = \bar{S}_1 + \mathcal{O}\left(1/n\right) \, ,
\end{equation}
which gives $\bar{S}_1 \approx 3 \sqrt{3} \ri$.
Let us proceed to systematically extract the coefficients $c_k$, $k \in \IN$. We first Taylor expand the quotients appearing in the RHS of Eq.~\eqref{eq: large_order_standard} in the large-$n$ limit and rearrange them to give
\be \label{eq: large_order_standard_2}
b_{2n} \sim \frac{(-1)^n  \bar{S}_1}{\pi \ri } \frac{\Gamma(2n)}{A^{2n}} \sum_{i=0}^{\infty} \frac{\mu_i}{(2n)^i} \quad n \gg 1 \, ,
\ee
where the new coefficients $\mu_ i$ are expressed in closed form as
\be \label{eq: coeff-change}
\mu_0 = c_0 = 1 \, , \quad \mu_i = \sum_{k=1}^{i} c_k A^k \left( \sum_{\substack{1 \le m_1, \dots, m_k \le i \\ m_1+ \dots +m_k = i}} \, \prod_{j=1}^k j^{m_j-1} \right) \, , \quad i \in \IN_{\ne 0} \, .
\ee
We define the sequence $Q_{2n}$, $n \in \IN$, such that
\begin{equation}
Q_{2n} = \pi \ri (-1)^n \frac{A^{2n}}{\bar{S}_1 \Gamma(2n)}  b_{2n} \sim \sum_{i = 0}^{\infty} \frac{\mu_i}{(2n)^i} \quad n \gg 1 \, ,
\end{equation}
and we obtain a numerical estimate of the coefficients $\mu_i$, $i \in \IN$, as the large-$n$ limits of the recursively-defined sequences
\begin{subequations}
\begin{align}
Q_{2n}^{(1)} &= 2n (Q_{2n}-1) =  \mu_1+ \mathcal{O}\left(1/n\right) \, , \\
Q_{2n}^{(i)} &= 2n (Q_{2n}^{(i-1)}-\mu_{i-1}) = \mu_i + \mathcal{O}\left(1/n\right) \, , \quad i \in \IN_{>0} \, .
\end{align}
\end{subequations}
We substitute the numerical values for the coefficients $\mu_i$ in their explicit relations with the coefficients $c_k$ in Eq.~\eqref{eq: coeff-change}, and term-by-term we find in this way that 
\begin{equation}
c_{2k} \approx b_{2k} \, , \quad c_{2k+1} \approx 0 \, , \quad k \in \IN \, ,
\end{equation}
that is, the coefficients of the one-instanton asymptotic series $\psi_1(\hbar)$ in Eq.~\eqref{eq: ansatzP2} identically correspond to the coefficients of the original perturbative series $\psi(\hbar)$ in Eq.~\eqref{eq: psi_d}. We conclude that $\psi_1(\hbar) = \psi(\hbar)$, as expected from the analytic solution. We comment that the numerical convergence of the large-$n$ limits of all sequences above has been accelerated using Richardson transforms.

Finally, we perform a numerical test of the discontinuity formula in Eq.~\eqref{eq: splus2}. More precisely, let us rotate the Borel $\zeta$-plane by an angle of $-\pi/2$ in order to move the branch cuts of $\hat{\psi}(\zeta)$ to the real axis. The corresponding change of variable is $z = - \ri \zeta$. We analogously rotate the complex $\hbar$-plane and introduce the variable $x = - \ri \hbar$. We fix a small positive angle $\epsilon \ll 1$ and a small positive value of $x \ll 1$. We compute numerically the lateral Borel resummations across the positive real axis as\footnote{Following~\cite{CD1, CD2}, the numerical precision of the lateral Borel resummations can be improved with the use of conformal maps.}
\begin{equation}
s^{\rm PB}_{\pm}(\psi)(x) = e^{\pm \text{i} \epsilon} \int_0^{\infty} \hat{\psi}_d^{\text{PB}} (z e^{\pm \text{i} \epsilon} x)  \, \re^{-z e^{\pm \text{i} \epsilon}} \rd z \, ,
\end{equation}
where $\hat{\psi}_d^{\text{PB}}(z)$ is the diagonal Pad\'e approximant of order $d/2$ of the truncated Borel series $\hat{\psi}_d(z)$. The corresponding discontinuity is evaluated as
\begin{equation}
\text{disc}^{\rm PB}(\psi)(x) = s^{\rm PB}_{+}(\psi)(x) - s^{\rm PB}_{-}(\psi)(x) = 2 \ri \Im \left( s^{\rm PB}_{+}(\psi)(x) \right) \, .
\end{equation}
Assuming that all higher-order perturbative series $\psi_n(x)$ are trivially equal to $\psi(x)$, the Stokes constants $\bar{S}_n$, $n \in \IZ_{>0}$, are estimated numerically by means of the recursive relation
\begin{equation}
\text{disc}^{\rm PB}(\psi)(x) \, e^{n A/x} - s^{\rm PB}_{-}(\psi)(x) \sum_{k = 1}^{n-1} \bar{S}_k e^{(n-k)A/x} = s^{\rm PB}_{-}(\psi)(x) \bar{S}_n + \mathcal{O}(e^{-A/x}) \, ,
\end{equation}
which reproduces the exact results obtained in Section~\ref{sec: exp-zero}. Analogously, we obtain a numerical estimate of the Stokes constants associated to the branch points on the negative real axis in the rotated Borel $z$-plane.
We remark that all numerical checks described here for $\psi(\hbar)$ are straightforwardly and successfully applied to the asymptotic series $\phi(\hbar) = \log \psi(\hbar)$ in Eq.~\eqref{eq: phi2}, confirming once more our analytic solution. 

\sectiono{The example of local \texorpdfstring{$\IF_0$}{F0}} \label{sec: localF0}
The total space of the canonical bundle over the Hirzebruch surface $\IF_0 = \IP^1 \times \IP^1$, which is $\CO(-2, -2) \rightarrow \IP^1 \times \IP^1$, called the local $\IF_0$ geometry, has one complex deformation parameter $\kappa$ and one mass parameter $\xi_{\IF_0}$. 
Its moduli space is identified with the family of mirror curves described by the equation
\be \label{eq: mcF0}
\re^x + \xi_{\IF_0} \re^{-x} + \re^{y} + \re^{-y} + \kappa = 0 \, , \quad x,y \in \IC \, .
\ee
For simplicity, we will impose the condition\footnote{We remark that the TS/ST correspondence is expected to hold for arbitrary values of the mass parameters, as suggested by the evidence provided in~\cite{CGuM}.} $\xi_{\IF_0}=1$, which implies the parametrization $z=\frac{1}{\kappa^2}$.
The large radius point, the maximal conifold point, and the orbifold point of the moduli space of local $\IF_0$ correspond to $z=0$, $z=1/16$, and $z= \infty$, respectively.
The quantization of the mirror curve in Eq.~\eqref{eq: mcF0}, under the assumption $\xi_{\IF_0} = 1$, gives the quantum operator
\be \label{eq: opF0}
\mO_{\IF_0}(\mx, \my) = \re^{\mx} + \re^{-\mx} + \re^{\my} + \re^{-\my} \, ,
\ee
acting on $L^2(\IR)$, where $\mx, \, \my$ are self-adjoint Heisenberg operators satisfying $[ \mx , \, \my ] = \text{i} \hbar$.
It was proven in~\cite{KM} that the inverse operator 
\be
\rho_{\IF_0} = \mO_{\IF_0}^{-1}
\ee
is positive-definite and of trace class. The fermionic spectral traces of $\rho_{\IF_0}$ are well-defined and can be computed explicitly~\cite{KMZ}.
In this section, we will study the resurgent structure of the first fermionic spectral trace
\be
Z_{\IF_0}(1, \hbar) = \text{Tr}(\rho_{\IF_0})
\ee
in the semiclassical limit $\hbar \rightarrow 0$.

\subsection{Computing the perturbative series} \label{sec: WignerF0}
As we have done for local $\IP^2$ in Section~\ref{sec: WignerP2}, we apply the phase-space formulation of quantum mechanics to obtain the WKB expansion of the trace of the inverse operator $\rho_{\IF_0} $ at NLO in $\hbar \rightarrow 0$, starting from the explicit expression of the operator $\mO_{\IF_0}$ in Eq.~\eqref{eq: opF0}, and following Appendix~\ref{app: Wigner}.
For simplicity, we denote by $O_W, \, \rho_W$ the Wigner transforms of the operators $\mO_{\IF_0}, \, \rho_{\IF_0}$, respectively. 
The Wigner transform of $\mO_{\IF_0}$ is obtained by performing the integration in Eq.~\eqref{eq: OW}. As we show in Example~\ref{ex: exampleWigner}, this simply gives the classical function
\be
O_{\rm W} =  \re^{x} + \re^{-x} + \re^{y} + \re^{-y} \, .
\ee
Substituting it into Eqs.~\eqref{eq: G2} and~\eqref{eq: G3}, we have
\begin{subequations}
\begin{align}
\CG_2 &= -\frac{\hbar^2}{4} \left[ \re^{x+y} + \re^{-x+y} + \re^{x-y} + \re^{-x-y} \right] + \CO(\hbar^4) \, , \\
\CG_3 &= -\frac{\hbar^2}{4} \left[ \re^{2x+y} + \re^{2x-y} + \re^{-2x+y} + \re^{-2x-y} - 2 \re^{x} -2 \re^{-x} + \left( x \leftrightarrow y \right) \right] + \CO(\hbar^4) \, , 
\end{align}
\end{subequations}
where $\left( x \leftrightarrow y \right)$ indicates the symmetric expression after exchanging the variables $x$ and $y$.
It follows from Eq.~\eqref{eq: rhoW2} that the Wigner transform of $\rho_{\IF_0}$, up to order $\hbar^2$, is then given by
\be \label{eq: rhoW2_F0}
\rho_{\rm W} = \frac{1}{O_{\rm W}} - \hbar^2 \frac{1}{O_{\rm W}^3} + \CO(\hbar^4) \, .
\ee
We note that the same result can be obtained by solving Eq.~\eqref{eq: cosL2} order by order in powers of $\hbar^2$.
Integrating Eq.~\eqref{eq: rhoW2_F0} over phase space, as in Eq.~\eqref{eq: tracePS}, we obtain the NLO perturbative expansion in $\hbar$ of the trace, that is, 
\be
\text{Tr}(\rho_{\IF_0}) = \frac{1}{2 \pi \hbar} \int_{\IR^2} \rho_{\rm W} \, \rd x \rd y =  \frac{1}{2 \pi \hbar} \int_{\IR^2} \frac{1}{O_{\rm W}} \, \rd x \rd y - \frac{\hbar}{2 \pi} \int_{\IR^2} \frac{1}{O_{\rm W}^3} \, \rd x \rd y + \CO(\hbar^4)  \, ,
\ee
and evaluating the integrals explicitly, we find
\be \label{eq: WWfinalF0}
\text{Tr}(\rho_{\IF_0}) = \frac{\pi}{4 \hbar} \left\{ 1 - \frac{\hbar^2}{64} + \CO(\hbar^4) \right\} \, .
\ee
We stress that the phase-space formalism adopted above provides, in principle, the perturbative expansion of $\text{Tr}(\rho_{\IF_0})$ at all orders in $\hbar$ by systematically extending all intermediate computations beyond order $\hbar^2$. However, as for the case of local $\IP^2$, a more efficient way to extract the perturbative coefficients is described below.

The integral kernel for the operator $\rho_{\IF_0}$ is given by~\cite{KMZ}
\be \label{eq: F0kernel}
\rho_{\IF_0}(x_1, \, x_2) = \frac{\re^{\pi \mb (x_1+x_2)/2}}{2 \mb \cosh(\pi (x_1-x_2)/\mb)} \frac{\Phi_{\mb} (x_1 + \text{i} \mb/4) }{\Phi_{\mb} (x_1 - \text{i} \mb/4)} \frac{\Phi_{\mb} (x_2 + \text{i} \mb/4)}{\Phi_{\mb} (x_2 - \text{i} \mb/4)} \, ,
\ee
where $\mb$ is related to $\hbar$ by 
\be
\pi \mb^2 =\hbar \, , 
\ee
and $\Phi_{\mb}$ is Faddeev's quantum dilogarithm. A summary of the properties of this function is provided in Appendix~\ref{app: Faddeev}.
As in the case of local $\IP^2$, the integral kernel in Eq.~\eqref{eq: F0kernel} can be analytically continued to $\hbar \in \IC'$. 
The first spectral trace has the integral representation~\cite{KMZ}
\be \label{eq: F0int}
\text{Tr}(\rho_{\IF_0}) = \frac{1}{2 \mb} \int_{\IR} \re^{\pi \mb x} \frac{\Phi_{\mb} (x + \text{i} \mb/4)^2}{\Phi_{\mb} (x - \text{i} \mb/4)^2} \, \rd x \, ,
\ee
which is a well-defined, analytic function of $\hbar \in \IC'$, under the assumption that $\Re (\mb) >0$.
The integral in Eq.~\eqref{eq: F0int} can be evaluated explicitly by analytically continuing $x$ to the complex domain, closing the integration contour from above, and applying Cauchy's residue theorem. The resulting expression for the first spectral trace of local $\IF_0$ is the sum of products of holomorphic and anti-holomorphic blocks given by $q$- and $\tilde{q}$-series, respectively. Namely, we have that~\cite{GuM}
\be \label{eq: F0fact}
\text{Tr}(\rho_{\IF_0}) = - \frac{\text{i}}{2} \left( G(q) \tilde{g}(\tilde{q}) + 8 \mb^{-2} g(q) \tilde{G}(\tilde{q}) \right) \, ,
\ee
where $q = \re^{2 \pi \text{i} \mb^2}$ and $\tilde{q} = \re^{- 2 \pi \text{i} \mb^{-2}}$. Note that the factorization in Eq.~\eqref{eq: F0fact} is not symmetric in $q, \, \tilde{q}$, and the holomorphic and anti-holomorphic blocks are given by different series. 
More precisely, they are
\begin{subequations} \label{eq: seriesF0}
\begin{align}
g(q) &= \sum_{m=0}^{\infty} \frac{(q^{1/2} ; \, q)_m^2}{(q ; \, q)_m^2} q^{m/2} = {}_2\phi_1\left(
\begin{matrix}
q^{1/2} ,  & q^{1/2}\\
  & q
\end{matrix}
; \, q , \, q^{1/2} \right) \, ,  \label{eq: seriesF0-g} \\
G(q) &= \sum_{m=0}^{\infty} \frac{(q^{1/2} ; \, q)_m^2}{(q ; \, q)_m^2} q^{m/2} \left( 1 + 4 \sum_{s=1}^{\infty} \frac{q^{s(m+1/2)}}{1+q^{s/2}} \right) \, , \label{eq: seriesF0-G} \\
\tilde{g}(\tilde{q}) &= \frac{1}{2} \sum_{m=0}^{\infty} \frac{(-1 ; \, \tilde{q})_m^2}{(\tilde{q} ; \, \tilde{q})_m^2} (-\tilde{q})^m = \frac{1}{2} {}_2\phi_1\left(
\begin{matrix}
-1 ,  & -1\\
  & \tilde{q}
\end{matrix}
; \, \tilde{q} , \, -\tilde{q} \right) \, , \label{eq: seriesF0-gtilde} \\
\tilde{G}(\tilde{q}) &=  \sum_{m=0}^{\infty} \frac{(-\tilde{q} ; \, \tilde{q})_m^2}{(\tilde{q} ; \, \tilde{q})_m^2} (-1)^m \left( \sum_{s=0}^{\infty} \frac{\tilde{q}^{(2s+1)(m+1)}}{1-\tilde{q}^{2s+1}} \right) \, , \label{eq: seriesF0-Gtilde}
\end{align}
\end{subequations}
where $(x ; \, q)_m$ is the $q$-shifted factorial defined in Eq.~\eqref{eq: qFactor}, and ${}_{r+1}\phi_s$ is the $q$-hypergeometric series defined in Eq.~\eqref{eq: qHypergeo}.
As in the case of local $\IP^2$, we will assume that $\Im(\mb^2) > 0$, so that $|q|, |\tilde{q}| < 1$, and the $q$- and $\tilde{q}$-series converge.

Let us consider the integral representation in Eq.~\eqref{eq: F0int} and derive its all-orders perturbative expansion in the limit $\hbar \rightarrow 0$. We perform the change of variable $y = 2 \pi \mb x$ and write it equivalently as
\be \label{eq: F0int2}
\text{Tr}(\rho_{\IF_0}) = \frac{1}{4 \pi \mb^2} \int_{\mathbb{R}} \exp \left( \frac{y}{2} + 2 \log \Phi_{\mb} \left( \frac{y + \ri \pi \mb^2/2}{2 \pi \mb} \right) - 2 \log \Phi_{\mb} \left( \frac{y - \ri \pi \mb^2/2}{2 \pi \mb} \right) \right)  \, \rd y \, .
\ee
The asymptotic expansion formula in Eq.~\eqref{eq: logPhib} for $\log(\Phi_{\mb})$ in the limit $\mb \rightarrow 0$ yields
\be \label{eq: logPhibF0}
\log \Phi_{\mb}\left( \frac{y \pm \ri \pi \mb^2/2}{2 \pi \mb} \right) = \sum_{k=0}^{\infty} (2 \pi \ri \mb^2)^{2k-1} \frac{B_{2k}(1/2)}{(2k)!} \text{Li}_{2-2k}(-\re^{y \pm \ri \pi \mb^2/2}) \, ,
\ee
where $B_n(z)$ is the $n$-th Bernoulli polynomial, and $\text{Li}_n(z)$ is the polylogarithm of order $n$. We eliminate the remaining $\mb$-dependence of the polylogarithms in Eq.~\eqref{eq: logPhibF0} by expanding them in turn around $\mb \rightarrow 0$. 
More precisely, we recall that the derivative of the polylogarithm function is
\be \label{eq: polyderiv}
\frac{\d \text{Li}_{s}(\re^{\mu})}{\d \mu} = \text{Li}_{s-1}(\re^{\mu})  \, , \quad s, \mu \in \IC \, ,
\ee
and we derive the Taylor expansion
\be \label{eq: polyderivF0}
\text{Li}_{2-2k}(-\re^{y \pm \ri \pi \mb^2/2}) = \sum_{m=0}^{\infty} \frac{1}{m!} \left( \pm \frac{ \ri \pi \mb^2}{2} \right)^m \text{Li}_{2-2k-m}(-\re^{y}) \, ,
\ee
for all $k \ge 0$. Let us denote the exponent of the integrand in Eq.~\eqref{eq: F0int2} by $V(y, \mb)$.
After recombining the terms in the nested expansions in Eqs.~\eqref{eq: logPhibF0} and~\eqref{eq: polyderivF0}, we obtain a well-defined, fully-determined perturbation series in $\mb^2$, which is
\be
V(y, \mb) = \frac{y}{2} + 2 \sum_{k,m=0}^{\infty} (2 \pi \ri \mb^2)^{2k+m-1} \frac{B_{2k}(1/2)}{4^m m! (2k)!} \text{Li}_{2-2k-m}(-\re^{y}) [1-(-1)^m] \, .
\ee
We note that the factor $ [1-(-1)^m]$ is zero for even values of $m$. Therefore, introducing the notation $m=2q+1$ and $p = k+q$, we have
\be \label{eq: Vinter}
V(y, \mb) = \frac{y}{2} +  \sum_{p=0}^{\infty} (2 \pi \text{i} \mb^2 )^{2p} \text{Li}_{1-2p}(-e^y) \sum_{q=0}^p \frac{B_{2p-2q}(1/2)}{4^{2q} (2q+1)! (2p-2q)!} \, .
\ee
This formula can be further simplified by using the symmetry and translation identities for the Bernoulli polynomials. Namely,
\be
B_n(1-z) = (-1)^n B_n(z) \, , \quad B_n(z+v) = \sum_{k=0}^n \binom{n}{k} B_k(z) v^{n-k} \, , 
\ee
where $n \in \IN$ and $z,v \in \IC$. Choosing $z=1/2$, $v=1/4$, and $n = 2p+1$, a simple computation shows that
\be \label{eq: bernoulliSum}
\frac{4}{(2p+1)!} B_{2p+1}(3/4) = \sum_{q=0}^p \frac{B_{2p-2q}(1/2)}{4^{2q} (2q+1)! (2p-2q)!} \, .
\ee
Substituting Eq.~\eqref{eq: bernoulliSum} into Eq.~\eqref{eq: Vinter}, and recalling the special case
\be \label{eq: poly1}
\text{Li}_1(- \re^y) = - \log(1+ \re^y) \, ,
\ee
we finally obtain the perturbative expansion
\be \label{eq: F0int3}
V(y, \mb) = \frac{y}{2} - \log(1+ \re^y) +  4 \sum_{p=1}^{\infty} (2 \pi \text{i} \mb^2 )^{2p} \frac{B_{2p+1}(3/4)}{(2p+1)!} \text{Li}_{1-2p}(-e^y)  \, ,
\ee 
where only one infinite sum remains.
We note that the term of order $\mb^{-2}$ in Eq.~\eqref{eq: F0int3} vanishes, that is, the potential $V(y, \mb)$ does not have a critical point around which to perform a saddle point approximation of the integral in Eq.~\eqref{eq: F0int2}. 
Moreover, for each $p \ge 1$, the coefficient function of $(2 \pi \ri \mb^2)^{2p}$ can be written explicitly as a rational function in the variable $t = \re^y$ with coefficients in $\IQ$.
Indeed, the polylogarithm of negative integer order is
\be \label{eq: polyneg}
\text{Li}_{-n}(z) = \frac{1}{(1-z)^{n+1}} \sum_{k=0}^{n-1} \genfrac\langle\rangle{0pt}{0}{n}{k} z^{n-k} \, , \quad n \in \IZ_{>0} \, , \quad z \in \IC \, ,
\ee
where $\genfrac\langle\rangle{0pt}{1}{n}{k}$ are the Eulerian numbers. Applying Eq.~\eqref{eq: polyneg} to Eq.~\eqref{eq: F0int3}, we find that
\be \label{eq: F0potential}
V(y , \, \mb) = \frac{1}{2} \log(t) - \log(1 + t) + \sum_{p=1}^{\infty} \tilde{\mb}^{4p} \frac{P_p(t)}{(1+t)^{2p}} \, ,
\ee
where $\tilde{\mb}^2 = 2 \pi \ri \mb^2$, and $P_p(t)$ is a $\IQ$-polynomial in $t$ of degree $2p-1$. Explicitly,
\be \label{eq: F0poly}
P_p(t) =  4 \frac{B_{2p+1}(3/4)}{(2p+1)!}  \sum_{m=1}^{2p-1} (-1)^m \genfrac\langle\rangle{0pt}{0}{2p-1}{2p-1-m} \, t^m\, , \quad p \ge 1 \, .
\ee

We will now show how, by Taylor expanding the exponential $\re^{V(y, \, \mb)}$ in the limit $\mb \rightarrow 0$, we obtain a second perturbative $\mb$-series with coefficients which are identified $\IQ$-rational functions in $t$, and which can be explicitly integrated term-by-term to give the all-orders $\hbar$-expansion of the first spectral trace $\text{Tr}(\rho_{\IF_0})$ in Eq.~\eqref{eq: F0int}. 
In particular, we find that
\be \label{eq: F0potential2}
\begin{aligned}
\frac{1+t}{t^{1/2}} \, \re^{V(y , \, \mb)} & = 1 + \sum_{r=1}^{\infty}  \frac{1}{r!} \left( \sum_{p=1}^{\infty} \tilde{\mb}^{4p} \frac{P_p(t)}{(1+t)^{2p}} \right)^r  = \\
& = 1 + \sum_{k=1}^{\infty} \frac{\tilde{\mb}^{4k}}{(1+t)^{2k}} \sum_{m \in \CP(k)} \frac{1}{|m|!} \binom{|m|}{N_1, \dots , N_k}  P_{m_1}(t) \cdots P_{m_{|m|}}(t) \, ,
\end{aligned}
\ee
where $\CP(k)$ is the set of all partitions $m = (m_1, \dots , m_{|m|})$ of the positive integer $k$, $|m|$ denotes the length of the partition, and $N_i \in \IN$ is the number of times that the positive integer $i \in \IZ_{> 0}$ is repeated in the partition $m$. Note that $\sum_{i=1}^k N_i = |m|$.
The expansion in Eq.~\eqref{eq: F0potential2} can be written in a more compact form as
\be \label{eq: F0potential3}
\re^{V(y , \, \mb)} = \frac{t^{1/2}}{1 + t} \left( 1 + \sum_{k=1}^{\infty} \tilde{\mb}^{4k}  \frac{P'_{k}(t)}{(1+t)^{2k}} \right) \, ,
\ee
where $P'_{k}(t)$ is a new $\IQ$-polynomial in $t$ of degree $2k-1$, which is given explicitly by
\be \label{eq: F0poly2}
P'_{k}(t) =  \sum_{m \in \CP(k)} \frac{1}{N_1! \cdots N_k!} P_{m_1}(t) \cdots P_{m_{|m|}}(t) = \sum_{n=1}^{2k-1} c_{k,n} t^n \, , \quad k \ge 1 \, .
\ee
The numbers $c_{k,n} \in \IQ$ are directly determined by the coefficients of the polynomials $P_{p}(t)$, $p \ge 1$, in Eq.~\eqref{eq: F0poly} via the exponential expansion formula above.
Let us now substitute Eqs.~\eqref{eq: F0potential3} and~\eqref{eq: F0poly2} into the integral representation for the first spectral trace in Eq.~\eqref{eq: F0int2}, which gives
\be
\text{Tr}(\rho_{\IF_0}) = \frac{1}{4 \pi \mb^2} \int_{\mathbb{R}} \frac{\re^{y/2}}{1 + \re^y} \rd y +  \frac{1}{4 \pi \mb^2} \sum_{k=1}^{\infty} (2 \pi \ri \mb^2)^{2k} \sum_{n=1}^{2k-1} c_{k,n} \int_{\mathbb{R}} \frac{\re^{y/2+ny}}{(1 + \re^y)^{1+2k}} \rd y \, .
\ee
Explicitly evaluating the integrals in terms of gamma functions, and recalling the property 
\be
\Gamma\left(\frac{1}{2} + n \right) = \frac{(2n)! \sqrt{\pi}}{4^n n!} \, , \quad n \in \IN \, ,
\ee
we obtain that
\be \label{eq: Gammaint}
\int_{\mathbb{R}} \frac{\re^{y/2 + ny}}{(1+\re^y)^{1 + 2k}} \rd y =  \frac{\Gamma \left(\frac{1}{2} + 2k - n \right) \Gamma \left(\frac{1}{2} + n \right)}{\Gamma \left(1+ 2k  \right)} = \pi \frac{(2n-1)!! (4k-2n-1)!!}{4^{k} (2k)!} \, ,
\ee
for all $k,n \ge 0$. Therefore, the all-orders expansion in $\mb \rightarrow 0$ of the first spectral trace of the local $\IF_0$ geometry can be written as
\be \label{eq: F0expfinal}
\text{Tr}(\rho_{\IF_0}) = \frac{1}{4 \mb^2} \left( 1 +  \sum_{k=1}^{\infty} \mb^{4k} (-1)^k \frac{\pi^{2k}}{(2k)!} \sum_{n=1}^{2k-1} c_{k,n} (2n-1)!! (4k-2n-1)!! \right) \, .
\ee
We note that, substituting $\pi \mb^2 = \hbar$ and factoring out $\pi/4\hbar$, Eq.~\eqref{eq: F0expfinal} proves that the resulting perturbative series in $\hbar^2$ has coefficients in $\IQ$ of alternating sign.
We describe in Section~\ref{sec: algorithm} how to numerically implement the algorithmic procedure above in order to efficiently compute the perturbative series for $Z_{\IF_0}(1, \hbar \rightarrow 0)$ up to very high order. The first few terms are
\be \label{eq: first-terms-F0}
1 - \frac{\hbar^2}{64} + \frac{19 \hbar^4}{49152} - \frac{1013 \hbar^6}{47185920} + \frac{814339 \hbar^8}{338228674560} - \frac{449996063 \hbar^{10}}{974098582732800} + \text{O}(\hbar^{12}) \, ,
\ee
multiplied by the global pre-factor in Eq.~\eqref{eq: F0expfinal}, which confirms our analytic calculation at NLO in Eq.~\eqref{eq: WWfinalF0}.

\subsubsection{Comments on the algorithmic implementation} \label{sec: algorithm} 
In order to obtain several hundreds of terms of the perturbative series in $\hbar$ for the first fermionic spectral trace of local $\IF_0$ reasonably fast, we write a numerical algorithm which reproduces the analytic procedure described above.
We comment here briefly on some technical details. 
Let us denote by $d_{\rm max}$ the maximum order in $\mb$ to be computed numerically. We will work in the variable $\tilde{\mb}$, which is related to $\mb$ by
\be
\tilde{\mb}^2 = 2 \pi \ri \mb^2 \, .
\ee
After removing all non-rational factors from the intermediate steps, and truncating all calculations at the power $\tilde{\mb}^{d_{\rm max}+1}$ at every step, the computational complexity is dominated by the heavy multiplication of large multivariate polynomials, which is required at the early stage of the exponential expansion, and by the manipulation of the special functions that appear at the last stage of the integral evaluation. We implement our algorithm as the following two-step process. 
\begin{itemize}
\item[(1)] Starting from the series in Eq.~\eqref{eq: F0potential}, removing by hand the factor $1/(1 + t)^{2p}$, $p \ge 1$, and truncating at order $d_{\rm max}+1$ in $\tilde{\mb}$, we introduce the polynomial
\be \label{eq: polyMulti}
\varphi_1(\tilde{\mb}, t) = \sum_{p=1}^{d_{\rm max}/4} \tilde{\mb}^{4p} P_p(t) \, , 
\ee
where $P_p(t)$ is defined in Eq.~\eqref{eq: F0poly}, and its coefficients are computed explicitly. Then, we apply the variable redefinition 
\be \label{eq: redef}
\tilde{\mb} = t^{d_{\rm max}/2} \, ,
\ee
which transforms the two-variables polynomial in Eq.~\eqref{eq: polyMulti} into a polynomial in the single variable $t$, which we denote by $\varphi_1(t)$. Note that there is a one-to-one map between the coefficients of $\varphi_1(\tilde{\mb}, t)$ and the coefficients of $\varphi_1(t)$.
This allows us to perform the exponential expansion in Eq.~\eqref{eq: F0potential2} in the univariate polynomial ring $\IQ[t]$, instead of the bivariate polynomial ring $\IQ[t][\mb]$, without loss of information.
Truncating at order 
\be
b_{\rm max} + 1 = (d_{\rm max}+1)d_{\rm max}/2 \, ,
\ee 
we denote the resulting polynomial after the exponential expansion as
\be \label{eq: polyUni}
\varphi_2(t) = \re^{\varphi_1(t)} = 1 + \sum_{m=1}^{b_{\rm max}} C_m t^m \, , 
\ee
where $C_m \in \IQ$ is known numerically.
\item[(2)] Note that $\varphi_2(t)$ corresponds to the series in brackets in Eq.~\eqref{eq: F0potential3}. Indeed, we can write
\be \label{eq: polyUni2}
\varphi_2(t) = 1+ \sum_{k=1}^{d_{\rm max}/4} \left( t^{d_{\rm max}/2} \right)^{4k} \sum_{n=1}^{2k-1} c_{k,n} t^n = 1+ \sum_{k=1}^{d_{\rm max}/4} \sum_{n=1}^{2k-1} c_{k,n} t^{2k d_{\rm max}+n} \, ,
\ee
where $c_{k,n} \in \IQ$ is defined in Eq.~\eqref{eq: F0poly2}. 
Comparing Eqs.~\eqref{eq: polyUni} and~\eqref{eq: polyUni2}, we have that
\be
c_{k,n} = C_{2 k d_{\rm max}+n}, 
\ee
for all $1 \le n \le 2k-1$ and $k \ge 1$. Therefore, we extract the polynomial $P'_k(t)$ in Eq.~\eqref{eq: F0potential3} by selecting the monomials of order $t^m$ in Eq.~\eqref{eq: polyUni} such that
\be
2 k d_{\rm max} + 1 \le m \le 2k d_{\rm max} + 2k - 1 \, . 
\ee
Finally, the numerical coefficient of the term $\tilde{\mb}^{4k}$, $k \ge 1$, in the perturbative expansion of the first spectral trace $\text{Tr}(\rho_{\IF_0})$, up to the global pre-factor, is given by the finite sum 
\be
\sum_{n=1}^{2k-1} C_{2 k d_{\rm max}+n} I(k, \, n) \, ,
\ee
where $I(k,\, n) \in \pi \IQ$ denotes the numerical result of the pre-evaluated integral in Eq.~\eqref{eq: Gammaint}. We stress that numerical integration is not necessary.
\end{itemize}

\subsection{Exploratory tests of higher instanton sectors} \label{sec: borelF0}
Let us denote by $\psi(\hbar)$ the formal power series appearing in the RHS of Eq.~\eqref{eq: F0expfinal}, that is, 
\be \label{eq: psiF0}
\psi(\hbar) = \text{Tr}(\rho_{\IF_0}) \frac{4 \hbar}{\pi} \in \IQ[\![\hbar]\!] \, .
\ee
We truncate the series to a very high but finite order $d \gg 1$ and denote the resulting $\IQ$-polynomial by $\psi_d(\hbar)$ and its Borel transform by $\hat{\psi}_d(\zeta)$. Explicitly,
\be \label{eq: psi_d_F0}
\psi_d(\hbar) = \sum_{n=0}^d b_{2n} \hbar^{2n} \in \IQ [\hbar] \, , \quad \hat{\psi}_d(\zeta) = \sum_{n=0}^d \frac{b_{2n}}{(2n)!} \zeta^{2n} \in \IQ [\zeta] \, ,
\ee
where the coefficients $b_{2n}$, $1 \le n \le d$, have been computed numerically as described in Section~\ref{sec: WignerF0}. The first few terms of $\psi_d(\hbar)$ are shown in Eq.~\eqref{eq: first-terms-F0}. It is straightforward to verify numerically that the perturbative coefficients satisfy the factorial growth
\be \label{eq: coeff-factorialF0}
b_{2n} \sim (-1)^n (2n)! (2 \pi^2)^{-2n} \quad n \gg 1 \, ,
\ee
and we conclude that $\psi(\hbar)$ is a Gevrey-1 asymptotic series. As we have done in Section~\ref{sec: numericsP2} in the case of local $\IP^2$, we assume here $\hbar \in \IC'$ and apply the machinery of Pad\'e--Borel approximation~\cite{CD1, CD2, CD3} to the truncated series $\hat{\psi}_d(\zeta)$ in order to extrapolate the complex singularity structure of the exact analytically-continued Borel function $\hat{\psi}(\zeta)$. 

Let $d$ be even. We compute the singular points of the diagonal Pad\'e approximant of order $d/2$ of the truncated Borel expansion $\hat{\psi}_d(\zeta)$, which we denote by $\hat{\psi}^{\rm PB}_d(\zeta)$, and we observe two dominant complex conjugate branch points at $\zeta = \pm 2 \pi^2 \ri$, which match the leading divergent growth of the perturbative coefficients in Eq.~\eqref{eq: coeff-factorialF0}. Two symmetric arcs of complex conjugate spurious poles accumulate at the dominant singularities, mimicking two branch cuts which emanate straight from the branch points along the positive and negative imaginary axis in opposite directions. 
Let us introduce $A= 2 \pi^2$ and $\zeta_n = n 2 \pi^2 \ri$, $n \in \IZ_{\ne 0}$. A zoom-in to the poles of the Pad\'e approximant in the bottom part of the upper-half complex $\zeta$-plane is shown in the leftmost plot in Fig.~\ref{fig: PB-F0}.
\begin{figure}[htb!]
\centering
\includegraphics[width=1.\textwidth]{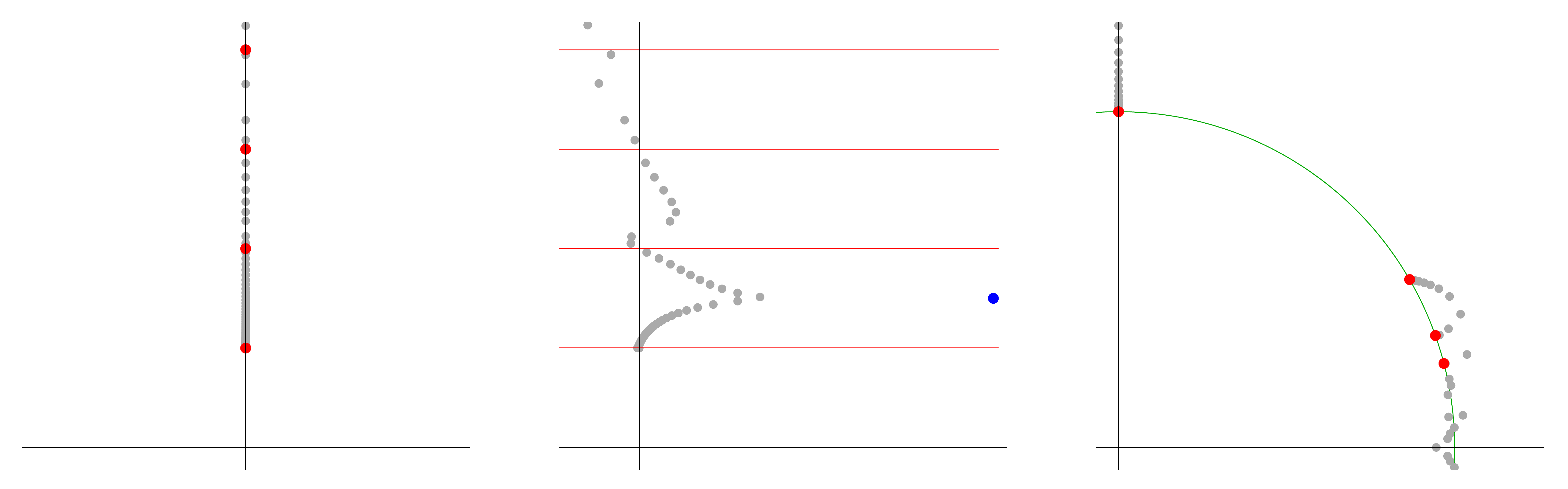}
\caption{In the leftmost plot, a zoom-in to the poles of $\hat{\psi}^{\rm PB}_d(\zeta)$ (in gray) in the bottom part of the upper-half complex $\zeta$-plane. We show the first few integer multiples of the dominant pole (in red). In the central plot, a zoom-in to the poles of $\hat{\psi}^{\rm PB}_{d, {\rm probe}}(\zeta)$ (in gray) in the bottom part of the upper-half complex $\zeta$-plane, including the test charge singularity at $\zeta = \zeta_{\text{probe}}$ (in blue). We show the horizontal lines intersecting the positive imaginary axis at the first few physical branch points (in red). In the rightmost plot, a zoom-in to the poles of $\hat{\psi}^{\rm PCB}_{d}(\eta)$ (in gray) in the upper-right quarter of the complex $\eta$-plane. We show the first few physical branch points (in red) and the unit circle (in green). The plots are obtained with $d=300$.}
\label{fig: PB-F0}
\end{figure}
Note that there might be subdominant true singularities of $\hat{\psi}_d(\zeta)$ which are obscured by the unphysical Pad\'e poles representing the dominant branch cuts. In particular, educated by the example of local $\IP^2$, we might guess that the dominant branch points at $\zeta = \zeta_{\pm 1}$ are repeated at integer multiples of $A$ along the imaginary axis to form a discrete tower of singularities. 
To reveal their presence, we apply again the interpretation of Pad\'e approximation in terms of electrostatic potential theory~\cite{Stahl, Saff}. More concretely, to test the presence of the suspected next-to-leading order branch point at $\zeta = \zeta_2$, we introduce by hand the singular term $\hat{\psi}_{\text{probe}}(\zeta) = \left( \zeta - \zeta_{\text{probe}} \right)^{-1/5}$, where $\zeta_{\text{probe}} = A  \left( 1/20 + 3 \ri/2 \right)$, and compute the $(d/2)$-diagonal Pad\'e approximant of the sum
\be
\hat{\psi}^{\rm PB}_{d, {\rm probe}}(\zeta) = (\hat{\psi}_d + \hat{\psi}_{\text{probe}})^{\rm PB} (\zeta) \, . 
\ee
The resulting Pad\'e--Borel poles distribution is distorted as shown in the central plot in Fig.~\ref{fig: PB-F0}, but the genuine physical singularities of $\hat{\psi}_d(\zeta)$ did not move. As a consequence, the targeted true branch point at $\zeta =  \zeta_2$ is now clearly visible, and so it is the test charge singularity at $\zeta = \zeta_{\text{probe}}$. The other true branch points at $\zeta = \zeta_3, \zeta_4$ are outlined. 

We can further improve the precision of the Pad\'e extrapolation and analytic continuation of the truncated Borel series $\hat{\psi}_d(\zeta)$ with the use of conformal maps. Namely, we perform the change of variable
\be
\zeta = A \frac{2 \eta}{1-\eta^2} \, , \quad \eta \in \IC \, ,
\ee
which maps the cut Borel $\zeta$-plane $\mathbb{C} \backslash \{(-\infty, \zeta_{-1}] \cup  [\zeta_1, +\infty) \}$ into the interior of the unit disk $|\eta| < 1$. 
The dominant branch points $\zeta= \zeta_{\pm 1}$ are mapped into $\eta = \pm \text{i}$, while the point at infinity is mapped into $\eta = \pm 1$. Correspondingly, the branch cuts $(-\infty, \zeta_{-1}]$ and $[\zeta_1, +\infty)$ along the imaginary $\zeta$-axis split each one into two identical copies which lie onto the two lower-half and upper-half quarters of the unit circle in the $\eta$-plane, respectively. The inverse conformal map is explicitly given by
\be
\eta = \pm \sqrt{\frac{-1 + \sqrt{1+(\zeta/A)^2}}{1 + \sqrt{1+(\zeta/A)^2}}} \,  , \quad \zeta \in \IC \, .
\ee
We compute the $(d/2)$-diagonal Pad\'e approximant of the conformally mapped Borel expansion, which we denote by $\hat{\psi}^{\rm PCB}_{d}(\eta)$. Its singularities in the complex $\eta$-plane are shown in the rightmost plot in Fig.~\ref{fig: PB-F0}.
We observe two symmetric arcs of spurious poles emanating from the conformal map images of $\zeta = \zeta_{\pm 1}$ along the imaginary axis in opposite directions. The smaller arcs of poles jumping along the unit circle towards the real axis represent the Pad\'e boundary of convergence joining the conformal map images of the repeated singularities at $\zeta = \zeta_n$, $n \in \IZ_{\ne 0}$. Note that the convolution of Pad\'e approximation and conformal maps naturally solves the problem of hidden singularities by separating the repeated branch points into different accumulation points on the unit circle in the conformally mapped complex $\eta$-plane. 
Thus, our numerical analysis motivates the following ansatz. The singularities of the exact Borel series $\hat{\psi}(\zeta)$ are logarithmic branch points at $ \zeta= \zeta_n$, $n \in \IZ_{\ne 0}$.
We remark that the complex singularity pattern unveiled here is entirely analogous to what has been found in Section~\ref{sec: numericsP2} for the local $\IP^2$ geometry, and again it is a particularly simple example of the peacock configurations described in Section~\ref{sec: resurgent_strings}. However, it turns out that the resurgent structure of the asymptotic series $\psi(\hbar)$ is more complex than what has been observed in other examples.

Let us go back to the factorization formula for the first spectral trace of local $\IF_0$ in Eq.~\eqref{eq: F0fact}. 
In the semiclassical limit $\hbar \rightarrow 0$, we have that $\tilde{q}=\re^{-2 \pi^2 \ri/\hbar} \rightarrow 0$ as well, and the anti-holomorphic blocks $\tilde{g}(\tilde{q}), \tilde{G}(\tilde{q})$ in Eqs.~\eqref{eq: seriesF0-gtilde} and~\eqref{eq: seriesF0-Gtilde} contribute trivially at leading order as
\be
\tilde{g}(\tilde{q}) \sim \frac{1}{2} \, , \quad \tilde{G}(\tilde{q}) \sim \tilde{q} \, .
\ee
On the other hand, the asymptotics of the holomorphic blocks $g(q), G(q)$ in Eqs.~\eqref{eq: seriesF0-g} and~\eqref{eq: seriesF0-G} for $q=\re^{2 \ri \hbar} \rightarrow 1$ depends a priori on the ray in the complex $\hbar$-plane along which the limit $\hbar \rightarrow 0$ is taken. Let us introduce the variable $\tau = \hbar/\pi$, such that $q = \re^{2 \pi \ri \tau}$, and take $\tau = \re^{\ri \alpha}/N$ with $\alpha \in \IR$ fixed and $N \rightarrow \infty$. The formula in Eq.~\eqref{eq: F0fact} gives the semiclassical asymptotics
\be \label{eq: F0fact-asym}
\text{Tr}(\rho_{\IF_0}) \sim - \frac{\ri}{4} G(q) -\frac{\ri 4}{\tau} \re^{-2 \pi \ri/\tau} g(q) \, .
\ee
Note that the contribution of $g(q)$ is suppressed by the exponentially-small factor $\re^{-2 \pi^2 \ri / \hbar}$, corresponding to the one-instanton non-perturbative sector, yielding that $\text{Tr}(\rho_{\IF_0}) \sim -\ri G(q)/4$ at leading order. We expect then to be able to recover the perturbative series $\psi(\hbar)$ in Eq.~\eqref{eq: psiF0} from the radial asymptotic behaviour of $G(q)$. Let us show that this is indeed the case.
We test the expected functional form
\be
G \left( \re^{2 \pi \ri \tau} \right) \sim C \tau^{-1} \left( 1 + \sum_{n=1}^{\infty} a_n \tau^n \right) \, , 
\ee
where $C \in \IC$, $a_n \in \IR$. We fix $0 < \alpha < \pi/2$ and take $N \in \IN$ to infinity. 
A numerical estimate of the overall constant $C$ is obtained from the convergence of the sequences
\begin{subequations}
\begin{align}
\Re \left( \tau G \left( \re^{2 \pi \ri \tau} \right) \right) &= \Re ( C ) + \mathcal{O}\left(1/N \right) \, , \\
\Im \left( \tau G \left( \re^{2 \pi \ri \tau} \right) \right) &= \Im ( C ) + \mathcal{O}\left(1/N \right) \, ,
\end{align}
\end{subequations}
which give $C \approx \ri$, and we proceed to systematically extract the coefficients $a_n$ as the large-$N$ limits of the recursively-built sequences
\be
\Re \left( \frac{1}{\tau^n} \left( - \ri \tau G \left( \re^{2 \pi \ri \tau} \right) - \sum_{j=0}^{n-1} a_j \tau^j \right) \right) = a_n + \mathcal{O}\left(1/N \right) \, , \quad n \in \IN_{>0} \, ,
\ee
where $a_0 =1$. We obtain in this way the high-precision numerical estimates
\be
a_{2n} \approx \pi^{2n} b_{2n} \, , \quad a_{2n+1} \approx 0 \, , \quad n \in \IN \, ,
\ee
where $b_{2n} \in \IQ$ are the coefficients of the perturbative series $\psi(\hbar)$ in Eq.~\eqref{eq: psiF0}, as expected. The numerical convergence of the large-$N$ limits of all sequences above has been accelerated with the help of Richardson transforms.
Let us now move on to the sub-dominant term in the RHS of Eq.~\eqref{eq: F0fact-asym} and determine the leading-order radial asymptotics of the $q$-hypergeometric series $g(q)$, for which we do not have an independent prediction. We observe that the standard functional ansatz
\be
g  \left( \re^{2 \pi \ri \tau} \right) \sim C_1 \re^{C_2/\tau} \tau^b  \, ,
\ee
where $C_1, C_2 \in \IC$, $b \in \IR$, fails our numerical tests, hinting at a possible logarithmic-type behavior.
We formulate the new ansatz\footnote{Note that the leading asymptotics of the standard hypergeometric function ${}_2 F_1(1/2, 1/2; 1; \re^{\pi \ri \tau})$ in the limit $\tau \rightarrow 0$ is known to be $- \frac{1}{\pi} \log(- \frac{\pi \text{i}}{16} \tau)$.} 
\begin{equation} \label{eq: new_ansatz}
g  \left( \re^{2 \pi \ri \tau} \right) \sim C_1 \log(\tau) + C_2 \, ,
\end{equation}
where $C_1, C_2 \in \IC$, and we test it as follows. Once more, let us fix $0 < \alpha < \pi/2$. Fore each $N \in \IN$, we define the numerical sequences $R_N = \Re (g(q))$ and $I_N = \Im (g(q))$, which satisfy
\begin{subequations} 
\begin{align}
R_N &\sim  -\Re (C_1) \log(N) + \Re (C_2) - \alpha \Im (C_1) \, , \\
I_N &\sim -\Im (C_1) \log(N) + \Im (C_2) + \alpha \Re (C_1) \, ,
\end{align}
\end{subequations}
for $N \gg 1$. We find $\Re (C_1) = -1/\pi$ from the convergence of the large-$N$ relation 
\be
S_N^{(1)}= N (R_{N+1} - R_N) \sim -\Re (C_1) \, .
\ee 
Analogously, we compute the large-$N$ limit of the sequence
\be
S_N^{(2)}= R_N + \Re (C_1) \log(N) \sim \Re (C_2) - \alpha \Im (C_1)  \, ,
\ee
which turns out to be independent of $\alpha$, giving $\Im(C_1) = 0$ and $\Re (C_2) =0.9225325 \dots$. Finally, the estimate $\Im (C_2) = 1/2$ is obtained from the convergence of the sequence 
\be
S_N^{(3)}= I_N - \alpha \Re (C_1) \sim \Im (C_2) \, ,
\ee
for $N \gg 1$. Again, all sequences are accelerated using Richardson transforms, as shown in the plots in Fig.~\ref{fig: radial_F0}.
\begin{figure}[htb!]
\centering
\includegraphics[width=1.\textwidth]{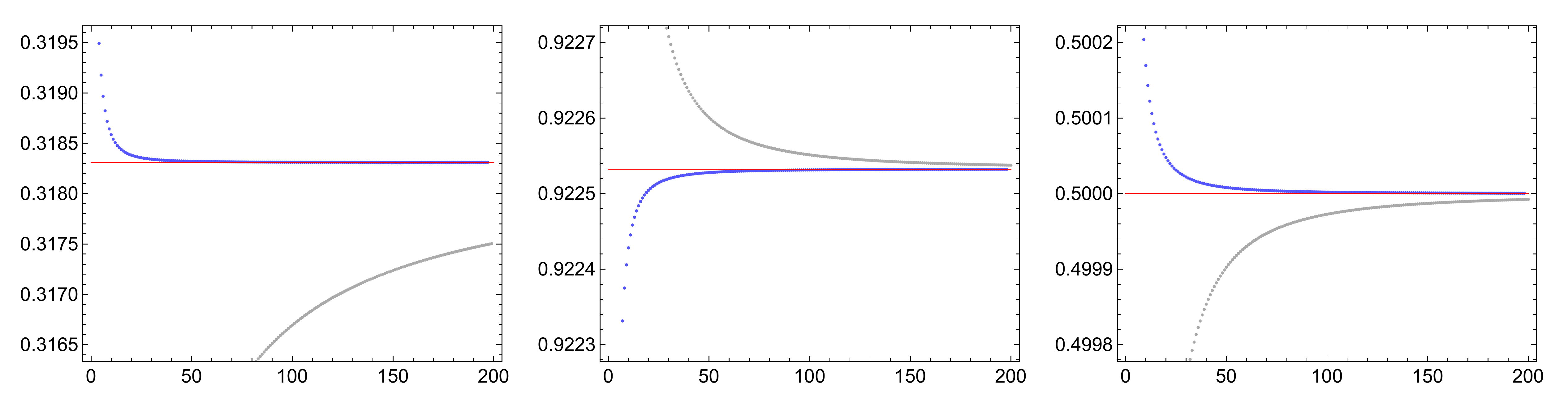}
\caption{The sequences $S_N^{(1)}$ (left), $S_N^{(2)}$ (center), and $S_N^{(3)}$ (right) are shown in gray, their second Richardson transforms in blue, and their estimated asymptotic limits in red. The plots are obtained with $\alpha = \pi/3$ and $N$ up to $200$.}
\label{fig: radial_F0} 
\end{figure}
We conclude that the $q$-series $g(q)$ has the proposed leading order asymptotics in Eq.~\eqref{eq: new_ansatz}. More precisely, we have that
\begin{equation} \label{eq: g_lead}
g \left( \re^{2 \pi \ri \tau} \right) \sim -\frac{1}{\pi} \log(\tau) + B + \frac{\text{i}}{2} \, ,
\end{equation}
where $B=0.9225325 \dots$. Note that the results that we have obtained from the radial asymptotic analysis of $g(q)$ and $G(q)$ are effectively independent of the choice of angle $0< \alpha < \pi/2$.

We remark that, as a consequence of Eq.~\eqref{eq: F0fact-asym}, the holomorphic block $g(q)$ contains sub-dominant, exponentially-suppressed perturbative corrections to the semiclassical perturbative expansion of the first spectral trace of local $\IF_0$, which is itself captured by the $q$-series $G(q)$. 
It follows then that the one-instanton resurgent contribution to the discontinuity, which we have denoted before by $\psi_1(\hbar)$, has a leading-order logarithmic behavior of the functional form in Eq.~\eqref{eq: new_ansatz}. Indeed, a straightforward independent test shows that the standard guess for $\psi_1(\hbar)$ in Eq.~\eqref{eq: ansatzP2} fails, since the corresponding numerical extrapolation of the one-instanton coefficients does not converge, in agreement with the radial asymptotics prediction. As a next step, we can simply extend the leading-order behavior in Eq.~\eqref{eq: new_ansatz} to a full perturbative expansion captured by a complete ansatz of the form
\begin{equation} \label{eq: new_ansatz_full}
g  \left( \re^{2 \pi \ri \tau} \right) \sim C_1 \log(\tau) \left( 1 + \sum_{n=1}^{\infty} d_n \tau^n \right) + C_2 \left( 1 + \sum_{n=1}^{\infty} e_n \tau^n \right) \, ,
\end{equation}
where $d_n, e_n \in \IC$, whose numerical investigation we leave for future work.

\sectiono{A new analytic prediction of the TS/ST correspondence} \label{sec: dual_limit}
Let us go back to the general framework of Section~\ref{sec: background} and consider the topological string on a toric CY threefold $X$. 
Because of the functional forms of the worldsheet and WKB grand potentials in Eqs.~\eqref{eq: J_WS} and~\eqref{eq: J_WKB}, respectively, there are appropriate scaling regimes in the coupling constants in which only one of the two components effectively contributes to the total grand potential in Eq.~\eqref{eq: J_total}. In the standard double-scaling limit~\cite{MP, KKN, GM}
\be \label{eq: SS1}
\hbar \rightarrow \infty \, , \quad \mu_j \rightarrow \infty \, , \quad \frac{\mu_j}{\hbar} = \zeta_j \; \; \text{fixed} \, , \quad j=1, \dots, g_{\Sigma} \, , 
\ee
under the assumption that the mass parameters $\bm{\xi}$ scale in such a way that $m_k = \xi_k^{2 \pi / \hbar}$, $k = 1, \dots, r_{\Sigma}$, are fixed as $\hbar \rightarrow \infty$, the quantum mirror map in Eq.~\eqref{eq: qMM} becomes trivial, and the total grand potential has the asymptotic genus expansion
\be \label{eq: tHooft}
J^{\text{'t Hooft}}(\bm{\zeta}, \bm{m}, \hbar) = \sum_{g =0}^{\infty} J_g(\bm{\zeta}, \bm{m}) \, \hbar^{2-2g} \, ,
\ee
where $J_g(\bm{\zeta}, \bm{m})$ is essentially the genus $g$ free energy of the conventional topological string at large radius after the B-field has been turned on.
As a consequence of the TS/ST correspondence, there is a related 't Hooft limit for the fermionic spectral traces which extracts the perturbative, all-genus expansion of the conventional topological string on $X$. Namely, 
\be \label{eq: SS2}
\hbar \rightarrow \infty \, , \quad N_j \rightarrow \infty \, , \quad \frac{N_j}{\hbar} = \lambda_j \; \; \text{fixed} \, , \quad j=1, \dots, g_{\Sigma} \, .
\ee
The saddle-point evaluation of the integral in Eq.~\eqref{eq: contour} in the double-scaling regime in Eq.~\eqref{eq: SS2}, which is performed using the standard 't Hooft expansion in Eq.~\eqref{eq: tHooft} for the total grand potential, represents a symplectic transformation from the large radius point in moduli space to the so-called maximal conifold point\footnote{The maximal conifold point can be defined as the unique point in the conifold locus of moduli space where its connected components intersect transversally.}~\cite{CGM2, CGuM}.
Specifically, it follows from the geometric formalism of~\cite{ABK} that the 't Hooft parameters $\lambda_j$ are flat coordinates on the moduli space of $X$ corresponding to the maximal conifold frame of the geometry, and we have the asymptotic expansion
\be \label{eq: SSexp}
\log Z_X^{\text{'t Hooft}}(\bm{N}, \bm{m}, \hbar) = \sum_{g =0}^{\infty} \mathcal{F}_g (\bm{\lambda}, \bm{m}) \, \hbar^{2-2g} \, ,
\ee
where the coefficient $\mathcal{F}_g (\bm{\lambda}, \bm{m})$ can be interpreted as the genus $g$ free energy of the standard topological string on $X$ in the maximal conifold frame.
We remark that, as a consequence of Eq.~\eqref{eq: SSexp}, the fermionic spectral traces provide a well-defined, non-perturbative completion of the conventional topological string theory on $X$. 

In this Section, we will show how a dual, WKB double-scaling regime associated to the limit $\hbar \rightarrow 0$ can be similarly introduced, in such a way that the symplectic transformation encoded in the integral in Eq.~\eqref{eq: contour} can be interpreted as a change of frame in the moduli space of $X$. We will obtain a new analytic prediction for the semiclassical asymptotics of the fermionic spectral traces in terms of the free energies of the refined topological string in the NS limit, which allows us to propose a non-trivial analytic test of the conjecture of~\cite{GHM, CGM2}.

\subsection{The WKB double-scaling regime} \label{sec: dual_tHooft}
We examine a second scaling limit of the total grand potential, which is dual to the standard 't Hooft limit in Eq.~\eqref{eq: SS1}. Namely, in the semiclassical regime
\be \label{eq: DS1}
\hbar \rightarrow 0 \, , \quad \mu_j, \xi_k \; \; \text{fixed} \, , \quad j=1, \dots, g_{\Sigma} \, , \quad k=1, \dots, r_{\Sigma} \, , 
\ee
the quantum mirror map $\bm{t}(\bm{z}, \hbar)$ becomes classical by construction, reducing to the formula in Eq.~\eqref{eq: MM2}, and the total grand potential in Eq.~\eqref{eq: J_total} retains only the WKB contribution coming from the NS limit of the refined topological string in Eq.~\eqref{eq: J_WKB}, which can be formally expanded in powers of $\hbar$ as
\be \label{eq: dual}
J^{\text{WKB}}(\bm{\mu}, \bm{\xi}, \hbar) = \sum_{n=0}^{\infty} J_n(\bm{\mu}, \bm{\xi}) \, \hbar^{2n-1} \, ,
\ee
where the fixed-order WKB grand potentials $J_n(\bm{\mu}, \bm{\xi})$ are given by
\begin{subequations} \label{eq: dual-fixed}
\begin{align}
J_0(\bm{\mu}, \bm{\xi}) &= \sum_{i=1}^s \frac{t_i}{2 \pi} \frac{\partial F^{\rm NS}_0(\bm{t})}{\partial t_i} - \frac{1}{\pi} F^{\rm NS}_0(\bm{t}) + 2 \pi \sum_{i=1}^s b_i t_i + A_0(\bm{\xi}) \, , \label{eq: dual-fixed-zero} \\ 
J_n(\bm{\mu}, \bm{\xi}) &= \sum_{i=1}^s \frac{t_i}{2 \pi} \frac{\partial F^{\rm NS}_n(\bm{t})}{\partial t_i} + \frac{(2n-2)}{2 \pi} F^{\rm NS}_n(\bm{t}) + A_n(\bm{\xi}) \, , \quad n \ge 1 \, , \label{eq: dual-fixed-n}
\end{align}
\end{subequations}
under the assumption that the function $A(\bm{\xi}, \hbar)$, which is defined in Eq.~\eqref{eq: J_WKB}, has the perturbative power series expansion\footnote{The assumption on $A(\bm{\xi}, \hbar)$ can be easily tested in examples.}
\be \label{eq: An-assumption}
A(\bm{\xi}, \hbar) = \sum_{n \ge 0} A_n(\bm{\xi}) \, \hbar^{2n-1} \, .
\ee
The constants $b_i$ are the same ones that appear in Eq.~\eqref{eq: WS2}, while $F_n^{\rm NS}(\bm{t})$ is the NS topological string amplitude of order $n$ in Eq.~\eqref{eq: NS}. Note that the worldsheet grand potential in Eq.~\eqref{eq: J_WS} does not contribute to the semiclassical perturbative expansion of the total grand potential, but it contains explicit non-perturbative exponentially-small effects in $\hbar$. 

In view of the integral representation in Eq.~\eqref{eq: contour}, we can define a second 't Hooft-like limit for the fermionic spectral traces $Z_X(\bm{N}, \bm{\xi}, \hbar)$, which corresponds to the semiclassical regime for the total grand potential in Eq.~\eqref{eq: DS1}, and it is dual to the standard 't Hooft limit in Eq.~\eqref{eq: SS2}. We refer to it as the WKB double-scaling regime. Namely,
\be \label{eq: DS2}
\hbar \rightarrow 0 \, , \quad N_j \rightarrow \infty \, , \quad N_j \hbar = \sigma_j \; \; \text{fixed} \, , \quad j=1, \dots, g_{\Sigma} \, .
\ee
Let us consider the case of a toric del Pezzo CY threefold $X$, that is, $g_{\Sigma} = 1$, for simplicity. The following arguments can then be straightforwardly generalized to the case of arbitrary genus.
Using the 't Hooft-like expansion in Eq.~\eqref{eq: dual} for the total grand potential, the integral formula in Eq.~\eqref{eq: contour} in the WKB double-scaling regime in Eq.~\eqref{eq: DS2} becomes
\be \label{eq: contourWKB}
Z_X(N, \bm{\xi}, \hbar) = \frac{1}{2 \pi \text{i}} \int_{\mathcal{C}} \rd \mu \, \exp \left(J^{\rm WKB}(\mu, \bm{\xi}, \hbar) - \frac{1}{\hbar} \mu \sigma \right) \, ,
\ee
where $\mathcal{C}$ is an integration contour going from $\re^{- \ri \pi/3} \infty$ to $\re^{+ \ri \pi/3} \infty$ in the complex plane of the chemical potential. Let us introduce the functions
\be
S(\mu, \bm{\xi}, \sigma) = \mu \sigma - J_0(\mu, \bm{\xi}) \, , \quad Z(\mu, \bm{\xi}, \hbar) = \exp \left( \sum_{n=1}^{\infty} J_n(\mu, \bm{\xi}) \hbar^{2n-1} \right) \, ,
\ee
and write Eq.~\eqref{eq: contourWKB} equivalently as
\be \label{eq: contourWKB2}
Z_X(N, \bm{\xi}, \hbar) = \frac{1}{2 \pi \text{i}} \int_{\mathcal{C}} \rd \mu \, \re^{-\frac{1}{\hbar} S(\mu, \bm{\xi}, \sigma)} \, Z(\mu, \bm{\xi}, \hbar) \, .
\ee
We identify the critical point $\mu = \mu^*$ of the integrand by solving the classical relation
\be \label{eq: saddle}
\frac{\d J_0(\mu^*, \bm{\xi})}{\d \mu} = \sigma \, ,
\ee
which gives $\sigma$ as a function of $\mu^*$, and vice-versa. Evaluating the integral in Eq.~\eqref{eq: contourWKB2} via saddle-point approximation around $\mu^*$ in the limit $\hbar \rightarrow 0$, we obtain the perturbative power series expansion
\be \label{eq: expWKB}
Z_X(N, \bm{\xi}, \hbar) =  \exp \left( \sum_{n=0}^{\infty} \CJ_n(\sigma, \bm{\xi}) \hbar^{2n-1} \right) \, .
\ee
The leading-order contribution is given by the Legendre transform
\be \label{eq: Legendre}
\CJ_0(\sigma, \bm{\xi}) = J_0(\mu^*, \bm{\xi}) -  \sigma \mu^* \, ,
\ee
where the saddle point $\mu^*$ is expressed as a function of $\sigma$ via Eq.~\eqref{eq: saddle}, the next-to-leading order correction is given by the one-loop approximation to the integral in Eq.~\eqref{eq: contourWKB2}, that is, 
\be \label{eq: one-loop}
\CJ_1(\sigma, \bm{\xi}) = J_1(\mu^*, \bm{\xi}) -  \frac{1}{2} \log \left( 2 \pi \frac{\d^2 J_0(\mu^*, \bm{\xi})}{\d \mu^2 } \right) \, ,
\ee
and the higher-order contributions $\CJ_n (\sigma, \bm{\xi})$, $n \ge 1$, can be computed systematically by summing over higher-loop Feynman diagrams.
Note that differentiating the formula in Eq.~\eqref{eq: Legendre} gives 
\be
\frac{\d \CJ_0(\sigma, \bm{\xi})}{\d \sigma} = - \mu^* \, .
\ee

\subsection{Interpreting the change of frame} \label{sec: sigma}
Let us consider the example of the local $\IP^2$ geometry. We recall that, in the parametrization of the moduli space given by the Batyrev coordinate $z$, the Picard--Fuchs differential equation associated to the mirror of local $\IP^2$ is given by
\be \label{eq: PF}
\CL_z \Pi = \left( \Theta^3-3z(3 \Theta+1)(3 \Theta+2)\Theta \right) \Pi = 0 \, ,
\ee
where $\Theta = z \rd/\rd z$, and $\Pi$ is the full period vector of the meromorphic differential one-form $\lambda = y(x) dx$. The Picard--Fuchs differential operator $\CL_z$ has three singular points: the large radius point at $z=0$, the conifold point at $z=-1/27$, and the orbifold point at $1/z=0$.
Solving the Picard--Fuchs equation locally around $z=0$ and using, for instance, the Frobenius method, gives a trivial constant solution and the two non-trivial independent solutions
\begin{subequations} \label{eq: w12}
\begin{align}
w_1(z) =& \begin{cases}
\log(-z) + \tilde{w}_1(z) \, ,  \quad &-1/27 < z < 0 \, ,\\
\log(z) + \tilde{w}_1(z) \, , \quad &z > 0 \, ,
\end{cases} \\
w_2(z) =& \begin{cases}
\log(-z)^2 + 2 \tilde{w}_1(z) \log(-z) + \tilde{w}_2(z)  \, , \quad &-1/27 < z < 0 \, ,\\
\log(z)^2 + 2 \tilde{w}_1(z) \log(z) + \tilde{w}_2(z)  \, , \quad &z > 0 \, ,
\end{cases}
\end{align}
\end{subequations}
where $\tilde{w}_1(z), \tilde{w}_2(z)$ are the formal power series
\begin{subequations} \label{eq: w12tilde}
\begin{align}
\tilde{w}_1(z) &= 3 \sum_{j=1}^{\infty} (-z)^j \frac{(3j-1)!}{(j!)^3} = -6 z + 45 z^2 - 560 z^3 + \frac{17325 z^4}{2} + \dots \, , \\
\tilde{w}_2(z) &= 18 \sum_{j=1}^{\infty} (-z)^j \frac{(3j-1)!}{(j!)^3} \sum_{n=j+1}^{3j-1} \frac{1}{n} = -18 z + \frac{423 z^2}{2} - 2972 z^3 + \frac{389415 z^4}{8}  + \dots \, .
\end{align}
\end{subequations}
Note that these series expansions converge for $|z|<1/27$, and, with an appropriate choice of normalization, the formulae in Eq.~\eqref{eq: periods} become
\be \label{eq: classicalP}
t = - w_1(z) \, , \quad \d_t F_0 = \frac{w_2(z)}{6} \, .
\ee
Let us derive exact expressions in terms of special functions for the analytic continuation of the classical periods at large radius. We remark that such closed formulae have already been computed in the well-known context of the standard 't Hooft regime~\cite{MZ}. However, this involves the conventional topological free energies after the B-field has been turned on, which is taken into account by implementing the change of sign $z \mapsto -z$. Using the explicit results of~\cite{MZ}, and reversing the effects of the change of sign in $z$, we find that the classical periods of local $\IP^2$ at large radius can be analytically continued and resummed in closed form in the two distinct regions $-1/27 < z < 0$ and $z > 0$ of moduli space. Choosing the branch of the logarithm functions appropriately, the first period $w_1(z)$ is given by
\be \label{eq: w1-closed} 
w_1(z) = \begin{cases}
\log(-z) - 6z \, {}_4F_3 \left( 1,1, \frac{4}{3}, \frac{5}{3} ; \, 2,2,2 ; \, -27z \right) \, , \quad &-1/27 < z < 0 \, ,\\
\log(z) - 6z \, {}_4F_3 \left( 1,1, \frac{4}{3}, \frac{5}{3} ; \, 2,2,2 ; \, -27z \right) \, , \quad &z > 0 \, ,
\end{cases} 
\ee
while the second period $w_2(z)$ is given by
\be \label{eq: w2-closed} 
w_2(z) = \begin{cases}
-\frac{5 \pi^2}{3} + \frac{3}{\pi \sqrt{3}} G_{3,3}^{3,2} \left(
\begin{matrix}
1/3 , & 2/3 , & 1 \\
0 , & 0 , & 0
\end{matrix} ; \, -27z
\right)  \, , \quad &-1/27 < z < 0 \, , \\
-\frac{2 \pi^2}{3} + \frac{3}{\pi \sqrt{3}} G_{3,3}^{3,2} \left(
\begin{matrix}
1/3 , & 2/3 , & 1 \\
0 , & 0 , & 0
\end{matrix} ; \, -27z
\right)  \\ \quad - 2 \pi \ri \log(z) +12 \pi \ri z \, {}_4F_3 \left( 1,1, \frac{4}{3}, \frac{5}{3} ; \, 2,2,2 ; \, -27z \right) \, ,  \quad &z > 0 \, , 
\end{cases}
\ee
where ${}_4F_3$ is a generalized hypergeometric function, and $G_{3,3}^{3,2}$ is a Meijer $G$-function. We show in the plots in Fig.~\ref{fig: periods} the analytically continued periods together with their series expansions at the origin.
\begin{figure}[htb!]
\center
 \includegraphics[width=0.9\textwidth]{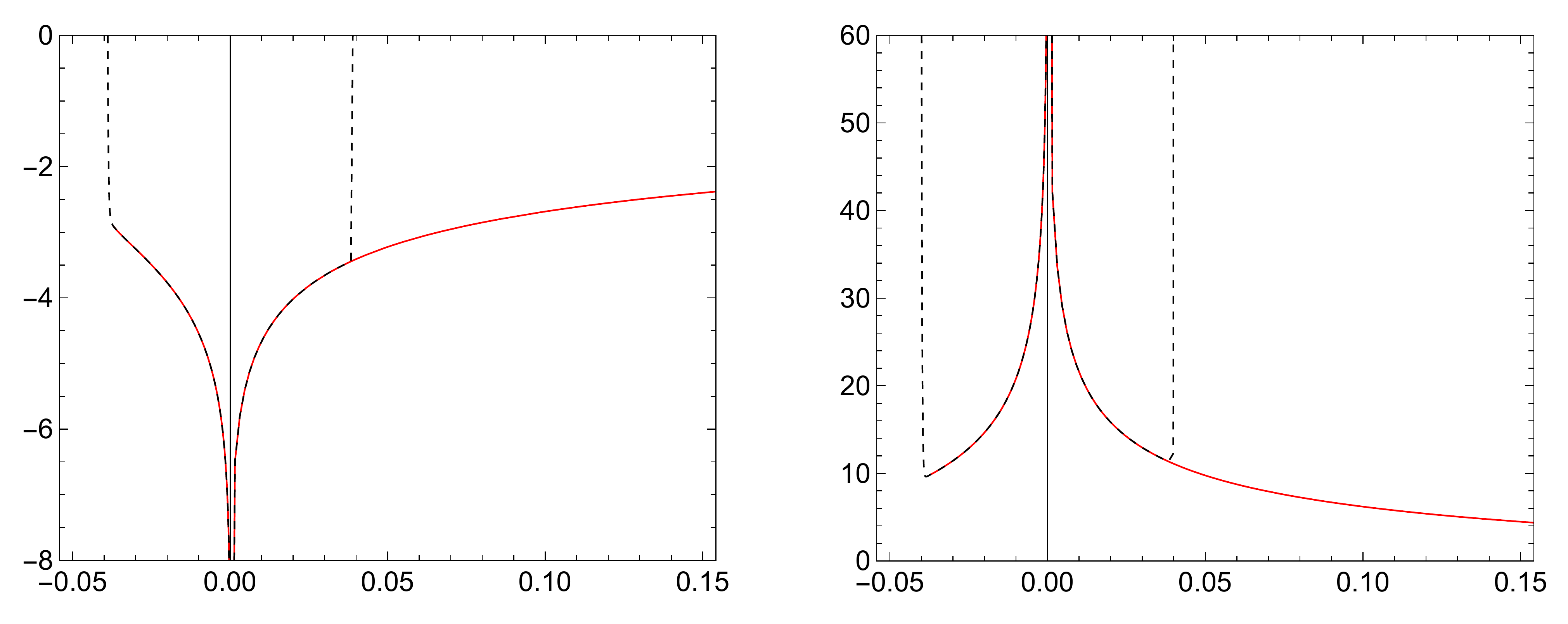}
 \caption{The functions $w_1(z)$ and $w_2(z)$ in Eqs.~\eqref{eq: w1-closed} and~\eqref{eq: w2-closed} (in solid red) on the left and on the right, respectively, and their large radius expansions in Eq.~\eqref{eq: w12} (in dashed black), for $-1/27< z < 0$ and $z > 0$. The formal power series in Eq.~\eqref{eq: w12tilde} are truncated at $j = 200$.} 
 \label{fig: periods}
\end{figure}
Let us comment briefly on their behavior at the critical points of the geometry. Both functions $w_1(z), w_2(z)$ have a vertical asymptote at the large radius point $z=0$, where they approach $\mp \infty$ from both sides, respectively. In the orbifold limit $z \rightarrow +\infty$, both periods have a horizontal asymptote corresponding to
\be \label{eq: orb-limits}
\lim_{z \rightarrow + \infty} w_1(z) = 0 \, , \quad \lim_{z \rightarrow + \infty} w_2(z) = -\frac{2 \pi^2}{3} \, .
\ee
Finally, at the conifold point $z = -1/27$, they reach the finite limits
\begin{subequations}
\begin{align}
\lim_{z \rightarrow -\frac{1}{27}^+} w_1(z) &= -3\log(3) +\frac{2}{9} \, {}_4F_3 \left( 1,1, \frac{4}{3}, \frac{5}{3} ; \, 2,2,2 ; \, 1 \right) = -2.907593524 \dots \, , \\
\lim_{z \rightarrow -\frac{1}{27}^+} w_2(z) &= \pi^2 \, .
\end{align}
\end{subequations}
We show in the plots in Fig.~\ref{fig: periods_deriv} the derivatives $\d_z w_1(z), \d_z w_2(z)$ together with their series expansions at the origin.
Note that  the function $w_1(z)$ is singular at the conifold, while $w_2(z)$ is not, and we have that
\be
\lim_{z \rightarrow -\frac{1}{27}^+} \d_z w_1(z) = -\infty \, , \quad \lim_{z \rightarrow -\frac{1}{27}^+} \d_z w_2(z) = 36 \sqrt{3} \pi \, .
\ee
\begin{figure}[htb!]
\center
 \includegraphics[width=0.9\textwidth]{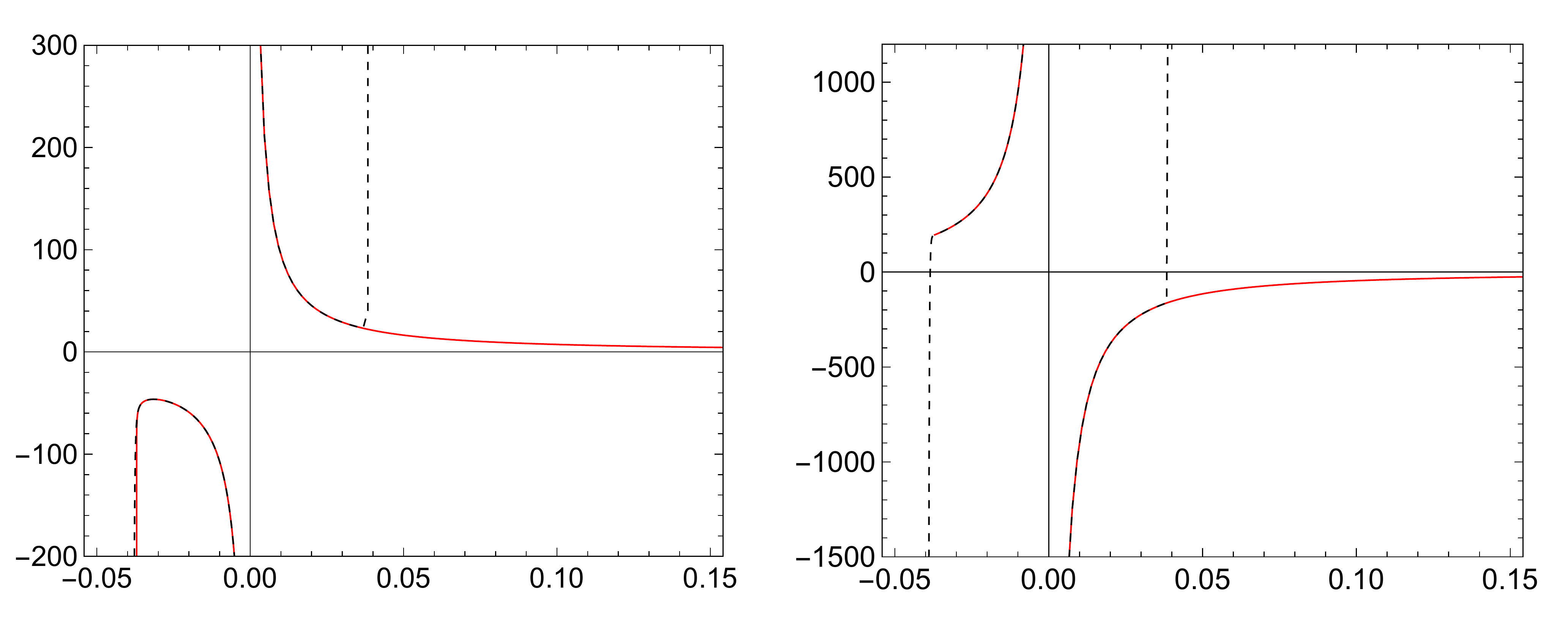}
 \caption{The derivatives $\d_z w_1(z)$ and $\d_z w_2(z)$ of the periods in Eqs.~\eqref{eq: w1-closed} and~\eqref{eq: w2-closed} (in solid red) on the left and on the right, respectively, and their large radius expansions following Eq.~\eqref{eq: w12} (in dashed black), for $-1/27< z < 0$ and $z > 0$. The formal power series in Eq.~\eqref{eq: w12tilde} are truncated at $j = 200$.} 
 \label{fig: periods_deriv}
\end{figure}

We recall that the function $A(\hbar)$, which appears in Eq.~\eqref{eq: J_WKB}, is conjectured in closed form for the local $\IP^2$ geometry, and it is given by~\cite{MP}
\be
A(\hbar) = \frac{3}{4} A_c (\hbar/\pi) - \frac{1}{4}A_c(3\hbar/\pi)
\ee
where the function $A_c (\hbar)$ can be expressed as~\cite{Ha+, HO}
\be
A_c(\hbar) = \frac{2 \zeta(3)}{\pi^2 \hbar} \left( 1 - \frac{\hbar^3}{16} \right) + \frac{\hbar^2}{\pi^2} \int_0^{\infty} \frac{x}{\re^{\hbar x}-1} \log \left( 1 - \re^{-2x} \right) \rd x \, .
\ee
It follows straightforwardly that the perturbative expansion in the limit $\hbar \rightarrow 0$ of $A(\hbar)$ satisfies the functional form in Eq.~\eqref{eq: An-assumption}. More precisely, we find that
\be \label{eq: A0}
A(\hbar) = \frac{4 \zeta(3)}{3 \pi \hbar} +  \frac{\hbar}{8 \pi} + \frac{\hbar^3}{2880 \pi} - \frac{\hbar^5}{604800 \pi} + \frac{\hbar^7}{33868800 \pi} + \CO (\hbar^9) \, .
\ee
Substituting the values $A_0= 4 \zeta(3)/3 \pi$ and $b=1/12$ in Eq.~\eqref{eq: dual-fixed-zero}, we obtain that the leading order WKB grand potential is given by
\be \label{eq: J-zero}
J_0(\mu) = \frac{t}{2 \pi} \d_t F_0(t) -\frac{1}{\pi} F_0(t) + \frac{\pi}{6} t + \frac{4 \zeta(3)}{3 \pi} \, ,
\ee
where $F_0(t) = F_0^{\rm NS}(t)$ is the genus-zero topological free energy at large radius. Integrating Eq.~\eqref{eq: classicalP} and fixing the integration constant appropriately, one has that~\cite{CKYZ}
\be \label{eq: genus-zero}
F_0(t) = \frac{t^3}{18}+3 \re^{-t}-\frac{45}{8} \re^{-2t}+\frac{244}{9} \re^{-3t}-\frac{12333}{64} \re^{-4t}+\frac{211878}{125} \re^{-5t} + \CO \left( \re^{-6t} \right) \, .
\ee
We can now apply Eqs.~\eqref{eq: saddle},~\eqref{eq: classicalP}, and~\eqref{eq: J-zero} in order to express the 't Hooft parameter $\sigma$ as a function of $z$. Recalling that the K\"ahler parameter $t$ is related to the chemical potential by $t = 3 \mu$, we find that
\be \label{eq: sigma-exact}
\sigma = \frac{3}{2 \pi} \left(t \, \d_t^2 F_0(t) - \d_t F_0(t) \right) + \frac{\pi}{2} = \frac{1}{4 \pi}\left( \frac{w_1(z) \, \d_z w_2(z)}{\d_z w_1(z)} - w_2(z)\right) + \frac{\pi}{2} \, .
\ee
Using the exact formulae for $w_1(z), w_2(z)$ in Eqs.~\eqref{eq: w1-closed} and~\eqref{eq: w2-closed}, we obtain the explicit dependence of $\sigma$ on the coordinate $z \in (-1/27, 0) \cup (0, +\infty)$, which is shown in the plot in Fig.~\ref{fig: sigma}.
\begin{figure}[htb!]
\center
 \includegraphics[width=0.7\textwidth]{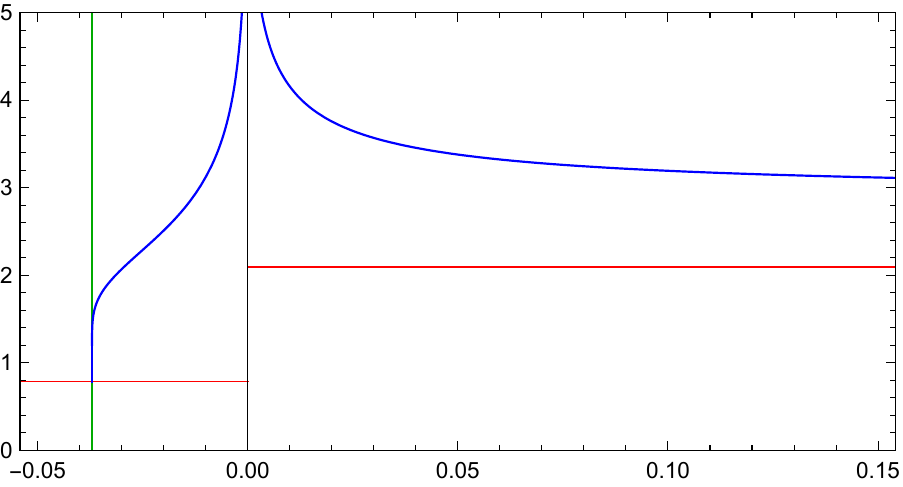}
 \caption{The 't Hooft parameter $\sigma$ (in solid blue) in Eq.~\eqref{eq: sigma-exact} as a function of the Batyrev coordinate $z \in (-1/27, 0) \cup (0, +\infty)$. We show the horizontal lines corresponding to the values $\sigma = \pi/4, 2\pi/3$ (in solid red) and the vertical line at $z=-1/27$ (in solid green).} 
 \label{fig: sigma}
\end{figure}
The two distinct analytic continuations of the classical periods correspond to two distinct analyticity regions of the 't Hooft parameter, whose behavior at the critical points is captured by the limits
\be
\lim_{z \rightarrow -\frac{1}{27}^+} \sigma(z) = \frac{\pi}{4} \, , \quad \lim_{z \rightarrow 0^{\pm}} \sigma(z) = +\infty \, , \quad \lim_{z \rightarrow + \infty} \sigma(z) = \frac{2 \pi}{3} \, .
\ee
Note that, analogously to the function $w_1(z)$, also $\sigma$ is singular at the conifold point of moduli space. In particular, we have that 
\be
\lim_{z \rightarrow -\frac{1}{27}^+} \d_z \sigma(z) = +\infty \, .
\ee
We stress that $\sigma$ is strictly positive for all values of $z$ considered here.

\subsection{Analysis of the leading-order behavior}
Let us now go back to the formula in Eq.~\eqref{eq: expWKB} and the new analytic prediction of the TS/ST correspondence which is contained in it. 
The asymptotic behavior of the fermionic spectral traces $Z_X(N, \bm{\xi}, \hbar)$ in the WKB double-scaling regime in Eq.~\eqref{eq: DS2} is determined by the WKB grand potential of the topological string theory, that is, by the total free energy of the refined topological string in the NS limit, after the transformation of local symplectic frame of the moduli space which is encoded in the integral in Eq.~\eqref{eq: contourWKB}.
We can write Eq.~\eqref{eq: expWKB} at leading order as
\be \label{eq: expWKB-zero}
Z_X(N, \bm{\xi}, \hbar) =  \exp \left(\CJ_0(\sigma, \bm{\xi}) \frac{1}{\hbar} + \CO(\hbar) \right) \, ,
\ee
where $\CJ_0(\sigma, \bm{\xi})$ is a highly non-trivial function of the 't Hooft parameter, and it is given explicitly by the Legendre transform in Eq.~\eqref{eq: Legendre}. We stress that the parametric dependence of the functional coefficients in Eq.~\eqref{eq: expWKB} on the moduli space makes studying directly the full resurgent structure of the given asymptotic expansion a much more difficult task than it is to investigate the numerical series $Z_X(N, \bm{\xi}, \hbar \rightarrow 0)$ at fixed $N$, as we have done in the previous sections of this paper. We can, however, perform a detailed analysis of the dominant contribution in Eq.~\eqref{eq: expWKB-zero}.

Let us consider the example of the local $\IP^2$ geometry again. Using Eqs.~\eqref{eq: classicalP},~\eqref{eq: J-zero}, and~\eqref{eq: sigma-exact}, we find that
\be \label{eq: CJ-zero}
\begin{aligned}
\CJ_0(\sigma) &= - \frac{t^2}{2 \pi} \d_t^2 F_0(t) + \frac{t}{\pi} \d_t F_0(t) - \frac{1}{\pi} F_0(t) + \frac{4 \zeta(3)}{3 \pi} = \\
&=-\frac{1}{6\pi} w_1(z) \left(\frac{ w_1(z) \d_z w_2(z)}{2 \d_z w_1(z)}+w_2(z) \right) - \frac{1}{\pi} F_0(t) + \frac{4 \zeta(3)}{3 \pi} \, ,
\end{aligned}
\ee
where we apply Eqs.~\eqref{eq: w1-closed} and~\eqref{eq: w2-closed} in order to express the classical periods as analytic functions of $z \in (-1/27, 0) \cup (0, +\infty)$, as before. If we consider the large radius expansion of the genus-zero topological free energy $F_0(t)$, which is given in Eq.~\eqref{eq: genus-zero}, we find that $\CJ_0(\sigma)$ has a vertical asymptote for $z \rightarrow 0^{\pm}$, where it approaches $- \infty$ from both sides, it remains negative for positive values of $z$, while it reaches zero at the value $z = z_0 = -0.0232698 \dots$, corresponding to $\sigma = \sigma_0 = 2.36250 \dots$, and it becomes positive for values of $z$ smaller than $z_0$. 
We show the resulting explicit dependence of $\CJ_0(\sigma)$ on $z$ near the large radius point in the plot in Fig.~\ref{fig: prediction}.
\begin{figure}[htb!]
\center
 \includegraphics[width=0.7\textwidth]{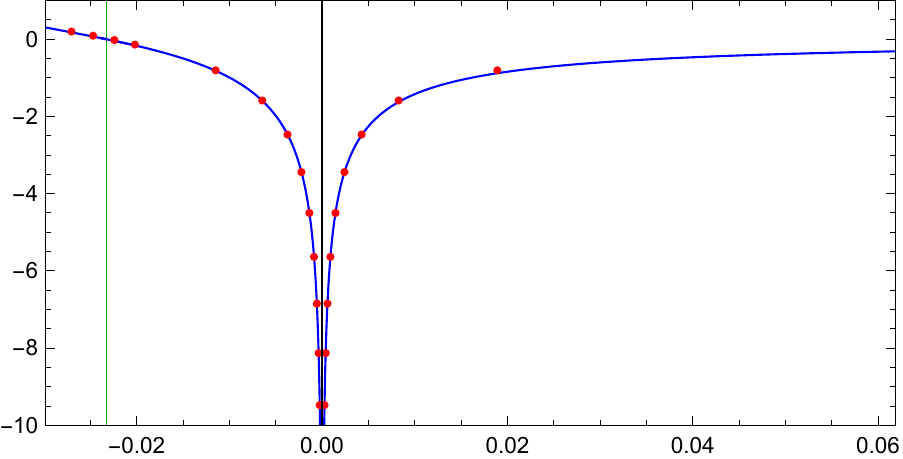}
 \caption{The transformed order-zero WKB grand potential $\CJ_0(\sigma)$ (in solid blue) in Eq.~\eqref{eq: CJ-zero} as a function of the Batyrev coordinate $z \approx 0$. We show the vertical line at $z=z_0$ (in solid green). We truncate the large radius expansion of $F_0(t)$ in Eq.~\eqref{eq: genus-zero} at order $\CO \left( \re^{-15 t} \right)$. We show the numerical data obtained from the sequences in Eq.~\eqref{eq: airy-seq} (as red dots) for $k=7$, $N$ up to $80$, and $\sigma = 2.2 , \, 2.3 , \, 2.4 , \, 2.5 , \, 3 , \, 3.5 , \, 4 , \, 4.5 , \, 5 , \, 5.5 , \, 6 , \, 6.5 , \, 7 $. The convergence is accelerated with a Richardson transform of order 2.}  
  \label{fig: prediction}
\end{figure}
We can similarly access the asymptotic behavior of $\CJ_0(\sigma)$ in the orbifold limit of the special geometry, where the 't Hoooft parameter $\sigma$ approaches the natural limit of $2 \pi/3$.
Recall that the Batyrev coordinate $z$ is related to the true complex structure parameter $\kappa$ by $z=\kappa^{-3}$. We express the analytically continued classical periods $w_1(z),w_2(z)$ for $z>0$ in Eqs.~\eqref{eq: w1-closed} and~\eqref{eq: w2-closed} as functions of $\kappa$. Substituting these into Eq.~\eqref{eq: classicalP}, and integrating in $t$, we find the orbifold expansion of the classical prepotential $F_0(\kappa)$. More precisely, we have that
\be \label{eq: orb-zero}
F_0(\kappa) = \frac{4 \zeta(3)}{3}-\frac{\pi^2 \Gamma(\frac{1}{3})}{9 \Gamma(\frac{2}{3})^2}\kappa + \frac{\pi \Gamma(\frac{1}{3})^2}{6 \sqrt{3}\Gamma(\frac{2}{3})^4}\kappa^2 + \frac{\pi^2 \Gamma(\frac{2}{3})^3}{18 \Gamma(\frac{1}{3})} \kappa^3 + \CO(\kappa^4) \, , 
\ee
after fixing the integration constant appropriately~\cite{GHM}. As a consequence of the orbifold limits in Eq.~\eqref{eq: orb-limits} and the orbifold expansion in Eq.~\eqref{eq: orb-zero}, the formula in Eq.~\eqref{eq: CJ-zero} gives the leading-order asymptotics
\be \label{eq: CJ-zero-orb}
\CJ_0 \left(\sigma \rightarrow \frac{2 \pi}{3}\right) \sim J_0 \left(\mu \rightarrow 0\right) \sim - \frac{1}{\pi} F_0(\kappa) + \frac{4 \zeta(3)}{3 \pi} = 0 + \CO(\kappa) \, ,
\ee
that is, the dominant term in the WKB expansion of $\log Z_{\IP^2}(N, \hbar)$, as given in Eq.~\eqref{eq: expWKB-zero}, vanishes at the orbifold point $\kappa = 0$ of moduli space, where it matches the order-zero WKB grand potential in Eq.~\eqref{eq: J-zero}. 
Let us point out that, expanding the conjectural formula for the spectral determinant in Eq.~\eqref{eq: GHM} at leading order in the semiclassical regime in Eq.~\eqref{eq: DS1}, we obtain that
\be \label{eq: XiWKB}
\log \Xi_{\IP^2}(\kappa, \hbar) \sim J^{\text{WKB}}(\mu, \hbar) \, , \quad \hbar \rightarrow 0 \, , \quad \mu \; \; \text{fixed} \, ,
\ee
while the theta function $\Theta(\mu, \hbar)$ gives sub-leading, oscillatory corrections to the dominant asymptotics. 
Since the fixed-$N$ fermionic spectral traces are originally defined as the functional coefficients in the orbifold expansion of $\Xi_{\IP^2}(\kappa, \hbar)$ in Eq.~\eqref{eq: expXi}, the statements in Eqs.~\eqref{eq: CJ-zero-orb} and~\eqref{eq: XiWKB} predict the absence of a global exponential behavior of the form $\re^{-1/\hbar}$ in their semiclassical asymptotic expansion, contrary to what occurs in the dual regime $\hbar \rightarrow \infty$, as observed in~\cite{GuM} and in Section~\ref{sec: P2infty}. This new prediction of the TS/ST correspondence is explicitly verified for the first spectral trace of both local $\IP^2$ and local $\IF_0$ in Sections~\ref{sec: WignerP2} and~\ref{sec: WignerF0}, thus providing a first successful quantum-mechanical test of the analytic formulation encoded in Eq.~\eqref{eq: expWKB}.
We comment that obtaining analytic results from Eqs.~\eqref{eq: expXi} and~\eqref{eq: GHM} is generally a difficult task, as it requires knowledge of the analytic continuation of the total grand potential to the orbifold frame of the geometry for the values of $\hbar$ of interest. 

We conclude by presenting here a numerical test of consistency of the analytic prediction in Eq.~\eqref{eq: CJ-zero} at large radius for the local $\IP^2$ geometry. It is well-known that, by expanding the total grand potential $J(\mu, \hbar)$ in Eq.~\eqref{eq: contour} in the large-$\mu$ limit, the fermionic spectral traces $Z_{\IP^2}(N, \hbar)$ can be decomposed into an infinite sum of Airy functions and their derivatives. More precisely, one has that~\cite{GHM}
\be \label{eq: expAiry}
\exp \left( J(\mu, \hbar) \right) = \exp \left(\frac{C(\hbar)}{3} \mu^3 + B(\hbar) \mu + A(\hbar) \right) \sum_{l,n} a_{l,n} \mu^n \re^{-l \mu} \, ,
\ee
where $a_{l,n} \in \IC$, and the sum runs over $n \in \IN$ and $l = 3 p+6 \pi q /\hbar$ for $p,q \in \IN$. We have denoted by $A(\hbar)$ the same function that appears in Eq.~\eqref{eq: J_WKB}, and we have introduced
\be
B(\hbar) = \frac{\pi}{2 \hbar} - \frac{\hbar}{16 \pi} \, , \quad C(\hbar) = \frac{9}{4 \pi \hbar} \, .
\ee 
Under some assumptions on the convergence of the expansion in Eq.~\eqref{eq: expAiry}, one can substitute it into the formula in Eq.~\eqref{eq: contour} and perform the integration term-by-term, which gives~\cite{GHM}
\be \label{eq: Airy}
Z_{\IP^2}(N, \hbar) = \frac{\re^{A(\hbar)}}{C(\hbar)^{1/3}} \sum_{l,n} a_{l,n} \left( - \frac{\d}{\d N} \right)^n \text{Ai} \left( \frac{N+l-B(\hbar)}{C(\hbar)^{1/3}} \right) \, ,
\ee
where $\text{Ai}(x)$ is the Airy function, and the indices $l,n$ are defined as above. The leading-order behavior for $N \rightarrow \infty$ and $\hbar$ fixed is given by the Airy function
\be \label{eq: Airy-zero}
Z_{\IP^2}(N, \hbar) \sim  \text{Ai} \left( \frac{N-B(\hbar)}{C(\hbar)^{1/3}} \right) \, , \quad N \gg 1 \, ,
\ee
while all additional terms in the RHS of Eq.~\eqref{eq: Airy} are exponentially-small corrections. Note that the series in Eq.~\eqref{eq: Airy} appears to be convergent, and it allows us to obtain highly accurate numerical estimates of the fermionic spectral traces~\cite{reviewM}. More precisely, we truncate it to the finite sum $Z_{\IP^2}^{(k)}(N, \hbar)$, where $k \in \IN$ denotes the number of terms that have been retained.
Since we are interested in the WKB double-scaling regime in Eq.~\eqref{eq: DS2}, let us fix a value of the 't Hooft parameter $\sigma \in \IR_+$ and take $\hbar = \sigma/N$ for $N \in \IN$. 
Thus, the sequence of numerical approximations $Z_{\IP^2}^{(k)}(N, \sigma/N)$ tends to the true function $Z_{\IP^2}(N, \hbar)$ for $k,N \rightarrow \infty$ for each choice of $\sigma$. Note, however, that we can only access in this way those values of $\sigma$ which correspond to the large radius frame in moduli space, that is, $z \approx 0$. As we have found in Section~\ref{sec: sigma}, this corresponds to $\sigma \gg 1$.
We obtain a numerical estimate of the transformed order-zero WKB grand potential $\CJ_0(\sigma)$ in Eq.~\eqref{eq: CJ-zero} near the large radius point from the convergence of the sequence
\be \label{eq: airy-seq}
\frac{\sigma}{N} \log \left( Z_{\IP^2}^{(k)}\left( N, \frac{\sigma}{N} \right) \right) \sim \CJ_0(\sigma) \, , \quad N \gg 1 \, ,
\ee
which is accelerated with the help of Richardson transforms. 
The resulting numerical data correctly captures the analytic behavior described by Eq.~\eqref{eq: CJ-zero} at large radius, including the change of sign of $\CJ_0(\sigma)$ occurring at $z=z_0$, as shown for a selection of points in the plot in Fig.~\ref{fig: prediction}.
As expected, the agreement increases as the points get closer to $z=0$, and it is systematically improved by taking larger values of $k \in \IN$. 
The two $z$-coordinates corresponding to each choice of $\sigma$ are obtained by inverting the relation in Eq.~\eqref{eq: sigma-exact}.

\sectiono{Conclusions} \label{sec: conclusions}
In this paper, we have described how the machinery of resurgence can be effectively applied to study the perturbative expansions of the fermionic spectral traces in the semiclassical regime $\hbar \rightarrow 0$ of the spectral theory dual to the topological string theory on a toric CY threefold. The resurgent analysis of these asymptotic series uncovers a rich mathematical structure of non-perturbative sectors, which appear in the complex Borel plane as infinite towers of periodic singularities, whose arrangement is known as a peacock pattern. We have conjectured that the Stokes constants associated to the logarithm of the fermionic spectral traces are rational numbers, thus representing a new, conjectural class of enumerative invariants of the CY.

We have analytically solved the full resurgent structure of the logarithm of the first spectral trace of the local $\IP^2$ geometry, which unveils a remarkable arithmetic construction. 
The Stokes constants are given by explicit divisor sum functions, which can be expressed as the Dirichlet convolution of simple arithmetic functions, and they are encoded in a generating function given in closed form by $q$-series. The perturbative coefficients are captured by special values of a known $L$-function, which admits a notable factorization as the product of a Riemann zeta function and a Dirichlet $L$-function, corresponding to the Dirichlet factors in the decomposition of the Stokes constants.
Analogously, we have presented a complete analytic solution to the resurgent structure of the perturbative series arising in the dual weakly-coupled limit $g_s \rightarrow 0$ of the conventional topological string on the same background, previously studied numerically in~\cite{GuM}. We have found that the Stokes constants are manifestly related to their semiclassical analogues, the mapping being realised by a simple exchange of divisors, while the perturbative coefficients are special values of the same $L$-function above after unitary shifts in the arguments of its factors. The more complex example of the local $\IF_0$ geometry appears to be only accessible via numerical tools, and we have found a logarithmic-type sub-leading asymptotics of the perturbative series of the first spectral trace. Finally, we have analyzed the topological string total grand potential of a toric CY threefold in an appropriate WKB 't Hooft-like regime associated to the semiclassical limit of the spectral theory, and we have obtained a non-trivial analytic prediction on the asymptotics of the fermionic spectral traces in terms of the total free energy of the refined topological string in the NS limit in a suitable symplectic frame.

Several questions and open problems follow from the investigation performed in this paper. We will now give a short account of possible directions for future work. 
We have often mentioned that an independent geometric and physical understanding of the non-perturbative sectors studied in~\cite{GuM} and in this paper is still missing, as it is an explicit identification our rational Stokes constants as enumerative invariants of the topological string. The number-theoretic fabric that we have discovered in the simple case of the local $\IP^2$ geometry in Section~\ref{sec: localP2} sheds some light on this problem, but it leaves room for further exploration. We would like to understand in a similar way other examples of toric CY threefolds, and possibly study higher-order fermionic spectral traces, in support of a potential generalization.  

We have described how the integral representation of the fermionic spectral traces in the WKB double-scaling regime can be interpreted as a symplectic transformation of the WKB grand potential at large radius. As explained in~\cite{ABK}, a change of symplectic frame in the moduli space of the CY $X$ corresponds to an electromagnetic duality transformation in $\text{Sp}(2s, \IZ)$ of the periods, where $s = b_2(X)$. However, a full geometric understanding of the effect of the change of symplectic basis of Section~\ref{sec: dual_limit} on the WKB grand potential would require further work, and it would involve, in particular, the use of the specific modular transformation properties of the topological string amplitudes in the NS limit, which have been studied in~\cite{Grassi}.

In the first part of this work, we have considered the numerical series obtained from the semiclassical expansion of the fermionic spectral traces at fixed $\bm{N}$. However, as we have seen in Section~\ref{sec: dual_limit}, we can rigorously define the perturbative expansion of the fermionic spectral traces in the WKB double-scaling regime directly. An advantage of this formulation is that it allows us to make a precise statement on the TS/ST prediction of the asymptotic behavior of these series. Studying their resurgent structure is a much more complex endeavor than it is for the numerical series, because the perturbative coefficients in this 't Hooft-like limit retain a full parametric dependence on the moduli space of the CY. Some work on similar problems of parametric resurgence has been done for the standard 't Hooft regime of conventional topological string theory~\cite{C-SMS} and in the context of the large-$N$ expansion of gauge theories~\cite{lecturesM}.

Let us mention that the analytic prediction that we have presented for the WKB double-scaling regime of the fermionic spectral traces is, in principle, verifiable from a matrix model perspective. A successful quantum-mechanical validation of our full statement would represent an important additional piece of evidence in support of the conjecture of~\cite{GHM, CGM2}. 

Finally, in the recent works of~\cite{GuM2, GuM3}, an operator formulation of the holomorphic anomaly equations, or BCOV equations~\cite{BCOV1, BCOV2}, satisfied by the conventional and NS topological string free energies is applied to derive the formal structure of the exact trans-series solutions, although a determination of which of the possible Borel singularities are realized and the values of their Stokes constants is missing. A detailed study of the connection and complementarity between the formalism of~\cite{GuM2, GuM3} and the framework proposed in~\cite{GuM} and further advanced in the present work might help us achieve a more comprehensive understanding of the resurgent structure of topological string theory.

\acknowledgments
I am thankful to Marcos Mari\~no for many discussions throughout the development of the project. I also thank Alba Grassi, Jie Gu, and Nikita Nikolaev for useful comments.
This work has been supported by the ERC-SyG project ``Recursive and Exact New Quantum Theory'' (ReNewQuantum), which received funding from the European Research Council (ERC) within the European Union's Horizon 2020 research and innovation program under Grant No. 810573, and by the Swiss National Centre of Competence in Research SwissMAP (NCCR 51NF40-141869 The Mathematics of Physics).

\appendix

\sectiono{Wigner transform of the inverse operator} \label{app: Wigner} 
Let $\mO$ be a quantum-mechanical operator acting on $L^2(\IR)$. We describe how to obtain the WKB expansion in phase space of the inverse operator $\rho = \mO^{-1}$ at NLO in the semiclassical limit $\hbar \rightarrow 0$. See~\cite{bookM2} for an introduction to the phase-space formulation of quantum mechanics. 
The Wigner transform of the operator $\mO$ is defined as
\be \label{eq: OW}
O_{\rm W}(x,y) = \int_{\IR} \re^{\frac{\ri y x'}{\hbar}}  \left\langle x - \frac{x'}{2} \right\rvert \mO \left\lvert x + \frac{x'}{2} \right\rangle \rd x' \, ,
\ee
where $x, \, y \in \IR$ are the phase-space coordinates, and the diagonal element of $\mO$ in the coordinate representation is given by
\be
\left\langle x \right\rvert \mO \left\lvert x  \right\rangle = \frac{1}{2 \pi \hbar} \int_{\IR} O_{\rm W}(x,y) \, \rd y \, . 
\ee
The trace of $\mO$ is obtained by integrating its Wigner transform over phase space, that is, 
\be \label{eq: tracePS}
\text{Tr}(\mO) =  \int_{\IR} \left\langle x \right\rvert \mO \left\lvert x  \right\rangle \, \rd x  =  \frac{1}{2 \pi \hbar} \int_{\IR^2} O_{\rm W}(x,y) \, \rd x \rd y \, . 
\ee
\begin{example} \label{ex: exampleWigner}
Let us consider the quantum-mechanical operator 
\be
\mO = \re^{m \mx + n \my}  = \re^{- \ri \hbar m n /2} \re^{m \mx} \re^{n \my} \, ,  \quad m, n \in \IZ \, , 
\ee
where $\mx, \, \my$ are the Heisenberg operators corresponding to the classical variables $x, \, y$ and satisfying $[ \mx , \, \my ] = \text{i} \hbar$, and the second equality follows from the Baker--Campbell--Hausdorff formula.
The Wigner transform in Eq.~\eqref{eq: OW} becomes
\be \label{eq: OWex1}
\ba
O_{\rm W}(x,y) &=  \re^{- \ri \hbar m n /2} \int_{\IR} \rd x' \re^{\frac{\ri y x'}{\hbar}}  \left\langle x - \frac{x'}{2} \right\rvert \re^{m \mx} \re^{n \my} \left\lvert x + \frac{x'}{2} \right\rangle= \\
&= \re^{- \ri \hbar m n /2} \int_{\IR} \rd x' \re^{\frac{\ri y x'}{\hbar}}  \re^{m (x-x'/2)} \left\langle x - \frac{x'}{2} \right\rvert \re^{n \my} \left\lvert x + \frac{x'}{2} \right\rangle= \\
&= \re^{- \ri \hbar m n /2} \re^{m x} \int_{\IR^2} \frac{\rd x' \rd y' }{2 \pi \hbar} \re^{\frac{\ri (y-y') x'}{\hbar}}  \re^{-mx'/2} \re^{n y'} \, .
\ea
\ee
After performing the change of variable $u = x'/\hbar$, and Taylor expanding the exponential factor $\re^{-m u \hbar / 2}$ around $\hbar = 0$, we get
\be \label{eq: OWex}
\ba
O_{\rm W}(x,y) &= \re^{- \ri \hbar m n /2} \re^{m x} \sum_{k = 0}^{\infty} \frac{1}{k!} \left( - \frac{\ri \hbar m}{2} \right)^k \int_{\IR} \rd y' \re^{n y'} \int_{\IR} \frac{\rd u}{2 \pi} \re^{\ri (y-y') u} (- \ri u)^k = \\ 
&= \re^{- \ri \hbar m n /2} \re^{m x} \sum_{k = 0}^{\infty} \frac{1}{k!} \left( - \frac{\ri \hbar m}{2} \right)^k \int_{\IR} \rd y' \re^{n y'} \delta^{(k)}(y-y') = \\ 
&= \re^{- \ri \hbar m n /2} \re^{m x} \sum_{k = 0}^{\infty} \frac{1}{k!} \left( \frac{\ri \hbar m n}{2} \right)^k \re^{n y}  = \re^{m x + n y} \, , 
\ea
\ee
where $\delta^{(k)}$ denotes the $k$-th derivative of the Dirac delta function. 
\end{example}

We recall now the definition of the Moyal $\star$-product of two quantum operators $\mA, \, \mB$ acting on $L^2(\IR)$, that is, 
\be \label{eq: moyal}
\mA \star \mB = A_{\rm W}(x, y) \, \exp \left[ \frac{\ri \hbar}{2} \overleftrightarrow{\Lambda} \right] \, B_{\rm W}(x,y) \, ,
\ee
where $\overleftrightarrow{\Lambda} = \overleftarrow{\partial_x} \overrightarrow{\partial_y} - \overleftarrow{\partial_y} \overrightarrow{\partial_x}$, and the arrows indicate the direction in which the derivatives act. 
Expanding around $\hbar = 0$, we get
\be \label{eq: moyal2}
\mA \star \mB = \sum_{n=0}^{\infty} \sum_{m=0}^n (-1)^m \binom{n}{m} \frac{1}{n!} \left( \frac{ \ri \hbar}{2} \right)^n \d_x^m \d_y^{n-m} A_{\rm W}(x, y) \, \d_y^m \d_x^{n-m} B_{\rm W}(x, y)  \, .
\ee
\begin{theorem}
The Wigner transform of the inverse operator $\rho = \mO^{-1}$ can be expressed in terms of $O_W$ as
\be \label{eq: rhoW}
\rho_{\rm W} = \sum_{r = 0}^{\infty} (-1)^r \frac{\CG_r}{O_{\rm W}^{r+1}} \, ,
\ee
where the quantities
\be \label{eq: functionsGr}
\CG_r = \left[ (\mO - O_{\rm W}(x,y))^r \right]_{\rm W}
\ee
are evaluated at the same point $(x, \, y)$ in phase space. 
\end{theorem}
We can compute the functions in Eq.~\eqref{eq: functionsGr} explicitly by using the Moyal $\star$-product in Eq.~\eqref{eq: moyal} and expand them in formal power series in $\hbar$. More precisely, the first four functions are $\CG_0=1$, $\CG_1=0$, 
\begin{subequations}
\begin{align}
\CG_2 &= O_{\rm W} \star O_{\rm W} - O_{\rm W}^2 = -\frac{\hbar^2}{4} \left[ \frac{\partial^2 O_{\rm W}}{\partial x^2} \frac{\partial^2 O_{\rm W}}{\partial y^2}  - \left( \frac{\partial^2 O_{\rm W}}{\partial x \partial y} \right)^2 \right] + \CO(\hbar^4) \, , \label{eq: G2} \\
\CG_3 &= O_{\rm W} \star O_{\rm W} \star O_{\rm W} - 3 (O_{\rm W} \star O_{\rm W}) O_{\rm W} + 2 O_{\rm W}^3 = \label{eq: G3} \\
&= -\frac{\hbar^2}{4} \left[ \left(\frac{\partial O_{\rm W}}{\partial x} \right)^2 \frac{\partial^2 O_{\rm W}}{\partial y^2}  + \frac{\partial^2 O_{\rm W}}{\partial x^2} \left( \frac{\partial O_{\rm W}}{\partial y} \right)^2  - 2 \frac{\partial O_{\rm W}}{\partial x} \frac{\partial O_{\rm W}}{\partial y} \frac{\partial^2 O_{\rm W}}{\partial x \partial y} \right] + \CO(\hbar^4) \, . \nonumber
\end{align}
\end{subequations}
It follows, then, from Eq.~\eqref{eq: rhoW} that the Wigner transform of $\rho$, up to order $\hbar^2$, is obtained by substituting Eqs.~\eqref{eq: G2} and~\eqref{eq: G3} into
\be \label{eq: rhoW2}
\rho_{\rm W} = \frac{1}{O_{\rm W}} + \frac{\CG_2}{O_{\rm W}^3} -  \frac{\CG_3}{O_{\rm W}^4} + \dots
\ee
We note that the same result can be obtained by using the properties of the $\star$-product only. Indeed, the identity 
\be
\rho_{\rm W} \star O_{\rm W} = O_{\rm W} \star \rho_{\rm W} = 1 \, , 
\ee
together with the definition in Eq.~\eqref{eq: moyal}, implies that
\be \label{eq: cosL}
\rho_{\rm W} \cos\left[ \frac{\hbar}{2} \overleftrightarrow{\Lambda}\right] O_{\rm W} = 1  \, .
\ee
Expanding Eq.~\eqref{eq: cosL} in powers of $\hbar^2$, we have
\be \label{eq: cosL2}
\rho_{\rm W} O_{\rm W} - \frac{\hbar^2}{8} \rho_{\rm W} \left( \overleftrightarrow{\Lambda}\right)^2 O_{\rm W} + \CO(\hbar^4) = 1  \, ,
\ee
where
\be \label{eq: lambda2}
\rho_{\rm W} \left( \overleftrightarrow{\Lambda}\right)^2 O_{\rm W} =  \frac{\d^2 \rho_{\rm W}}{\d x^2} \frac{\d^2 O_{\rm W}}{\d y^2} + \frac{\d^2 \rho_{\rm W}}{\d y^2} \frac{\d^2 O_{\rm W}}{\d x^2} - 2 \frac{\d^2 \rho_{\rm W}}{\d x \d y} \frac{\d^2 O_{\rm W}}{\d x \d y} \, ,
\ee
and, solving order by order in $\hbar$, we find the Wigner transform of $\rho$ at NLO in $\hbar \rightarrow 0$.
Note that the formalism described here can be used to systematically extract the expansion up to any order by extending all intermediate computations beyond order $\hbar^2$.

\sectiono{Quantum dilogarithms and other useful formulae} \label{app: Faddeev}
We call quantum dilogarithm the function of two variables defined by the series~\cite{Faddeev2, Kirillov}
\be \label{eq: dilog}
(x q^{\alpha}; \, q)_{\infty} = \prod_{n=0}^{\infty} (1- x q^{\alpha+n}) \, , \quad \alpha \in \IR \, ,
\ee
which is analytic in $x,q \in \IC$ with $|x|, |q| <1$, and it has asymptotic expansions around $q$ a root of unity.
Furthermore, we denote by
\be \label{eq: qFactor}
(x; \, q)_m = \prod_{n=0}^{m-1} (1- x q^n) = \frac{(x; \, q)_{\infty}}{(x q^m; \, q)_{\infty}} \, , \quad m \in \IZ \, ,
\ee
with $(x; \, q)_0=1$, the $q$-shifted factorials, also known as $q$-Pochhammer symbols, and by
\be \label{eq: qHypergeo}
{}_{r+1}\phi_s\left(
\begin{matrix}
a_0 , & a_1 , & \dots , & a_r\\
b_1 , & b_2 , & \dots , & b_s
\end{matrix}
; \, q , \, x \right) = \sum_{n=0}^{\infty} \frac{(a_0; \, q)_n (a_1; \, q)_n \dots (a_r; \, q)_n}{(q; \, q)_n (b_1; \, q)_n \dots (b_s; \, q)_n} \left((-1)^n q^{\binom{n}{2}} \right)^{s-r} x^n \, ,
\ee
where $a_i, b_j, x \in \IC$, $r,s \in \IN$, the (unilateral) $q$-hypergeometric series, also called (unilateral) basic hypergeometric series.
Faddeev's quantum dilogarithm $\Phi_{\mb}(x)$ is defined in the strip $| \Im (z) | < | \Im (c_{\mb}) |$ as~\cite{Faddeev1, Faddeev2}
\be \label{eq: intPhib}
\Phi_{\mb}(x) = \exp \left( \int_{\IR + \ri \epsilon} \frac{\re^{-2 \ri x z}}{4 \sinh(z \mb ) \sinh(z \mb^{-1})} \frac{\rd z}{z} \right) \, ,
\ee
which can be analytically continued to all values of $\mb$ such that $\mb^2 \notin \IR_{\le 0}$. When $\Im(\mb^2) >0$, the formula in Eq.~\eqref{eq: intPhib} is equivalent to
\be \label{eq: seriesPhib}
\Phi_{\mb}(x) = \frac{( \re^{2 \pi \mb (x + c_{\mb})}; \, q)_{\infty}}{(  \re^{2 \pi \mb^{-1} (x - c_{\mb})}; \, \tilde{q})_{\infty}} = \prod_{n=0}^{\infty} \frac{1-  \re^{2 \pi \mb (x + c_{\mb})} q^n}{1- \re^{2 \pi \mb^{-1} (x - c_{\mb})} \tilde{q}^n}\, ,
\ee
where $q = \re^{2 \pi \text{i} \mb^2}$, $\tilde{q} = \re^{- 2 \pi \text{i} \mb^{-2}}$, and $c_{\mb} = \ri (\mb + \mb^{-1})/2$. 
It follows that $\Phi_{\mb}(x)$ is a meromorphic function with poles at the points $ x = c_{\mb} + \ri m \mb + \ri n \mb^{-1}$ and zeros at the points $x = -c_{\mb} - \ri m \mb - \ri n \mb^{-1}$, for $m,n \in \IN$. 
It satisfies the inversion formula
\be
\Phi_{\mb}(x) \Phi_{\mb}(-x) = \re^{\pi \ri x^2} \Phi_{\mb}(0)^2 \, , \quad \Phi_{\mb}(0) = \left( \frac{q}{\tilde{q}} \right)^{1/48} = \re^{\pi \ri (\mb^2 + \mb^{-2})/24} \, ,
\ee
and the complex conjugation formula
\be
\Phi_{\mb}(x)^* = \frac{1}{\Phi_{\mb^*}(x^*) } \, ,
\ee
and it is a quasi-periodic function. More precisely, we have that
\begin{subequations}
\be
\frac{\Phi_{\mb}(x+ c_{\mb} + \ri \mb)}{\Phi_{\mb}(x + c_{\mb})} = \frac{1}{1-q \re^{2 \pi \mb x}} \, , 
\ee
\be
\frac{\Phi_{\mb}(x+ c_{\mb} + \ri \mb^{-1})}{\Phi_{\mb}(x + c_{\mb})} = \frac{1}{1-\tilde{q}^{-1} \re^{2 \pi \mb^{-1} x}} \, .
\ee
\end{subequations}
In the limit $\mb \rightarrow 0$, under the assumption that $\Im(\mb^2) >0$, Faddeev's quantum dilogarithm has the asymptotic expansion~\cite{AK}
\be \label{eq: logPhib}
\log \Phi_{\mb}\left( \frac{x}{2 \pi \mb} \right) = \sum_{k=0}^{\infty} (2 \pi \ri \mb^2)^{2k-1} \frac{B_{2k}(1/2)}{(2k)!} \text{Li}_{2-2k}(-\re^x) \, ,
\ee
where $\text{Li}_n(z)$ is the polylogarithm of order $n$, and $B_n(z)$ is the $n$-th Bernoulli polynomial. Similarly, in the limit $\mb \rightarrow 0$, under the assumption that $\Im(\mb^2) >0$, the following special cases of the quantum dilogarithm in Eq.~\eqref{eq: dilog} have the asymptotic expansions~\cite{Katsurada}
\begin{subequations}
\begin{align}
\log (x ; \, q)_{\infty} = & \frac{1}{2} \log(1-x) + \sum_{k=0}^{\infty} (2 \pi \ri \mb^2)^{2k-1} \frac{B_{2k}}{(2k)!} \text{Li}_{2-2k}(x) \, , \label{eq: logPhiNC} \\
\log (q^{\alpha} ; \, q)_{\infty} = & - \frac{\pi \ri}{12 \mb^2} - B_1(\alpha) \log(- 2 \pi \ri \mb^2) - \log \frac{\Gamma(\alpha)}{\sqrt{2 \pi}}  \label{eq: logPhiK} \\
& - B_2(\alpha) \frac{\pi \ri \mb^2}{2} - \sum_{k=2}^{\infty} (2 \pi \ri \mb^2)^k \frac{B_k B_{k+1}(\alpha)}{k (k+1)!} \, , \quad \alpha > 0 \, , \nonumber
\end{align}
\end{subequations}
where $\Gamma(\alpha)$ is the gamma function.

\sectiono{Hadamard's multiplication theorem} \label{app: Hadamard}
We briefly recall the content of Hadamard's multiplication theorem~\cite{Hadamard, Hadamard2}, following the introduction by~\cite{Titchmarsh}.
\begin{theorem} \label{theo: HadamardT}
Consider the two formal power series
\be
f(z) = \sum_{n=0}^{\infty} a_n z^n \, , \quad g(z) = \sum_{n=0}^{\infty} b_n z^n \, ,
\ee
and suppose that $f(z), g(z)$ are convergent for $|z| < R, R'$, respectively, where $R, R' \in \IR_{+}$, and that their singularities in the complex $z$-plane are known. 
Let us denote by $\{\alpha_i\}$ the set of singularities of $f(z)$ and by $\{\beta_j\}$ the set of singularities of $g(z)$.
We introduce the formal power series
\be \label{eq: prodHa}
F(z) = (f \diamond g)(z) = \sum_{n=0}^{\infty} a_n b_n z^n \, ,
\ee
also known as the Hadamard product of the given series $f(z)$ and $g(z)$, which we denote with the symbol $\diamond$. 
Then, $F(z)$ has a finite radius of convergence $r > R R'$, and its singularities belong to the set $\{\alpha_i \beta_j\}$ of products of the singular points of $f(z)$ and $g(z)$. 
Furthermore, $F(z)$ admits the integral representation
\be \label{eq: intHa}
F(z) = \frac{1}{2 \pi \ri} \int_{\gamma} f(s) g(z/s) \frac{\rd s}{s} \, ,
\ee
where $\gamma$ is a closed contour encircling the origin $s=0$ on which 
\be \label{eq: condHa}
|s| < R \, , \quad \left| \frac{z}{s} \right| < R' \, .
\ee
\end{theorem}

Let us conclude by citing two results of~\cite{ZetaSeries} which we use in the explicit resummation of the Hadamard factors considered in Section~\ref{sec: localP2}. Namely, 
\begin{subequations}
\begin{align}
\sum_{k=1}^{\infty} \frac{\zeta(2k, a)}{2k} x^{2k} &=\frac{1}{2} \log \left( \Gamma (a-x)\Gamma (a+x) \right) -  \log \Gamma (a) \, , \quad |x| < |a| \, , \label{eq: zetaseriesEven} \\
\sum_{k=1}^{\infty} \frac{\zeta(2k+1, a)}{2k+1} x^{2k+1} &= \frac{1}{2} \log \left( \frac{\Gamma(a-x)}{\Gamma(a+x)} \right) + x \Psi(a) \, , \quad |x| < |a| \, , \label{eq: zetaseriesOdd}
\end{align}
\end{subequations}
where $\zeta(z, a)$ denotes the Hurwitz zeta function, $\Psi(a)$ denotes the digamma function, $x \in \IC$, and $a \in \IC \backslash \IZ_{\le 0}$.

\sectiono{Elements of alien calculus} \label{app: alien}
We briefly recall here the basic notions of alien calculus that are used in this paper. We refer to~\cite{ABS, MS, Dorigoni}. 
Let $\phi(z)$ be a resurgent asymptotic series, and let $\hat{\phi}(\zeta)$ be its Borel transform. We denote by $\zeta_{\omega}$, $\omega \in \Omega_{\theta}$, the singularities of $\hat{\phi}(\zeta)$ that lie on the same Stokes line at an angle 
\be \label{eq: lineA}
\theta = \arg(\zeta_{\omega})
\ee
in the complex $\zeta$-plane. For simplicity, we will now number the singularities along the given Stokes ray according to their increasing distance from the origin. Let us fix a value $\omega = r \in \IN_{\ne 0}$.
When analytically continuing $\hat{\phi}(\zeta)$ from the origin to the singularity $\zeta_r$ along the direction in Eq.~\eqref{eq: lineA}, each singularity $\zeta_i$, $i=1, \dots, r-1$, must be encircled by either passing slightly above or slightly below it. This creates ambiguity in the prescription. We label by $\epsilon_i = \pm 1$ the two choices, and we introduce the notation
\be
\hat{\phi}_{\zeta_1, \dots, \zeta_r}^{\epsilon_1, \dots , \epsilon_{r-1}} (\zeta)
\ee
to indicate that the analytic continuation is performed in such a way that the singularity $\zeta_i$ is avoided above or below according to the value of $\epsilon_i$ for $i=1, \dots, r-1$. Suppose that the local expansion of the Borel transform at $\zeta = \zeta_r$ has the form
\be
\hat{\phi}_{\zeta_1, \dots, \zeta_r}^{\epsilon_1, \dots , \epsilon_{r-1}} (\zeta) = -\frac{1}{2 \pi \ri \xi } c_{\zeta_1, \dots, \zeta_r}^{\epsilon_1, \dots , \epsilon_{r-1}} - \frac{\log(\xi)}{2 \pi \ri} \hat{\phi}_{r ; \, \zeta_1, \dots, \zeta_r}^{\epsilon_1, \dots , \epsilon_{r-1}} (\xi) + \dots \, ,
\ee
where $\xi = \zeta- \zeta_r$, the dots denote regular terms in $\xi$, $c_{\zeta_1, \dots, \zeta_r}^{\epsilon_1, \dots , \epsilon_{r-1}}$ is a complex number, and $\hat{\phi}_{r ; \, \zeta_1, \dots, \zeta_r}^{\epsilon_1, \dots , \epsilon_{r-1}} (\xi)$ is the germ of an analytic function at $\xi = 0$. The alien derivative at the singularity $\zeta_r$, which has been introduced in Eq.~\eqref{eq: alienStokes}, acts on the formal power series $\phi(z)$ as
\be \label{eq: alien_formula}
\Delta_{\zeta_r} \phi(z) = \sum_{\epsilon_1, \dots , \epsilon_{r-1}} \frac{p(\epsilon)! q(\epsilon)!}{r!} \left( c_{\zeta_1, \dots, \zeta_r}^{\epsilon_1, \dots , \epsilon_{r-1}} + \CB^{-1} \hat{\phi}_{r ; \, \zeta_1, \dots, \zeta_r}^{\epsilon_1, \dots , \epsilon_{r-1}} (z) \right) \, ,
\ee
where $\CB^{-1}$ denotes the inverse Borel transform, and $p(\epsilon),q(\epsilon)$ are the number of times that $\pm 1$ occur in the set $\{ \epsilon_1, \dots , \epsilon_{r-1}\}$, respectively. Furthermore, $\Delta_{\zeta} \phi(z) = 0$ if $\zeta \in \IC$ is not a singular point in the Borel plane of $\phi(z)$.
\begin{example}
Let us consider the simple example of 
\be
\hat{\phi}_{\zeta_1, \dots, \zeta_r}^{\epsilon_1, \dots , \epsilon_{r-1}} (\zeta) = - \frac{\log(\xi)}{2 \pi \ri} S_r + \dots \, ,
\ee
where $S_r \in \IC$ is a constant which depends only on the choice of the singularity $\zeta_r$. Since the inverse Borel transform acts trivially on numbers, we have that
\be \label{eq: exampleA}
\Delta_{\zeta_r} \phi(z) = S_r \sum_{\epsilon_1, \dots , \epsilon_{r-1}} \frac{p(\epsilon)! q(\epsilon)!}{r!} =  S_r \sum_{p=0}^{r-1} \frac{p! (r-1-p)!}{r!} \binom{r-1}{p} = S_r \, .
\ee
We observe that, in the case of 
\be
\hat{\phi}_{\zeta_1, \dots, \zeta_r}^{\epsilon_1, \dots , \epsilon_{r-1}} (\zeta) = -\frac{1}{2 \pi \ri \xi } S_r + \dots \, ,
\ee
where $S_r \in \IC$ is again a number which depends only on the choice of the singularity $\zeta_r$, the alien derivative at $\zeta_r$ acts on $\phi(z)$ according to the same formula in Eq.~\eqref{eq: exampleA}.
\end{example}
We conclude by recalling that the alien derivative $\Delta_{\zeta}$, $\zeta \in \IC$, is indeed a derivation in the algebra of resurgent functions. In particular, it satisfies the expected Leibniz rule when acting on a product, that is, 
\be \label{eq: alienD}
\Delta_{\zeta} \left( \phi_1(z) \phi_2(z) \right) = \left(\Delta_{\zeta}\phi_1(z)  \right) \phi_2(z) + \phi_1(z) \left(\Delta_{\zeta} \phi_2(z) \right) \, ,
\ee
where $\phi_1(z), \phi_2(z)$ are two given resurgent formal power series. As a consequence, the alien derivative also acts naturally on exponentials. Namely,
\be \label{eq: alienExp}
\Delta_{\zeta} \re^{\phi(z)} = \sum_{k=0}^{\infty} \frac{1}{k!} \Delta_{\zeta} \phi^k(z) = \sum_{k=1}^{\infty} \frac{1}{(k-1)!} \phi^{k-1}(z) \Delta_{\zeta} \phi(z) =  \re^{\phi(z)} \Delta_{\zeta} \phi(z) \, .
\ee

\bibliographystyle{JHEP}

\linespread{0.6}
\bibliography{biblio}

\end{document}